\documentclass[11pt]{article}
\pdfoutput=1 

\usepackage{jheppub}
\usepackage{graphicx}
\usepackage{amsmath,amsfonts,amssymb,amsthm}
\usepackage{braket}
\usepackage{subfigure} 
\usepackage{tikz}
\usetikzlibrary{decorations.pathmorphing}

\definecolor{Blue}{HTML}{2171b5}
\usepackage{hyperref}

\hypersetup{colorlinks=true,urlcolor=[rgb]{0,0,0.5},citecolor=[rgb]{0.5,0,0},linkcolor=[rgb]{0,0,0.4}}

\usepackage{comment}

\tikzset{meter/.append style={draw, inner sep=10, rectangle, font=\vphantom{A}, minimum width=30, line width=.8,
 path picture={\draw[black] ([shift={(.1,.3)}]path picture bounding box.south west) to[bend left=50] ([shift={(-.1,.3)}]path picture bounding box.south east);\draw[black,-latex] ([shift={(0,.1)}]path picture bounding box.south) -- ([shift={(.3,-.1)}]path picture bounding box.north);}}}
\usetikzlibrary{decorations.pathmorphing}

\newtheorem{definition}{Definition}
\newtheorem{corollary}{Corollary}
\newtheorem{theorem}{Theorem}
\newtheorem{lemma}{Lemma}
\newtheorem{proposition}{Proposition}

\newtheorem{example}{Example}
\newtheorem{conjecture}{Conjecture}
\newtheorem{fact}{Fact}

\usepackage{tabu}

\def\Tr{{\rm Tr}}

\def\op{{\cal O}}
\def\C{\mathbb{C}}

\def\EE{\mathbb{E}}

\def\vev#1{\langle{#1}\rangle}
\def\ketbra#1{ |{#1}\rangle\!\langle{#1}| }
\def\iden{\mathbb{I}}
\def\1den{\hbox{$1\hskip -1.2pt\vrule depth 0pt height 1.53ex width 0.7pt \vrule depth 0pt height 0.3pt width 0.12em$}}
\def\with{\quad {\rm with} \quad}

\def\for{\quad {\rm for} \quad}
\def\and{\quad {\rm and} \quad}

\def\ni{\noindent}
\def\nn{\nonumber\\}

\def\ra{\rightarrow}
\def\ie{{\rm i.e.\ }}

\def\CC{{\cal C}}

\def\CE{{\cal E}}

\def\CJ{{\cal J}}

\def\CT{{\cal T}}

\def\Wg{{\cal W}\! g}
\def\ktitle{\ensuremath{k}}

\title{Models of quantum complexity growth}

\author[a,b,c]{Fernando G.S.L. Brand\~ao,}
\author[a]{Wissam Chemissany,}
\author[d,a]{Nicholas Hunter-Jones,*\,}
\author[a,b]{\\ Richard Kueng,*\,}
\author[a,b,c]{John Preskill}

\affiliation[a]{Institute for Quantum Information and Matter,\\
California Institute of Technology, Pasadena, CA 91125}
\affiliation[b]{Department of Computing and Mathematical Sciences,\\ California Institute of Technology, Pasadena, CA 91125}
\affiliation[c]{Walter Burke Institute for Theoretical Physics,\\
California Institute of Technology, Pasadena, CA 91125}
\affiliation[d]{Perimeter Institute for Theoretical Physics, Waterloo, ON N2L 2Y5}
\affiliation[]{\vspace*{4pt}
{\bf \hspace*{-4pt}*Corresponding authors:
\href{mailto:nickrhj@pitp.ca}{nickrhj@pitp.ca} and
\href{mailto:rkueng@caltech.edu}{rkueng@caltech.edu}
}}

\abstract{
The concept of quantum complexity has far-reaching implications spanning theoretical computer science, quantum many-body physics, and high energy physics.
The quantum complexity of a unitary transformation or quantum state is defined as the size of the shortest quantum computation that executes the unitary or prepares the state. It is reasonable to expect that the complexity of a quantum state governed by a chaotic many-body Hamiltonian grows linearly with time for a time that is exponential in the system size; however, because it is hard to rule out a short-cut that improves the efficiency of a computation, it is notoriously difficult to derive lower bounds on quantum complexity for particular unitaries or states without making additional assumptions. To go further, one may study more generic models of complexity growth.
We provide a rigorous connection between complexity growth and unitary $k$-designs, ensembles which capture the randomness of the unitary group.
This connection allows us to leverage existing results about design growth to draw conclusions about the growth of complexity.
We prove that local random quantum circuits generate unitary transformations whose complexity grows linearly for a long time, mirroring the behavior one expects in chaotic quantum systems and verifying conjectures by Brown and Susskind.
Moreover, our results apply under a strong definition of quantum complexity based on optimal distinguishing measurements.
}

\begin{document} 
\maketitle
\flushbottom

\section{Motivation and overview}

The \textit{complexity} of a computation is a measure of the resources needed to perform the computation. In a classical model of computation, the complexity of a Boolean function may be defined as the minimal number of elementary steps needed to evaluate the function. The precise number of steps needed depends on how the model is chosen, but this notion of complexity provides a useful way to quantify the hardness of a computational problem because how the number of steps scales with the size of the input to the problem has only weak dependence on the choice of model. By broad consensus, a computational task is considered to be feasible if its complexity grows no faster than a power of the input size, and intractable otherwise; using this criterion, all classical models of computation agree about which problems are  (classically) ``easy'' and which ones are ``hard.'' 

Likewise, we may separate computational tasks into those that are easy or hard for a quantum computer. The circuit model of quantum computation provides a convenient way to quantify quantum complexity --- namely, the quantum complexity of a Boolean function is the minimal size of a quantum circuit which computes the function and outputs the right answer with high success probability. Here by the size of the circuit we mean the number of quantum gates in the circuit. These gates are chosen from a universal set of gates, where each gate in the set is a unitary transformation acting on a constant number of qubits or qudits. Though there are many ways to choose the universal gate set, any set of universal gates can simulate another accurately and efficiently, so that circuit size provides a useful model-independent measure of complexity. From a physicist's perspective, a quantum computation is a process governed by a local time-dependent Hamiltonian, and an intractable computation is a process that requires a time which grows superpolynomially with the system size. Such intractable processes are not expected to be observed in Nature. 

Furthermore, in quantum physics, in contrast to classical digital computation, there is a meaningful notion of complexity not only for processes, but also for quantum states. Starting from a state in which all the bits are set to 0, any string of $n$ classical bits can be prepared by flipping at most $n$ bits. But the time needed to prepare a pure $n$-qubit quantum state, starting from a product state, even if we are permitted to use any time-dependent Hamiltonian which is a sum of terms with constant weight and bounded norm, can be exponential in $n$. In fact, because the volume of the $n$-qubit Hilbert space is \textit{doubly exponential} in $n$, while the number of quantum circuits with $T$ gates is merely exponential in $T$, \textit{most} $n$-qubit pure quantum states have exponentially large complexity. That is, for a typical pure state in the $n$-qubit Hilbert space, the time needed to prepare the state with some small constant error $\delta$, starting from a product state, grows exponentially with $n$. Thus, nearly all quantum states of any macroscopic system will forever be far beyond the grasp of the quantum engineers \cite{poulin2011quantum}. 

While the complexity of quantum \textit{circuits} has long been a foundational concept in quantum information theory \cite{bernstein_complexity_1997}, appreciation that quantum \textit{state} complexity is an important concept has blossomed relatively recently. For example, the complexity of ground-state wave functions may be used to classify topological phases of matter at zero temperature \cite{Chen2010}. Furthermore, a chaotic quantum Hamiltonian $H$ can be usefully characterized by saying that evolution governed by $H$ over a long time period generates highly complex states. A particularly intriguing proposal is that, in the context of the AdS/CFT correspondence, the complexity of a quantum state of the boundary theory corresponds to the volume in the bulk geometry which is hidden behind the event horizon of a black hole \cite{SusskindCCBH14,SScomp14,CABH15,brown_second_2018}. 

When we say a quantum state is highly complex, we mean there is no easy way to prepare the state, but how can we be sure? Perhaps we were not clever enough to think of an ingenious short-cut that prepares the state efficiently. It's not possible in practice to enumerate all the quantum circuits that approximate a specified state to find one of minimal size. For that reason, it is quite difficult to obtain a useful lower bound on the complexity of the quantum state prepared by a specified many-body Hamiltonian in a specified time. It is reasonable to expect that, for a chaotic Hamiltonian $H$ and an unentangled initial state, the complexity grows linearly in time for an exponentially long time, but we do not have the tools to prove it from first principles for any particular $H$.

One possible approach is to rely on highly plausible complexity theory assumptions to derive nontrivial conclusions about the complexity of states generated by particular circuits or Hamiltonians \cite{susskind2018black,aaronson2016complexity,bohdanowicz2017universal}. Another is to consider ensembles of circuits, and to derive lower bounds on complexity which hold with high probability when samples are selected from these ensembles. We follow the latter approach here, drawing inspiration from recent work by Susskind \cite{susskind2018black} and Brown and Susskind \cite{brown_second_2018}. These authors state a conjecture about the complexity growth of geometrically local random quantum circuits (see \autoref{fig:complexity-growth}): 

\begin{conjecture}[Brown, Susskind \cite{brown_second_2018}; Susskind \cite{susskind2018black}] \label{conj:lenny}
Most local random circuits of size $T$ have a complexity that scales linearly in $T$ for an exponentially long time.
\end{conjecture}

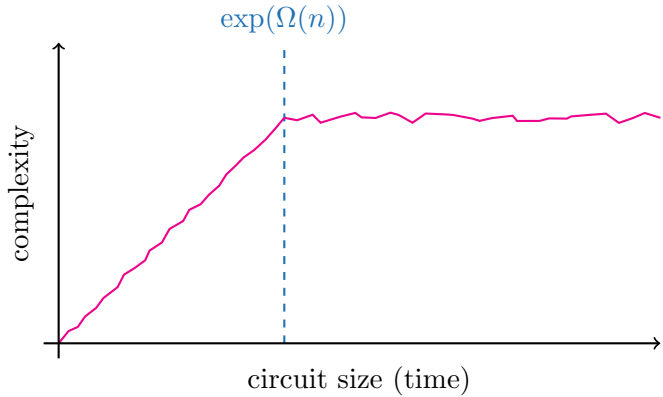
\begin{figure}
\centering
\begin{tikzpicture}
\draw[thick,dashed,color=Blue] (3,0) -- (3,4);
\node at (3,4.3) {\textcolor{Blue}{$
\exp (\Omega(n))$}};
\draw[magenta,thick, decorate, decoration={random steps,segment length=5pt,amplitude=1pt}] (0,0) -- (3,3);
\draw[magenta,thick, decorate, decoration={random steps,segment length=5pt,amplitude=2pt}] (3,3) -- (8,3);
\draw[thick,->] (-0.2,0) -- (8,0);
\node[rotate=90] at (-0.5,2) {complexity};
\draw[thick,->] (0,-0.2) -- (0,4);
\node at (4,-0.5) {circuit size (time)};
\end{tikzpicture}
\caption{\emph{Expected complexity growth in random circuits.} 
\autoref{conj:lenny} states that, for random quantum circuits acting on $n$ qubits, the circuit complexity grows linearly with circuit size (time) until it saturates at a value exponentially large in $n$.
Our work provides rigorous evidence supporting this picture for quantum systems with sufficiently large local dimension; see \autoref{cor:RQCcomp}.
}

\label{fig:complexity-growth}
\end{figure}

\noindent Our goal is to strengthen the evidence supporting this conjecture.

Brown and Susskind provided evidence for this scaling law by means of a simple counting argument; see also \cite{ChaosDesign}. For a fixed finite set of universal quantum gates, consider the ensemble of all circuits with size $T$. By definition, this ensemble accurately approximates (to within a specified error $\delta$) all unitary transformations with complexity $T$ or less. Furthermore, the number of circuits increases exponentially with $T$, and, because the unitary group has a very large volume, it seems reasonable to assume that ``collisions'' between circuits are rare unless $T$ is very large; that is, that the number of distinct unitary transformations realized by this ensemble (where ``distinct'' means more than distance $\delta$ apart) is comparable to the number of circuits. This means that the number of circuits with size $T'$ is too small to account for more than a small fraction of the unitary transformations realized by circuits of size $T$ if $T'$ is much smaller than $T$. In other words, most random circuits with size $T$ have complexity at least $T'$, where $T'$ is comparable to $T$.

This argument hinges on a crucial assumption, which sounds plausible but is hard to prove: \emph{collisions between circuits of subexponential size are rare}.
Collisions certainly occur for any circuit size $T$, and necessarily become common for circuits of exponential size, where $T$ is comparable to the Hilbert space dimension so that the exponential of $T$ is comparable to the Hilbert space volume. Thus an analytic treatment of complexity growth seems like a daunting combinatorial task. 

The work \cite{brandao_local_2016} provides some rigorous support for \autoref{conj:lenny}. There, the authors show that local random circuits can
``fool'' short measurement procedures. That is, a typical quantum state prepared by the local random circuit, acting on an initial product state, cannot be distinguished from a maximally mixed state by any procedure that is much simpler than running the circuit backwards and verifying that the initial product state is recovered. 
Although not stated in this fashion in \cite{brandao_local_2016}, their results imply that, with high probability, a local random circuit of size $T$ has complexity $\Omega(T^{1/11})$. While this argument rigorously proves a weakened version of Conjecture~\ref{conj:lenny}, there are still issues we wish to address:
\begin{enumerate}
\item[(i)] \emph{Restricted notion of complexity:} The authors implicitly define complexity as the capability of fooling short measurement protocols. While this operational notion of complexity is well motivated, the actual measurement procedures considered are quite restrictive. In particular, they do not take into account ancilla-assisted measurements --- a mainstay of modern quantum information.
\item[(ii)] \emph{Collisions are not treated explicitly:} 
The ensemble of local random circuits of size $T$ defines a probability distribution on the $n$-qubit unitaries. If we are only interested in specifying unitary transformations up to some specified error $\delta$, collisions occur, so that some unitaries are more likely than others. The arguments in \cite{brandao_local_2016} show that the unitaries sampled from this distribution typically have complexity $\Omega(T^{1/11})$, but do not rule out the possibility that the distribution is highly nonuniform. It is at least a logical possibility, compatible with the findings of \cite{brandao_local_2016}, that the ensemble contains only a small number of unitaries which have high complexity, all of which occur with relatively high probability. To conclude that the ensemble contains many high-complexity unitaries, we need to know more about the properties of the probability distribution governing the ensemble.

\item[(iii)] \emph{Polynomial relation between circuit size and complexity:} 
The relation between circuit size $T$ and expected minimal complexity $T^{1/11}$ is polynomial, not (yet) linear as required by Conjecture \ref{conj:lenny}.

\end{enumerate}

In this work we make progress toward a rigorous proof of Conjecture~\ref{conj:lenny} by developing a general framework which addresses some of the shortcomings of the previously known rigorous evidence in favor of the conjecture \cite{brandao_local_2016}. 
In particular, we define and use a \emph{strong} notion of complexity, which captures the difficulty of distinguishing a given circuit from the most useless possible quantum channel: the completely depolarizing channel $\mathcal{D}(\rho)=\frac{\Tr(\rho)}{d} \iden$ that maps any state to the maximally mixed state.

\begin{definition}[strong complexity: informal definition]\label{def:strong-complexity}
The complexity of a quantum circuit $U$ is the minimal circuit size required to implement an ancilla-assisted measurement that is capable of distinguishing $\rho \mapsto U \rho U^\dagger$ from the completely depolarizing channel $\rho \mapsto \frac{1}{d} \mathbb{I}$.
\end{definition}

\noindent We refer to Sec.~\ref{sec:defcomp} for a more detailed motivation and a precise statement of this definition. For now, we emphasize that this strong definition of complexity implies other (weaker) definitions, such as the minimal circuit size required to approximate $U$.

Our first main contribution is a rigorous connection between complexity growth and the notion of \emph{approximate unitary $k$-designs} \cite{Dankert09,gross_evenly_2007}. We use the notation $\left\{p_i,U_i \right\}$ for an ensemble of unitary transformations in which the unitary $U_i$ occurs with probability $p_i$.
A unitary $k$-design is an ensemble with strong pseudo-random properties; 
an approximate $k$-design accurately approximates the first $k$ moments of the Haar measure on the unitary group. 
Hence a $k$-design with large $k$ behaves essentially like a Haar-random ensemble of unitaries, while a small-$k$ design can be highly structured. For instance, the $n$-qubit Pauli group forms a 1-design, while the $n$-qubit Clifford group is a 3-design \cite{webb_clifford_2015,zhu_multiqubit_2017,Kueng15}.
The design order $k$ allows us to interpolate between these two very different regimes. Intuitively, we would expect that the complexity of a $k$-design grows with $k$. Our first technical contribution makes this intuition precise: a linear growth in design implies a linear growth in (strong) complexity.

\begin{theorem}[informal statement] \label{thm:main1_intro}
Let $\left\{ p_i, U_i \right\}$ be an approximate unitary $k$-design. 
Then, a randomly selected (according to the weights) element is very likely to have strong circuit complexity $\approx k$.
\end{theorem}

\noindent We refer to \autoref{thm:main_circuit} for a more detailed, quantitative statement. This result strengthens the assertions in \cite{brandao_local_2016} by allowing ancilla-assisted measurement procedures. To do so we prove novel bounds on Haar moments, see Sec.~\ref{sub:technical} for details.
Our second technical contribution 
shows that the $k$-design property alone severely limits the likelihood of collisions.

\begin{lemma} \label{lem:weights-intro}
Let $\left\{ p_i, U_i \right\}$ be an approximate $k$-design. Then, the associated weight distribution cannot be too spiky: $\max_i \,p_i \lesssim k! d^{-2k}$.
\end{lemma}

\noindent This result formalizes the intuitive idea that giving unusually high weight to some unitaries cannot be compatible with the $k$-design property, but we are not aware of any precise statements along these lines in the existing literature. 
Importantly, because \autoref{lem:weights-intro} establishes that the distribution is nearly flat, knowing that sampling from a $k$-design yields a high-complexity unitary with high probability (as stated in \autoref{thm:main1_intro}) allows us to infer that there must be many distinct high-complexity unitaries in the ensemble. Here our reasoning is based on an approximate version of Laplace's definition of probability: if events are assigned nearly equal probabilities, then the probability of property $X$ is approximately the number of events with property $X$ divided by the total number of events. Together, \autoref{thm:main1_intro} and \autoref{lem:weights-intro} imply the following corollary:

\begin{corollary} \label{cor:main_complexity_growth}
\emph{Any} approximate $k$-design contains exponentially many (in $k$) unitaries that have circuit complexity $\Omega(k)$.
\end{corollary}

\noindent
While \autoref{cor:main_complexity_growth} does not by itself strongly constrain how these high-complexity unitary transformations are distributed geometrically within the $n$-qudit unitary group, we are also able to prove a stronger result: \emph{An approximate $k$-design contains exponentially many (in $k$) high-complexity unitaries whose pairwise distance \textnormal{(i.e.\ the distance between any pair of unitaries)} is almost maximal in the diamond norm.} This stronger statement rules out the possibility that most of the high-complexity unitaries reside inside a few tightly packed clusters within $U(d)$.

Approximate unitary $k$-designs are a central concept in quantum information, where their
pseudo-random properties have found extensive application across subfields, e.g. state distinguishability \cite{Ambainis2007}, decoupling \cite{szehr2013}, state tomography \cite{scott_tight_2006,kueng_low_2017}, randomized benchmarking \cite{emerson05}, equilibration \cite{brandao_local_2016} (and references therein),
information scrambling \cite{HaydenPreskill,ChaosDesign}, and many more. As a result, several probabilistic constructions are known. Applying  \autoref{cor:main_complexity_growth} to any of these constructions establishes a rigorous model for quantum complexity growth. In particular,
\begin{itemize}
\item \emph{Local random quantum circuits with polynomial design growth:} Ref.~\cite{brandao_local_2016} proves that the set of all geometrically local circuits of size $T=O(n^2 k^{11})$ forms an approximate unitary $k$-design.\footnote{Note that here we discuss the size of the circuit, if we parallelize the application of gates, the depth of the circuit required to form an approximate design scales linearly in $n$.} \autoref{cor:main_complexity_growth} therefore implies that local circuits of size $T$ contain at least $\exp (\Omega(T^{1/11}))$ elements with strong complexity $\Omega(T^{1/11})$. 

\item \emph{Stochastic quantum Hamiltonians with polynomial design growth:} 
One can study the growth of complexity in continuous-time models of chaotic dynamics, rather than the discrete-time dynamics embodied by random circuits \cite{FastScrambling,onorati_mixing_2017,Nakata16}.
Stochastic Hamiltonian dynamics, in which a local Hamiltonian fluctuates as a function of time, has been shown to realize approximate $k$-designs \cite{onorati_mixing_2017} with a relationship between $k$ and the evolution time similar to what was established in \cite{brandao_local_2016} for local random circuits.
Further progress achieved in \cite{Nakata16} shows that, for a particular class of stochastic Hamiltonians, evolution time linear in $k$ suffices to generate approximate $k$-designs for $k=o(\sqrt{n})$. \autoref{cor:main_complexity_growth} therefore implies that with high probability the complexity grows linearly in time, at least for a while. 

\item \emph{Local random circuits with linear design growth:} Very recently, one of the authors substantially improved on the results of \cite{brandao_local_2016} using an exact mapping from random circuits to the statistical mechanics of a lattice model \cite{NHJ19}, showing that local circuits of size $T=O(n^2 k)$ form approximate $k$-designs in the limit of large local dimension (Hilbert space dimension $d=q^n$ with $q$ large). Combined with \autoref{cor:main_complexity_growth} this establishes a \emph{linear} relation between circuit size and complexity. Thus we can prove the following statement analogous to \autoref{conj:lenny}: 
\end{itemize}

\begin{corollary} \label{cor:RQCcomp}
The set of all local circuits of size $T$ contains at least $\exp (\Omega(T))$  elements with strong complexity $\Omega (T)$, provided that the local dimension is sufficiently large: $q \geq q_0 (T)$.
\end{corollary}

More precise statements of our main results, and a more detailed comparison to previous work, can be found in Sec.~\ref{sec:quantum-complexity}.

\section{Quantum complexity and unitary designs}
\label{sec:quantum-complexity}

\subsection{Operational definitions of complexity}
\label{sec:defcomp}

\subsubsection{State complexity}

We consider systems comprised of $n$ qudits with local dimension $q$: $d = q^n$.
Existing works on complexity typically start with identifying a class of states that are \emph{useful} starting states for quantum computations. In this work we will take a reverse approach and start with identifying a \emph{useless} state. The maximally mixed state
\begin{equation}
\rho_0 = \frac{\iden}{d}
\end{equation}
is unique in the sense that it is invariant under arbitrary unitary evolutions, including any quantum circuit.
Intuitively, useful starting states should be as far away from this useless state as possible. 
If we measure distance in trace norm, this intuition is true to some extent. Any pure state $|\psi \rangle \! \langle \psi|$ obeys 
\begin{equation}
\frac{1}{2}\left\| | \psi \rangle \! \langle \psi | - \rho_0 \right\|_1 = 1- \frac{1}{d}\,.
\end{equation}
But this is clearly too coarse for distinguishing the usefulness of different states.
In order to achieve such a task, we recall the operational interpretation of the trace distance. 
It corresponds the optimal bias achievable in distinguishing the state $|\psi \rangle \! \langle \psi|$ from $\rho_0$ using a single measurement \cite{Holevo1973,Helstrom1976}. We refer to Sec.~\ref{sub:distinguishing_states} for a more detailed exposition.
The optimal measurement achieving this bias is $M = | \psi \rangle \! \langle \psi| $ and \emph{does} depend on the state in question. Such a measurement may be challenging to implement for states that we would intuitively assign a high complexity to (such as random states) and very easy for states that we consider useful (such as computational basis states). 
We can interpolate between these extreme cases by limiting the resources available to implement distinguishing measurements. Let $\mathbb{H}_d$ denote the space of $d\times d$ Hermitian matrices. 
For fixed $r \in \mathbb{N}$, we consider the class of measurements $\mathsf{M}_r(d) \subset \mathbb{H}_d$ that can be implemented by combining (at most) $r$ 2-local gates from a fixed, universal gate set $\mathsf{G} \subset U(4)$. We refer to Sec.~\ref{sub:easy_measurement} for further details and justification.
The maximal bias achievable with such a restricted set of measurements is the solution to the following optimization problem:
\begin{align}
\beta_{\mathrm{qs}}^\sharp \left(r, | \psi \rangle \right) 
=\,\, &\textrm{maximize} \quad \left| \Tr \left( M \left( | \psi \rangle \! \langle \psi| -\rho_0\right)\right)\right| \\
&\textrm{subject to} \quad M \in \mathsf{M}_r (d)\,. \nonumber
\end{align}
We may decompose the true optimal measurement as $| \psi \rangle \! \langle \psi| = U |0 \rangle \! \langle 0| U^\dagger$ for some $U \in U(d)$. The unitary $U$ may be approximated to arbitrary precision by 2-local circuits chosen from a universal gate set \cite{dawson_kitaev_2005}. This ensures 
\begin{equation}
\beta^\sharp (r, | \psi \rangle) \longrightarrow \frac{1}{2} \| | \psi \rangle \! \langle \psi| - \rho_0 \|_1 = 1 - \frac{1}{d}\quad \textrm{as} \quad r \to \infty\,.
\end{equation}
For simple states, like computational basis states, this convergence happens rapidly, while generic states $| \psi \rangle$ require exponentially large circuit depths. This observation is the motivation for the following definition of complexity.

\begin{definition}[Strong state complexity] \label{def:state_complexity}
Fix $r \in \mathbb{N}$ and $\delta \in (0,1)$. We say that a pure state $| \psi \rangle $ has strong $\delta$-state complexity at most $r$ if
\begin{equation}
\beta_{\mathrm{qs}}^\sharp (r, | \psi \rangle) \geq 1-\frac{1}{d}-\delta\,, \quad \textrm{which we denote as} \quad \mathcal{C}_\delta \left( | \psi \rangle \right) \leq r\,.
\end{equation}
\end{definition}

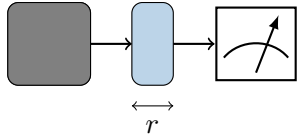
\begin{figure}
\begin{center}
\begin{tikzpicture}[baseline,scale=0.55]
\draw[rounded corners] (-2,0.5) rectangle (0,2.5);
\draw[rounded corners, fill = gray] (-2,0.5) rectangle (0,2.5);
\draw[rounded corners] (1,0.5) rectangle (2,2.5);
\draw[rounded corners,fill=Blue,opacity=0.3] (1,0.5) rectangle (2,2.5);
\draw[<->] (1,0) -- (2,0);
\node at (1.5,-0.5) {$r$};
\draw[->,thick] (0,1.5) -- (1,1.5);
\draw[->,thick] (2,1.5) -- (3,1.5);
\node[meter] at (4,1.5) {};
\end{tikzpicture}
\end{center}
\caption{\emph{Pictographic illustration of strong state complexity (\autoref{def:state_complexity}).} A black-box either outputs a (known) pure state $\rho = | \psi \rangle \! \langle \psi|$, or the maximally mixed state $\rho_0=\frac{1}{d} \mathbb{I}$. The task is to correctly guess which one it produced by applying a pre-processing circuit $V$ (blue) of limited size $r$ and performing a simple measurement (right).
We say that $| \psi \rangle$ has \emph{strong state complexity} less than $r$ if the probability of correctly distinguishing both possibilities is close to optimal.
} \label{fig:state_complexity}
\end{figure}

This definition has a ready operational interpretation that is illustrated in \autoref{fig:state_complexity}.
The following result relates it to more traditional definitions.

\begin{lemma} \label{lem:state_comp_stronger}
 Suppose that $| \psi \rangle \in \mathbb{C}^d$ obeys $\mathcal{C}_\delta (| \psi \rangle) \geq r+1$ for some $\delta \in (0,1)$ and $r \in \mathbb{N}$.
Then,
\begin{equation}
\min_{\textrm{size}(V) \leq r}
\frac{1}{2} \left\| | \psi \rangle \! \langle \psi| - V |0 \rangle \! \langle 0| V^\dagger \right\|_1 > \sqrt{\delta}\,,
\label{eq:simple_state_complexity}
\end{equation}
i.e.\ it is impossible to accurately produce $\ket{\psi}$ with fewer than $r$ elementary gates.
\end{lemma}

The converse is false in general. To see this, select a generic state $\ket{\widetilde\psi}$ on $(n-1)$ qudits and set $\ket\psi = \ket{0}\otimes \ket{\widetilde\psi} $. 
Then, the quantity in Eq.~\eqref{eq:simple_state_complexity} is dominated by the (traditional) complexity of $\ket{\widetilde\psi}$, which may be very high. Nonetheless, the simple distinguishing measurement $M = \ketbra{0} \otimes \iden$ (ignore everything but the first qudit) achieves
\begin{equation}
\Tr \left( M ( \ketbra{\psi} - \rho_0 ) \right) = \Tr \left( \ketbra{0} \big(\ketbra{0} - \tfrac{1}{q} \iden\big) \right)
= 1 - \frac{1}{q}\,,
\end{equation}
which is high, especially for large local dimension $q$. This example highlights that
\autoref{def:state_complexity} is indeed more stringent than traditional definitions of state complexity.


\begin{proof}[Proof of \autoref{lem:state_comp_stronger}]

By contraposition. 
Let $\mathsf{G}_r \subset U(d)$ denote the class of unitary circuits that are comprised of at most $r$ 2-local gates chosen from a universal gate set $\mathsf{G}$.
Suppose there exists a size-$r$ circuit $V \in \mathsf{G}_r$ such that
$
\frac{1}{2} \left\| | \psi \rangle \! \langle \psi| - V |0 \rangle \! \langle 0| V^\dagger \right\|_1 \leq \sqrt{\delta}.
$
The state difference in question has rank two which allows for explicitly computing the trace distance:
$
\frac{1}{2} \| | \psi \rangle \! \langle \psi| - V |0 \rangle \! \langle 0| V^\dagger \|_1
= \sqrt{1- | \langle 0| V^\dagger| \psi \rangle|^2}$.
The assumption is therefore equivalent to $| \langle 0| V^\dagger | \psi \rangle|^2 \geq 1-\delta$  and we conclude
\begin{align}
\beta^\sharp (r, | \psi \rangle) \geq \Tr \left( V| 0 \rangle \! \langle 0| V^\dagger \left(| \psi \rangle \! \langle \psi| - \rho_0 \right) \right) = | \langle 0| V^\dagger| \psi \rangle|^2 - \frac{1}{d}  \geq 1- \frac{1}{d} - \delta\,,
\end{align}
because $V | 0 \rangle \! \langle 0| V^\dagger \in \mathsf{M}_r$. This in turn implies
$\mathcal{C}_\delta (| \psi \rangle) \leq r$
and the claim follows.
\end{proof}

\subsubsection{Unitary complexity}

We define the complexity of unitary channels $\mathcal{U}(\rho) = U \rho U^\dagger$ in a fashion similar to state complexity. 
We start with identifying the completely depolarizing channel as the most \emph{useless} channel conceivable:
\begin{equation}
\mathcal{D}(\rho) = \rho_0= \frac{\mathbb{I}}{d} \quad \textrm{for all states} ~ \rho\,. 
\end{equation}
The \emph{diamond distance} between $\mathcal{D}$ and any unitary channel is close to maximal:
\begin{equation}
\frac{1}{2} \left\| \mathcal{U} - \mathcal{D} \right\|_\diamond = 1 - \frac{1}{d^2}\,.
\end{equation}
As detailed in Sec.~\ref{sub:distinguishing_channels}, the diamond distance also has an appealing operational definition \cite{Watrous2018}. 
It corresponds to the maximal bias achievable for the task of distinguishing the channel $\mathcal{U}$ from $\mathcal{D}$ with a single channel use.
The optimal strategy may involve a quantum memory. Choose a state in the doubled Hilbert space $| \phi \rangle \! \langle \phi|$, with $\ket\phi \in \C^d\otimes \C^d$ and input one half into the unknown channel, while the other half remains unchanged in the quantum memory. Subsequently, perform a two-outcome measurement on the output state to distinguish both channels. 

An optimal strategy for distinguishing $\mathcal{U}$ from $\mathcal{D}$ corresponds to choosing a maximally entangled (Bell) state $| \Omega \rangle \in \mathbb{C}^d \otimes \mathbb{C}^d$ as input and measuring $M= (U \otimes \mathbb{I}) \ketbra{\Omega} (U^\dagger \otimes \mathbb{I})$. Equivalently, choose $(U^\dagger \otimes \mathbb{I}) |\Omega \rangle$ as input and measure $M = \ketbra{\Omega}$ on the output. 
Similar to the state complexity argument, the optimal input state, or the optimal outcome measurement (or both) depend on the unitary $U \in U(d)$ describing the channel $\mathcal{U}$. 
This may be challenging to implement, especially if $U$ corresponds to a complicated circuit. 
We restrict apparatus power by bounding the total circuit sizes that are allowed to implement such a measurement procedure.
Let $\mathsf{G}_{r'} \subset U(d^2)$ be the set of all unitary circuits on $2n$ qudits (register+memory) that are comprised of at most $r'$ elementary gates. Likewise, let $\mathsf{M}_{r''} \subset \mathbb{H}_d^{\otimes 2}$ denote the class of all two-outcome measurements on $2n$ qudits that require circuit size at most $r''$ to implement. 
The optimal bias achievable under such restrictions is
\begin{align}
\beta^\sharp_{\mathrm{qc}}(r,U) =\,\, &\textrm{maximize} \quad \big| \Tr \left( M \left(\left( \mathcal{U} \otimes \mathcal{I} \right) (| \phi \rangle \! \langle \phi|) - \left( \mathcal{D} \otimes \mathcal{I} \right)(| \phi \rangle \! \langle \phi|)\right) \right) \big| \\
&\textrm{subject to} \quad M \in \mathsf{M}_{r'}\,,~ | \phi \rangle = V |0 \rangle\,,~ V \in \mathsf{G}_{r''}\,,~ r = r'+r''\,, \nonumber
\end{align}
where the identity channel $\mathcal{I}: \mathbb{H}_d \to \mathbb{H}_d$ indicates that the memory is left unchanged. 
As $r$ increases, more complicated measurements and state preparations become possible. At some point this will include ever more precise approximations of the optimal measurement \cite{dawson_kitaev_2005}:
\begin{equation}
\beta^\sharp (r, U) \longrightarrow \frac{1}{2} \| \mathcal{U} - \mathcal{D} \|_\diamond = 1 - \frac{1}{d^2} \quad \textrm{as} \quad r \to + \infty\,.
\end{equation}
Similar to the state case, the rate of convergence does depend on the complexity of the unknown unitary $U$. This is the basis for our operational definition of unitary complexity.

\begin{definition}[Strong unitary complexity] \label{def:unitary_complexity}
Fix $r \in \mathbb{N}$ and $\delta \in (0,1)$. We say that a unitary $U \in U(d)$ has strong $\delta$-unitary complexity at most $r$ if
\begin{equation}
\beta^\sharp_{\mathrm{qc}} (r,U) \geq 1- \frac{1}{d^2} - \delta\,,
\quad \textrm{which we denote as} \quad \mathcal{C}_\delta (U) \leq r\,.
\end{equation}
\end{definition}

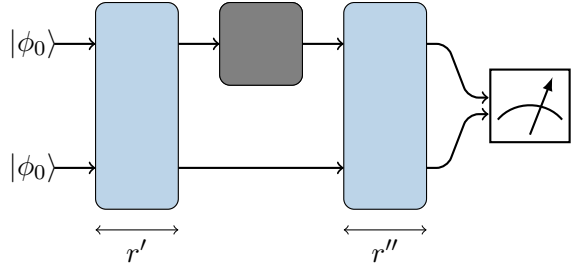
\begin{figure}
\begin{center}
\begin{tikzpicture}[baseline,scale=0.55]
\draw[rounded corners,fill=Blue,opacity=0.3] (-5,-2.5) rectangle (-3,2.5);
\draw[rounded corners] (-5,-2.5) rectangle (-3,2.5);
\draw[<->] (-5,-3) -- (-3,-3);
\node at (-4,-3.5) {$r'$};
\draw[rounded corners] (-2,0.5) rectangle (0,2.5);
\draw[rounded corners, fill = gray] (-2,0.5) rectangle (0,2.5);
\draw[rounded corners] (1,-2.5) rectangle (3,2.5);
\draw[rounded corners,fill=Blue,opacity=0.3] (1,-2.5) rectangle (3,2.5);
\draw[<->] (1,-3) -- (3,-3);
\node at (2,-3.5) {$r''$};
\draw[->, thick, rounded corners] (-6,1.5) -- (-5,1.5);
\node at (-6.5,1.5) {$| \phi_0 \rangle$};
\node at (-6.5,-1.5) {$| \phi_0 \rangle$};
\draw[->,thick, rounded corners] (-6,-1.5) -- (-5,-1.5);
\draw[->,thick] (-3,1.5) -- (-2,1.5);
\draw[->,thick] (0,1.5) -- (1,1.5);
\draw[->,thick] (-3,-1.5) -- (1,-1.5);
\draw[->, thick, rounded corners] (3,1.5) -- (3.5,1.5) -- (4,0.2) -- (4.5,0.2);
\draw[->,thick, rounded corners] (3,-1.5) -- (3.5,-1.5) -- (4,-0.2) -- (4.5,-0.2);
\node[meter] at (5.5,0) {};
\end{tikzpicture}
\end{center}
\caption{\emph{Pictographic illustration of strong unitary complexity (\autoref{def:unitary_complexity}).} A black box (center) takes quantum states as inputs and applies either a unitary channel $\mathcal{U}(\rho) = U \rho U^\dagger$, or the depolarizing channel $\mathcal{D}(\rho) = \rho_0$.
The task is to correctly guess which evolution occurred. 
The rules of the game allow short pre- and post-processing circuits (blue) that may involve a quantum memory. The final guess must be based on a simple measurement (right).
We say that $U$ has complexity less than $r=r'+r''$ if the probability of correctly distinguishing both options is close to optimal.
}
\label{fig:unitary_complexity}
\end{figure}

The operational motivation for this definition is sketched in \autoref{fig:unitary_complexity}.
Strong unitary complexity 
(\autoref{def:unitary_complexity}) is more stringent than traditional definitions that use approximation errors in some norm.

\begin{lemma} \label{lem:unitary_complexity_stronger}
Suppose that $U \in U(d)$ obeys $\mathcal{C}_\delta(U) \geq r+1$ for some $\delta \in (0,1)$, $r \in \mathbb{N}$ and measurement procedures that include the Bell-measurement $\ketbra{\Omega}$.
 Then,
\begin{equation}
\min_{\textrm{size}(V) \leq r}
\frac{1}{2} \left\| \mathcal{U} - \mathcal{V} \right\|_\diamond > \sqrt{\delta}\,,
\end{equation}
i.e.\ it is impossible to accurately approximate $U$ by circuits comprised of fewer than $r$ elementary gates.
\end{lemma}

Again, the converse relation is false in general.
Including the Bell measurement $\ketbra{\Omega}$ simplifies the proof of this claim considerably. Also, relatively short circuits allow for transforming computational basis states into Bell states. For $q=2$, a depth-two circuit comprised of $n$ Hadamard gates and $n$ CNOTs suffices.

\begin{proof}[Proof of \autoref{lem:unitary_complexity_stronger}]
By contraposition. Assume there exists $V \in U(d)$ with $\mathrm{size}(V) \leq r$
 such that
$
\tfrac{1}{2} \left\| \mathcal{U} - \mathcal{V} \right\|_\diamond \leq \sqrt{\delta}$. Then,
\begin{align}
\sqrt{\delta} &\geq \frac{1}{2} \| \mathcal{U} - \mathcal{V} \|_\diamond \geq \frac{1}{2} \left\| (U \otimes \mathbb{I}) | \Omega \rangle \! \langle \Omega| (U^\dagger \otimes \mathbb{I}) - (V \otimes \mathbb{I}) | \Omega \rangle \! \langle \Omega |(V^\dagger \otimes \mathbb{I}) \right\|_1 \nn
&= \sqrt{1- | \langle \Omega| V^\dagger U \otimes \mathbb{I} | \Omega \rangle|^2}\,,
\end{align}
as the second expression involves a trace distance of two pure states which can be computed explicitly.
Next, note that $M = (V \otimes \mathbb{I}) | \Omega \rangle \! \langle \Omega| (V^\dagger \otimes \mathbb{I})$ is a legitimate distinguishing measurement, because $\mathrm{size}(V) \leq r$ and $M=\ketbra{\Omega}$ is permitted by assumption. 
Together with $\mathcal{D} \otimes \mathcal{I} (| \Omega \rangle \! \langle \Omega|) = \rho_0^{\otimes 2}$ we conclude
\begin{align}
\beta^\sharp_{\mathrm{qc}} (r, U) &\geq \Tr \left(  (V \otimes \mathbb{I}) | \Omega \rangle \! \langle \Omega| (V^\dagger \otimes \mathbb{I})
\left( (U \otimes \mathbb{I}) | \Omega \rangle \! \langle \Omega| (U^\dagger \otimes \mathbb{I}) - \rho_0^{\otimes 2} \right)\right) \nn
&= \big| \langle \Omega| V^\dagger U \otimes \mathbb{I} | \Omega \rangle \big|^2 -\langle \Omega | V^\dagger  \rho_0 V \otimes \rho_0| \Omega \rangle 
\geq  1- \delta^2 - \frac{1}{d^2}\,.
\end{align}
\end{proof}

\subsection{Approximate unitary designs}

The concept of \emph{unitary $k$-designs} \cite{Dankert09,gross_evenly_2007} provides an interpolation between two extreme cases: (i) small collections of highly structured unitaries that form the basic building blocks of quantum computing devices (e.g. local Pauli gates, or elements of the Clifford group). (ii) generic (Haar random) unitaries that lack any sort of structure and require circuits of exponential size to approximate. 

Roughly speaking, an ensemble $\mathcal{E}=\left\{p_i, U_i \right\}$ of unitaries is a unitary $k$-design if it exactly reproduces the first $k$ moments of the Haar measure over the unitary group.
More precisely, given the twirling channels $\mathcal{T}^{(k)}_{U}(X) = \int \mathrm{d}U U^{\otimes k} X (U^\dagger)^{\otimes k}$ and $\mathcal{T}_{\mathcal{E}}^{(k)}(X) = \sum_i p_i U_i^{\otimes k} X (U_i^\dagger)^{\otimes k}$, an ensemble $\mathcal{E}$ is a unitary design with order $k$ if
\begin{equation}
\mathcal{T}_{\mathcal{E}}^{(k)}(X) =  \mathcal{T}_{U}^{(k)} (X) \quad \textrm{for all $X$ in the $k$-fold tensor product}. \label{eq:k-design_intro}
\end{equation}
Although seemingly abstract, this notion captures important physical concepts. 1-designs are in one-to-one correspondence with unitary operator frames, while 2-designs sufficiently capture the notion of \emph{scrambling} \cite{HaydenPreskill,ChaosDesign}. 

Unitary $k$-designs are known to exist for any dimension $d$ and any order $k$. Nevertheless, explicit constructions are notoriously difficult to find. This challenge can be overcome by relaxing the notion of a $k$-design. Indeed, for most applications it is sufficient to ensure that Eq.~\eqref{eq:k-design_intro} is only approximately true, see \autoref{def:approximate_design} for a precise statement.
Several conventions for choosing an appropriate distance measure $\| \cdot \|$ have been put forth, but here we opt for the diamond distance $\| \cdot \|_\diamond$ which quantifies the distinguishability of two ensembles.

In contrast to exact $k$-designs, several explicit constructions for approximate $k$-designs have been established \cite{brandao_local_2016,onorati_mixing_2017,Nakata16,NHJ19,HM18,HL09}, including local random circuits and various Brownian circuits/stochastic quantum Hamiltonians. 
These constructions allow us to relate abstract insights about complexity growth in designs to concrete random circuit models.

\subsection{Complexity by design}

\subsubsection{State complexity growth}

\begin{theorem} \label{thm:main_state}
Consider the set of (pure) states in $d=q^n$ dimensions that results from applying all unitaries associated with an $\epsilon$-approximate $2k$-design to a fixed starting state $| \psi_0 \rangle$.
Then, this set contains at least
\begin{equation}
\binom{d+k-1}{k} \left( \frac{1}{1+\epsilon}- 2 d n^r |\mathsf{G}|^r \left( \frac{16k^2}{d(1-\delta)^2} \right)^{k} \right)
\end{equation}
\emph{distinct} states that obey $\mathcal{C}_\delta (| \psi \rangle) \geq r+1$ each.
Qualitatively, this number is of order $(d/k)^k$ as long as $r$ obeys
\begin{equation}
r \lesssim \frac{k (n - 2\log (k))}{\log (n)}\,.
\end{equation}
\end{theorem}

Because of collisions, the emphasis on distinct is justified; two or more distinct unitaries can lead to the same final state.

\subsubsection{Unitary complexity growth}

\begin{theorem} \label{thm:main_circuit}
A discrete approximate $2k$-design in $d=q^n$ dimension contains at least
\begin{equation*}
\frac{d^{2k}}{k!} \left( \frac{1}{1+\epsilon} - 3d^2 n^{2r} |\mathsf{G}|^r \left( \frac{1024 k^4}{d(1-\delta)^2}\right)^k \right)\,
\end{equation*}
distinct unitaries that obey $\mathcal{C}_\delta (U) \geq r+1$ each.
Qualitatively, this number is of order $(d^2/k)^k$ as long as $r$ obeys
\begin{equation}
r \lesssim \frac{k (n- 4 \log (k))}{\log (n)}\,.
\end{equation}
\end{theorem}

\subsection{Moment bounds} \label{sub:technical}

Both \autoref{thm:main_state} and \autoref{thm:main_circuit} follow from an initial probabilistic statement combined with
relatively straightforward counting arguments.
These probabilistic statements highlight that it is very unlikely to distinguish random $k$-design elements from their average
with a fixed measurement procedure. Markov's inequality --- $\mathrm{Pr}\left[S \geq \tau \right] = \mathrm{Pr}\left[S^k \geq \tau^k\right] \leq \mathbb{E}[S^k]/\tau^k$ for nonnegative random variables $S$ --- reduces this probabilistic assertion to a question about moment growth. The larger the moments we can control, the stronger this assertion becomes. Designs appropriately capture this feature: a $k$-design accurately approximates Haar-random moments up to order $k$.
This is why designs with growing $k$ become increasingly complex.

For state complexity, the associated Haar moment computation is relatively straightforward:
\begin{equation}
\mathbb{E}_{\ket{\psi}} \left[ \Big(\Tr( M \ketbra{\psi}) - \mathbb{E}_{\ket{\psi}} \big[ \Tr (M \ketbra{\psi}) \big] \Big)^{k} \right]  \leq \left(\frac{k^2}{d}\right)^{k/2}
\label{eq:state-moment-intro}
\end{equation}
for any fixed measurement $M$, see e.g.\ \autoref{prop:state_moments} below.

However such simple moments do not cover strong unitary complexity. Quantum channels allow for more sophisticated measurement procedures that render the associated Haar moment computations nontrivial. Our main technical contribution is a novel bound that addresses this setting.

\begin{theorem} \label{thm:main-technical-intro}
Fix a bipartite input state $|\phi \rangle \in \mathbb{C}^d \otimes \mathbb{C}^d$ and a measurement $M$ of compatible dimension. Then,
\begin{equation*}
\mathbb{E}_U \left[ \bigg(\Tr \Big( M \big(U \otimes \mathbb{I}\big) \ketbra{\phi} \big(U^\dagger \otimes \mathbb{I}\big) \Big) - \mathbb{E}_U \Big[ \Tr \Big( M \big(U \otimes \mathbb{I}\big) \ketbra{\phi} \big(U^\dagger \otimes \mathbb{I}\big) \Big)\Big]\bigg)^k\right] \leq \frac{C_k (k!)^2}{d^{k/2}}\;, \label{eq:unitary-moment-intro}
\end{equation*}
where $C_k=\tfrac{1}{k+1}\binom{2k}{k} < \tfrac{4^k}{k}$ denotes the $k$-th Catalan number.
\end{theorem}

This bound is considerably more general than existing ones in the literature and may be more generally applicable. Ref.~\cite{brandao_local_2016}, for instance, utilizes Eq.~\eqref{eq:state-moment-intro} only.
We establish this result by combining Schur-Weyl duality \cite{Fulton91,christandl_phd_2006} with Weingarten calculus \cite{Weingarten78,Collins04} and auxiliary arguments from tensor network theory \cite{bridgeman_handwaving_2017,kliesch_guaranteed_2017} and convex optimization \cite{Roc70:Convex-Analysis,Bar02:Course-Convexity}.
We believe that the dimensional scaling in the final bound is tight, but there may be room for further improving the $k$-dependent constants. In particular, we do not know if the Catalan number is necessary, or merely an artifact of our proof technique.

\subsection{Relation to previous work} \label{sec:previous}

Quantum complexity has recently become a popular subject in high energy physics. A considerable amount of attention has been devoted to understanding the complexity accumulated after an exponentially long time. 
A joint work by Susskind and Aaronson \cite{SusskindCCBH14,susskind2018black,aaronson2016complexity} points to an intriguing connection to theoretical computer science: unless $\textsc{PSPACE}\subseteq\textsc{BQP}/\textrm{poly}$ (a hypothetical relation between different computational complexity classes that is widely believed to be false),
the circuit complexity of certain Hamiltonian evolutions $U=\exp (-iHt)$ achieves super-polynomial values for exponentially long time scales $t$.
In a similar vein, Bohdanowicz and Brand\~{a}o \cite{bohdanowicz2017universal} constructed a family of Hamiltonians that provably achieves superpolynomial complexity in exponential time, unless $\textsc{PSPACE}=\textsc{BQP}$.

These arguments address late-time complexity and therefore do not yield insights regarding early-time complexity growth. In this regard, relations between complexity growth and approximate $k$-designs have recently been pointed out in \cite{ChaosDesign,ChaosRMT}.
Specifically, \cite{ChaosDesign} defined a notion of the complexity of generating an ensemble of unitaries and gave a lower bound on the ensemble complexity in terms of the distance to forming a unitary design. They also argued that the lower bound of the complexity of a $k$-design is linear in $k$.
Our arguments and results may be regarded as a substantial refinement of these ideas.

The notion of strong complexity put forward in our work has its conceptual roots in quantum information.
Encompassing this mindset is the statement from Ref.~\cite{gross_most_2009}: ``most states are too entangled to be useful as computational resources.''
At the core of this argument is the following observation. Haar-random pure states are 
so highly entangled that \emph{local} measurements 
yield almost uniformly random outcomes. 
In turn, any quantum device that relies on local measurements and uses known, but Haar-random, states could be efficiently simulated by tossing classical coins! This prevents any genuine quantum advantage for computation.

Strong state complexity (\autoref{def:state_complexity}) may be thought as a formal version of this observation. Measuring the maximally mixed state $\rho_0$ always results in a uniform outcome distribution. Moreover, Ref.~\cite{gross_most_2009} makes essential use of the fact that the measurements are constrained to be ``simple'' (in their case: local measurements augmented by classical post-processing).
The core of their argument may be summarized as follows: low complexity measurements do not allow for distinguishing a Haar-random state from the maximally mixed state. 
We present a variant of this argument in Sec.~\ref{sub:haar_complexity} below.

While \cite{gross_most_2009} only considers Haar-random pure states, similar arguments have been established for states that are less generic, see e.g.\  \cite[Section~3]{brandao_local_2016}.
Although not stated at this level of generality,
\cite[Corollary~10]{brandao_local_2016} effectively points out that states generated by approximate $k$-designs fool short quantum circuits: with high probability they cannot be distinguished from the maximally mixed state by means of any measurement with small circuit size.
They also extend this result to circuits \cite[Corollary~11]{brandao_local_2016}.
With high probability, a randomly selected (according to the weights) $k$-design element cannot be approximated by any short-sized circuit $V$ in the sense that $\| U - V \|_\infty$ is small. 

The second main result of our work, \autoref{thm:main_circuit}, improves upon this result in two ways. Firstly, the strong unitary complexity (\autoref{def:unitary_complexity}) is more stringent than their more traditional definition. While \autoref{thm:main_circuit} does imply \cite[Corollary~11]{brandao_local_2016}, the converse is not necessarily true.

Secondly, and more importantly, both Corollary~10 and 11 in \cite{brandao_local_2016} are probabilistic. 
While this is enough to deduce average-case behavior, a strong quantitative statement about the number of $k$-design elements with high circuit complexity remains beyond the scope of these assertions. 
A worst case caricature may help to illustrate this subtlety.
Suppose that the weights accompanying a unitary $k$-design are extremely spiky. 
A single high-complexity unitary, say $U_1 \in U(d)$ is accompanied by an exceedingly large weight $p_1 \simeq 1$, while all other design unitaries $U_i$ have low complexity and almost vanishing weights $p_i \simeq 0$.
Such a weight distribution would not contradict the assertion of \cite[Corollary~11]{brandao_local_2016}. The single high complexity circuit occurs with high probability (over the weights). Nonetheless, the hypothetical $k$-design does only contains a single high-complexity element.

Here we overcome this issue by explicitly ruling out the possibility of such extreme cases ever occurring. The definition of an approximate $k$-design alone implies that the weights cannot be too spiky, see \autoref{lem:weights-intro}.
This bound on the weights allows us to convert probabilistic (average case) statements into quantitative ones. Not only does the average circuit complexity grow linearly with the order $k$ of an approximate design, the absolute number of distinct circuits that have high complexity must also grow \emph{exponentially} with $k$.

Interest in state complexity has been stimulated by its potential role in quantum gravity and the AdS/CFT correspondence; see Sec.~\ref{sec:holography} for further discussion.
Recently, the relevance to holographic duality of \emph{computational pseudorandomness} has been emphasized. Specifically, the authors of  \cite{bouland2019computational} argue that one can construct two \emph{mixed} quantum states on the boundary ($A$ and $B$) such that both $A$ and $B$ can be efficiently prepared, yet $A$ and $B$ cannot be distinguished from maximally mixed states by polynomial-size quantum circuits. Furthermore, the corresponding bulk states ($A'$ and $B'$) \emph{can} be distinguished efficiently from one another. This observation indicates that the holographic dictionary which relates bulk and boundary states must have high computational complexity. 

We stress that this concept of \emph{pseudorandom quantum states}, which can be efficiently prepared yet cannot be distinguished from random by computationally bounded observers, is applicable to mixed states, or ensembles of pure states, but not to individual pure quantum states. If a particular pure state can be prepared efficiently by a quantum circuit, that state can always be distinguished efficiently from a maximally mixed state by running the circuit backwards. An ensemble of pure states can be pseudorandom only if it contains superpolynomially many pure states, where the observer who draws a sample from the ensemble and attempts to distinguish this sampled state from a maximally mixed state has no information about which sample was chosen. In contrast, in our definition of complexity for pure states, the observer is permitted to use a different distinguishing circuit for each possible pure state. 
On the other hand, the existence of pseudorandom quantum states \cite{ji2017pseudorandom} indicates that, for mixed states, our definition of state complexity, namely the computational cost of \emph{distinguishing} the state from a maximally mixed state, can differ substantially from another natural definition, the computational cost of \emph{preparing} the state. 

\section{Complexity growth in random circuits}
\label{sec:models}

The rigorous statements put forward in \autoref{thm:main_state} and \autoref{thm:main_circuit} gain additional meaning when applied to concrete examples. The literature contains several proofs of design growth in random circuits. Combining these with our rigorous insights yields a number of concrete models for complexity growth.

\subsection{Local random circuits}
\label{sub:local}

For concreteness, we focus here on systems comprised of $n$ qubits, i.e.\ $q=2$ and $d=2^n$.
Let $\mathsf{G} \subset U(4)$ be a (finite) universal gate set containing inverses, i.e.\ $g^{-1}=g^\dagger \in \mathsf{G}$ whenever $g \in \mathsf{G}$.
We can generate \emph{$\mathsf{G}$-local random circuits} by sequentially applying a random gate $g \in \mathsf{G}$ to a randomly selected pair of neighboring qubits.
Repeating this procedure independently for $T$ steps results in random circuits of size $T$.
We refer to the application of each gate as a time step, such that size $T$ circuits are of depth $T$ and have thus evolved to time $T$.
Intuitively, the larger $T$, the more random the circuit becomes. 
A seminal result by Brand\~{a}o, Harrow, and Horodecki confirms this intuition in a precise sense.

\begin{theorem}[Corollary~7 in \cite{brandao_local_2016}] \label{thm:brandao} \label{thm:local_random_circuits}
Fix $d=2^n$, $\epsilon >0$, $k \leq \sqrt{d}$, and let $\mathsf{G} \subset U(d^2)$ be a universal gate set containing inverses.\footnote{In addition to containing inverses, \cite{brandao_local_2016} also required that the gate set $\mathsf{G}$ be comprised of algebraic entries, but recent results suggest that both these restrictions may be relaxed \cite{RQCseeds19,RQCepnets19}.}
Then, the set of all $\mathsf{G}$-local random circuits of size $T$ forms an $\epsilon$-approximate $k$-design if
\begin{equation}
T \geq C n \lceil \log_2 (k) \rceil^2 k^{9.5} \big(nk+ \log \left(1/\epsilon\right)\big)\,, \label{eq:local_random_circuit}
\end{equation}
where $C>0$ is a (large) constant which depends on $\mathsf{G}$.
\end{theorem}

We emphasize that the weights associated with each unitary in this ensemble are defined implicitly by this random procedure. Several different $T$-sized circuits may give rise to the same final unitary, say $U_1$, while another one, say $U_2$, may exclusively be obtained from a single circuit geometry. The weights associated with $U_1$ and $U_2$ take into account this behavior, i.e.\ $p_1 \geq 2 p_2$ for our example. However, the fact that the entire ensemble still forms an approximate $k$-design limits potential fluctuations.
This in turn imposes lower bounds on the minimal number of distinct unitaries and severely limits the potential for collisions. It cannot be too likely that two or more different random circuits coincide. 
These features were conjectured by Brown and Susskind \cite[Sec.~6.5]{brown_second_2018} who, in turn, base their counting argument that relates circuit size and complexity on an extreme version of this conjecture: \emph{collisions do not occur at all}. 
One of the main results of this work is rigorous proof for a weaker version of their conjectured relation between circuit size and complexity. It is an immediate consequence of \autoref{thm:main_circuit} and \autoref{thm:brandao}.

\begin{corollary}[Polynomial relation between circuit size and circuit complexity for local random circuits] \label{cor:local_random_complexity}
Fix $\delta \in (0,1)$, $r \leq 2^{n/2}$ and set
$
T \geq C n^2 
\left( \frac{\log_2 (n) r}{n} \right)^{11}.
$
Then, the set of all $\mathsf{G}$-local circuits of size $T$ contains at least 
$ \tilde{C} 2^{\log (n) r}$ 
unitaries that obey $\mathcal{C}_\delta (U) >r$.
Here, $C, \tilde{C} >0$ are constants that implicitly depend on $\delta$ and $\mathsf{G}$.
\end{corollary}

This result establishes a polynomial relation between the length $T$ of $\mathsf{G}$-local circuits and the strong $\delta$-unitary complexity  that may be achieved in such a model.\footnote{Recently, \cite{HM18} showed that higher dimensional local random quantum circuits form approximate designs in $O(n^{1/D}{\rm poly}(k))$ depth, with some (high degree) polynomial dependence on $k$. \autoref{thm:main_circuit} then gives a polynomial growth of complexity for these higher dimensional circuits.}
The relation $T \simeq r^{11}$ a consequence of \autoref{thm:brandao},
which features a similar relation between the degree $2k$ of an approximate $2k$-design and the  circuit size $T$ required to implement it. This relation between complexity and circuit size can certainly be improved, which we will soon discuss, but there are fundamental limits: a lower bound on on the design depth for random circuits is known. A converse result (Proposition~8 in \cite{brandao_local_2016}) states that for $\epsilon\leq 1/4$ and $k\leq d^{1/2}$, the size of random circuits on $n$ qudits must be at least 
\begin{equation}
T \geq \frac{2kn \log q}{q^4 \log k}
\quad
\textrm{to form an $\epsilon$-approximate $k$-design}\,.
\end{equation}
See Sec.~\ref{sec:epcover} for a rederivation of this claim. 

\subsection{Relating two conjectures}

Fix $q=2$, $d=2^n$ ($n$ qubits) and 
suppose that the aforementioned lower bound were not only necessary, but also (approximately) sufficient: $\mathsf{G}$-local circuits of size $T \simeq \frac{2nk}{\log_2 (n)}$ generate (sufficiently accurate) approximate $2k$-designs. Under this assumption, $\mathsf{G}$-local random circuits of size $T$ contain at least
$d^{2k}/(k!)^2$ elements with circuit complexity $r \simeq T$. 
If we assume that $T \leq \tfrac{2n}{\log_2 (n)} \sqrt{d}$, then this bound can be simplified further as
\begin{equation}
K \gtrsim \frac{d^{2k}}{(k!)^2} 
= 2^{2nk -2 \log (k!)}
\gtrsim 2^{2k (n  - \log (k) )}
= d^{\frac{k}{2}}= 2^{\frac{nk}{2}} \simeq 2^{\log_2(n)T} \geq 2^T\,.
\end{equation}
This is essentially 
\autoref{conj:lenny}: the set of all $\mathsf{G}$-local circuits of size $T$ contains an exponentially growing set of elements with complexity $r \simeq T$.
This observation provides a relation between Conjecture~\ref{conj:lenny} (linear growth in complexity) to an existing Conjecture in quantum information \cite{brandao_local_2016}:

\begin{conjecture}[Linear growth in design]
$\mathsf{G}$-local circuits of size $T=O(nk)$ form approximate $k$-designs.
\label{conj:lineark}
\end{conjecture}

\ni To achieve a linear growth in complexity we need a linear growth in design.

\subsection{Linear growth in design for local random circuits at large local dimension}

We again consider a 1d system comprised of $n$ qudits of local dimension $q$, with total dimension $d=q^n$, and evolve the system by a random circuit consisting of local 2-site unitaries drawn Haar-randomly from $U(q^2)$. The results of \cite{brandao_local_2016} also ensure that such random circuits form approximate $k$-designs when the size is $O(n^2 k^{11})$. Although \autoref{conj:lineark}, a linear design growth in $\mathsf{G}$-local random circuits with local qubits, is currently out of reach, progress was made recently in \cite{NHJ19}, improving the $k$-dependence for Haar-local random circuits in the limit of large local dimension and giving a linear growth in the circuit size to form a unitary $k$-design. 

\begin{theorem}[\cite{NHJ19}] Random quantum circuits on $n$ qudits of local dimension $q$ form approximate unitary $k$-designs when the circuit size is $T = O(n^2k)$ for some $q>q_0$, where $q_0$ depends on the size of the circuit.\footnote{We note that \cite{NHJ19} computed the circuit depth, whereas the discussion here involves the circuit size, giving an extra factor of $n$.}
\label{thm:lineark}
\end{theorem}

The approach of \cite{NHJ19} was to consider the frame potential, capturing the 2-norm distance to forming an approximate design, and make use of an exact statistical mechanical mapping \cite{NVH17,RQCstatmech} in order to write the frame potential as the partition function of a triangular lattice model. The contributions to the partition function can be interpreted as domain walls in the lattice model. In the limit of large $q$, \cite{NHJ19} showed that only a simple sector of domain walls contribute, allowing for the calculation of the $k$-design circuit depth. More precisely, by computing the single domain wall terms and showing that the multidomain wall terms contribute at subleading order in $1/q$, it was proved that local random circuits exhibit a linear growth in design for some $q>q_0$, where $q_0$ depends on the circuit size $T$ and moment $k$.

\autoref{thm:lineark} and \autoref{cor:local_random_complexity} allow us to establish \autoref{conj:lenny} for local random circuits with Haar-random 2-site unitaries in the limit of large $q$. 

\begin{corollary}[Linear complexity growth]
Given the set of local random circuits of size $T$ at large $q$, most circuits have strong complexity $\Omega(T)$, \ie growing linearly in $T$ for a long time.
\end{corollary}

Although the \autoref{thm:main_circuit} still applies for local Haar-random quantum circuits, giving a lower bound on the number of distinct unitaries with high complexity. Its meaning becomes less clear when we have a continuous ensemble. We can consider an ensemble of finite cardinality by constructing an $\varepsilon$-covering of the set of random circuits. We review the notion of an $\varepsilon$-covering in Sec.~\ref{sec:epcover} and give a bound on the cardinality of a covering for local random circuits. Constructing a coarse net then shows that exponentially many random quantum circuits, with Haar-random 2-site unitaries, have high complexity.

Lastly, we emphasize that we have not proved linear complexity growth up to time scales of order $d$. While taking a large enough $q$ will ensure linear design growth for times exponential in $n$, such a limit still pushes the true exponential time scales of interest, $t\sim d=q^n$, out of reach. Proving an optimal design growth for local random circuits away from the large $q$ limit would allow us to better probe late-time complexity.

\subsection{Brownian circuits/Stochastic quantum Hamiltonians}
There also exist continuous-time models of chaotic dynamics, analogous to random circuits, which scramble in $O(\log n)$ time \cite{FastScrambling}. In a similar spirit to models of random walks on the unitary group, one can define a one-parameter family of Hamiltonians which generate a time-dependent unitary evolution.
The Hamiltonian on $n$ qubits at a time step $s$ is given by a sum of random all-to-all 2-body interactions, meaning we sum over all possible 1 and 2-local interactions with independently chosen Gaussian random couplings
\begin{equation}
H_s = \sum_{i<j}\sum_{\alpha,\beta} J_{s,i,j,\alpha,\beta} S_i^\alpha S_j^\beta\,,
\label{eq:stocH}
\end{equation}
where $S_i^\alpha$ is a Pauli operator acting on site $i$ with $\alpha=\{0,1,2,3\}$. The couplings are each drawn independently from a Gaussian distribution with zero mean and variance $\sigma^2$. Not only are the couplings random in space, but are further chosen randomly at each time-step $s$. 
The time evolution to time $t$ is then given by
\begin{equation}
U_t = \prod_{s=1}^t \mathrm{e}^{-i H_s \delta t}\,,
\end{equation}
where we consider the continuum limit $\delta t\ra 0$ with the variance of the couplings scaling as $\sigma^2 = J/\delta t $ so that the interactions strength increases proportional to the inverse time step and where $J$ is a constant. 

It was shown in \cite{onorati_mixing_2017}, using similar techniques to \cite{brandao_local_2016}, that these Brownian circuits form $k$-designs in polynomial time. 

\begin{theorem}[Corollary~10 in \cite{onorati_mixing_2017}] \label{thm:brownian_designs}
For $d=2^n$ and $\epsilon >0$, Then the ensemble of time-evolutions by stochastic Hamiltonians in Eq.~\eqref{eq:stocH}, forms an $\epsilon$-approximate $k$-design for times
\begin{equation}
t \geq C \lceil \log_2 (k) \rceil^2 k^{9.5} (nk + \log (1/\epsilon))\,, 
\label{eq:brownian_designs}
\end{equation}
where $C>0$ is a constant.
\end{theorem}

For the Brownian circuit models, the constant prefactor $C$ depends on the local dimension, here chosen to be 2, but also on the interaction strength of the couplings $J$, $C \sim 1/J$, meaning if the interactions are stronger then the depth required to form a design decreases accordingly. 

We can again use the polynomial relation between complexity and design to discuss complexity growth. \autoref{thm:main_circuit} and \autoref{thm:brownian_designs} together give that Brownian circuits have a complexity growing polynomially in time as $\Omega(t^{1/11})$. 

\subsection{Nearly time-independent Hamiltonian dynamics}
There is another random quantum circuit-like construction of a time-dependent Hamiltonian with varying couplings over discrete time steps. This ``nearly time-independent'' model of \cite{Nakata16} forms $k$-designs in a depth $O(n^2 k)$ up to moments $k=o(\sqrt{n})$, achieving the nearly optimal lower bound with a linear growth in design for a short time.

Consider a 1d system of $n$ qudits, with $d=q^n$, and define a time-dependent set of random couplings
\begin{equation}
\CJ(t,g) = \Big\lbrace \lambda/(\lfloor t/2\rfloor +1)\,,~~ \lambda\in [-g/2,g/2]\Big\rbrace\,,
\end{equation}
as well as two ensembles of Hamiltonians with time-dependent couplings
\begin{align}
\CE_Z(t) &= \bigg\lbrace -\sum_{j<k} h_{jk} Z_j Z_k - \sum_j b_j Z_j\,, ~~{\rm with}~~ h_{jk}\in \CJ(t,h)\,,~ b_j\in \CJ(t,b)\bigg\rbrace \\
\CE_X(t) &= \bigg\lbrace -\sum_{j<k} h_{jk} X_j X_k - \sum_j b_j X_j\,, ~~{\rm with}~~ h_{jk}\in \CJ(t,h)\,,~ b_j\in \CJ(t,b)\bigg\rbrace \,,
\end{align}
where $g=\lfloor t/2\rfloor /2$ and $b = \lfloor t/2\rfloor +1/2$. We then define the time-evolution of our system: we evolve by an $X$-type Hamiltonian $H_X \sim  \CE_X $ at even time steps and a $Z$-type Hamiltonian $H_Z \sim \CE_Z$ at odd time steps. Then the unitary time-evolutions form an $\epsilon$-approximate $k$-design for $k = o(n^{1/2})$, after $T$ time steps, where
\begin{equation}
T\geq (k+1/2 + (1/n) \log_2 (1/\epsilon))\,.
\end{equation}

This construction forms unitary $k$-designs almost linearly in time, with the caveat that the time scale is limited to $\sim \sqrt{n}$. Thus we get a linear growth in design at early times, but not exponentially in $n$. Consequently, this implies a linear growth in complexity at (very) early times. 

\subsection{Comment on time-independence}
We have discussed a few explicit models of complexity growth in systems that are random in both space and time. As we emphasized, one of our results is that a polynomial growth in design implies a polynomial growth growth in complexity (\autoref{cor:main_complexity_growth}). Thus, the random circuit and Brownian circuit models, which form approximate $k$-designs in poly$(k)$ depth, are also explicit examples of systems with a long time polynomial growth in complexity.

But for an ensemble of time-evolutions to form a $k$-design, randomness in time is likely essential. For instance, consider an ensemble of time-evolutions generated by an ensemble of Hamiltonians, $\CE_t = \{ \mathrm{e}^{-iHt}, ~H\in\CE_H\}$, where $\CE_H$ could be a disordered spin system or any ensemble of random Hermitian matrices. The rigid structure of eigenvalues then prohibits the late-time Haar-randomness.

Interestingly, the Gaussian unitary ensemble (GUE), an ensemble of $d\times d$ random Hermitian matrices with a unitarily-invariant measure, does come close to an approximate $k$-design in 2-norm for moments $k\ll d$ at a specific time-scale $t\sim \sqrt{d}$ \cite{ChaosRMT}. But at later times, the 2-norm distance between the ensemble of unitaries generated by GUE Hamiltonians and the Haar ensemble becomes large. More generally, one expects that any ensemble of unitary evolutions by time-independent Hamiltonians will not form a $k$-design at late times. A general argument for this is as follows \cite{ChaosDesign}, under the exponential map $\lambda \ra \mathrm{e}^{i\lambda t}$, the eigenvalues of a Hamiltonian are distributed as time-evolving phases on the unit circle. In the limit $t\ra \infty$, the phases become uncorrelated and uniformly distributed, unlike the correlated and logarithmically repelling eigenvalues of Haar random unitaries. Thus, to understand the complexity growth of (ensembles of) time-independent Hamiltonian evolution, we would need to look beyond their design properties for an alternative approach, for instance, by studying the approximate invariance of the ensemble \cite{ChaosRMT,kinvchaos}.

\section{Complexity in holographic systems} \label{sec:holography}

Much of the recent interest in quantum complexity in the high-energy literature has centered on the conjectured relation between complexity growth and the long-time growth of black holes interiors \cite{SusskindCCBH14,SScomp14,SusskindEnt14}.
More specifically in the context of the AdS/CFT correspondence, the region behind the horizon of an eternal AdS-Schwarzschild black hole grows linearly in time for an exponential time ($t\sim \mathrm{e}^n)$. 
The holographic picture is a two-sided geometry connected by a wormhole, where the throat of the wormhole is growing in time. 
The claim is that the quantum complexity of the dual CFT state is the long-time linearly increasing quantity which captures the wormhole growth. There have been a number of proposals for what bulk quantity actually computes the complexity, including the volume and action of the AdS wormhole. The complexity/volume conjecture states the the computational complexity of the boundary state is equal to the volume of the wormhole. More precisely, the complexity of time-evolved thermofield double state of the two boundary CFTs is equal to spatial volume behind the horizon of the two-sided geometry on a maximal time slice \cite{SScomp14}. The `complexity equals action' conjecture posits that the action computed on a certain region of the bulk geometry which extends behind the horizon (the Wheeler-DeWitt patch), is the quantity which is dual to the complexity \cite{CABH15,CA15}. A lot of progress has been made checking these conjectures and studying complexity growth in holographic systems, see for instance \cite{holocomp1,holocomp2,holocomp3,holocomp4,holocomp5,holocomp6,holocomp7}. 

In this work we have rigorously computed the complexity growth in a number of random circuit models, by relating the growth in design to the growth of complexity, and were able to prove a linear growth in complexity for local random circuits in the limit of large local dimension (albeit, not for an exponentially long time). As we mentioned, the connection between unitary designs and quantum complexity will likely not inform complexity growth in holography as evolution by time-independent Hamiltonians will not converge to approximate designs. Thus, to study complexity growth in holography we need to explore properties beyond the Haar-randomness of the evolution. 

\subsubsection*{Strong complexity in the bulk}
We will briefly discuss why we believe our proposed strong definition of complexity (in terms of a distinguishing measurement), is congruent with expectations from the bulk and might be more suited for holography than the standard definition in terms of the circuit complexity. 

One feature we expect complexity growth will exhibit in holography, and fast scrambling systems more generally, is the switchback effect \cite{SScomp14}. Consider a time-evolved local operator $\op(t) = \mathrm{e}^{-iHt} \op \mathrm{e}^{iHt}$ (sometimes called a precursor), where $\op$ might be a single site Pauli. For such an operator, we anticipate a delay in the onset of the linear complexity growth. For the traditional definition of complexity, consider the minimal circuit approximating the evolution operator $\mathrm{e}^{-iHt}$. The reason for this delay is the exact cancellation of gates outside the lightcone of the spreading operator. Once the operator grows to be the size of the system (more precisely, to have support on weight $n$ Pauli operators) after a timescale called the scrambling time, we expect the complexity of $\op(t)$ to begin its long time linear growth. Such an effect is also present in the bulk for both complexity-volume and action conjectures. This feature is also present in complexity growth of $\op(t)$ under the strong definition of complexity in \autoref{def:state_complexity}. To be concrete, consider a system of $n$ qubits and the evolved state $\mathrm{e}^{-iHt} \op \mathrm{e}^{iHt}\ket{\psi_0}$, where $H$ is a chaotic but local Hamiltonian and we take $\ket{\psi_0}$ to be an unentangled product state. Prior to the scrambling time, the optimal measurement to distinguish the evolving state from the maximally mixed state is a simple measurement of a qubit outside the lightcone of the evolving operator. It is not until the scrambling time, when operator has grown to have support on all sites, that the complexity of the distinguishing measurement begins to grow. 

A more interesting example, where the strong and weak definitions of complexity differ, is that of one-clean qubit. This is essentially the argument given in \autoref{lem:state_comp_stronger}, to prove that measurement complexity is a stronger definition than standard circuit complexity. Consider a simple initial state $\ket{\psi_0}$, which has been evolved for an exponential time such that $\ket{\psi(t)}$ is maximally complex. If we add a single unentangled qubit to the state $\ket{\psi(t)}\otimes \ket{0}$, then the minimal circuit complexity will be unchanged, but maximal potential complexity increases and the complexity of the state can continue to grow for a long time until it saturates at the new maximal value. For the complexity of a distinguishing measurement, adding a single clean qubit resets the complexity to an order one value, as the optimal measurement is simply the projection onto the clean qubit.  Ref.~\cite{brown_second_2018} proposed the notion of uncomplexity as the difference of the complexity of a state or unitary from its maximal complexity, where uncomplexity can be thought of as a resource to do useful computation. As we described, our strong definition of complexity directly encodes this potential for useful quantum computation.

\subsubsection*{Entanglement growth by design}
The suggestion that complexity be the dual of the long-time geometric growth in the bulk was motivated by the observation that the wormhole grows long past the time-scales at which entropic quantities saturate. Given that we have discussed long-time growth in complexity from a long-time growth in design, it is worth commenting on the saturation of entropies after a short growth in design order.

The entanglement entropies for $k$-designs were studied in \cite{Liu2018}.
Specifically, they looked at the R\'enyi-$\alpha$ entropies of a density matrix $\rho$: $S^{(\alpha)}(\rho) = \tfrac{1}{1-\alpha} \log \left(\Tr (\rho^\alpha) \right)$. For any state, the R\'enyis are bounded above and below by the min-entropy $S_{\min}(\rho) := \lim_{\alpha \to \infty}S^{(\alpha)} (\rho) = -\log (\|\rho \|_\infty)$.\footnote{To see this, recall the relation between Schatten-$\alpha$ norms in $d$ dimensions: $\|\rho \|_\infty \leq \|\rho\|_\alpha \leq d^{1/\alpha}\|\rho \|_\infty$. This ensures that for any state $\rho$, we have
$S_{\rm min} (\rho) \leq S^{(\alpha)}(\rho) \leq S_{\rm min}(\rho) + \tfrac{\log d}{\alpha}$. Note that as we take $\alpha$ to be greater than $n$, these R\'enyi entropies concentrate ever sharper around the min-entropy.} For an $n$-qubit system, consider the reduced density matrix $\rho_A = \Tr_{\bar A} \ketbra{\psi}$ on a subsystem $A$ consisting of half the qubits, so that $d_A = d_{\bar A}$. Ref.~\cite{Liu2018} showed that for states $\ket\psi$ drawn from a $(k>\log d)$-design, the min-entropy of $\rho_A$ is nearly maximal.
Therefore, all entropies are nearly maximal once the design order is $k\approx n$. Considering then the time-evolved states of a fast-scrambling system which forms unitary designs linearly in time, all entropies will saturate at a time of order $n$. Our arguments ensure complexity growth of approximate $k$-designs well beyond this entropy saturation threshold. 

\section{Proof of the main results}
\label{sec:results}

\subsection{Motivating example computations for Haar random states}\label{sub:haar_complexity}

\subsubsection{Most states have high complexity}

The Hilbert space of $n$ qudits is enormous, $d=q^n$. Nonetheless, only a tiny fraction of all possible (pure) quantum states seems to be useful for quantum computation, see e.g.\  \cite{gross_most_2009}. 
Strong state complexity (\autoref{def:state_complexity}) captures this curious aspect. In order to get a quantitative handle on the set of all pure states we endow it with the uniform measure $\mathrm{d} \psi$ that is induced by the Haar measure on the unitary group $U(d)$. 
Then, random pure states $| \psi \rangle \! \langle \psi|$ behave like the maximally mixed state $\rho_0$ in expectation. This behavior extends to the outcome statistics of arbitrary (fixed) measurements:
\begin{equation}
\mathbb{E}_{\ket{\psi}} \left[ \Tr \left( M | \psi \rangle \! \langle \psi| \right) \right] = \Tr \left( M \mathbb{E}_{\ket{\psi}} \left[ | \psi \rangle \! \langle \psi| \right] \right) = \Tr \left( M \rho_0 \right)\,.
\end{equation}
Concentration of measure (Levy's lemma) ensures that deviations from this average case behavior are exponentially suppressed in concrete instances:
\begin{equation}
\mathrm{Pr} \left[ \left| \Tr \left( M (| \psi \rangle \! \langle \psi| - \rho_0 ) \right) \right| \geq \tau\right] \leq 2 \exp \left( - \frac{d \tau^2}{9 \pi^3} \right) \quad \textrm{for any} \quad \tau \geq 0.
\label{eq:haar_concentration}
\end{equation}
We refer to \autoref{prop:measure_concentration} in the appendix for a proof of this well-known result.
We can combine this assertion with a union bound (Boole's inequality) to conclude for any $r \in \mathbb{N}$ and $\delta \in (0,1)$
\begin{align}
\mathrm{Pr} \left[ \mathcal{C}_\delta (| \psi \rangle) \leq r \right]
&= \mathrm{Pr} \left[ \max_{M \in \mathsf{M}_r} \left| \Tr \left( M (| \psi \rangle \! \langle \psi| - \rho_0) \right) \right| \geq 1- d^{-1} -\delta \right] \nn
&\leq 2.0072 |\mathsf{M}_r| \exp \left( - \frac{d (1- \delta)^2}{9 \pi^3} \right)\,.
\label{eq:haar-probability}
\end{align}
Suppose that $\mathsf{M}_r$ arises from combining at most $r$ elements of a fixed universal gate set $\mathsf{G} \subset U(q^2)$. A naive counting argument reveals $\left| \mathsf{M}_r \right| \leq 2d n^r | \mathsf{G}|^r$. We conclude that the $\mathrm{Pr} \left[ \mathcal{C}_\delta (| \psi \rangle) \leq r \right]$ remains exponentially suppressed (in $d=q^n$) until 
\begin{equation}
r \simeq \frac{q^n}{\log (n)}\,.
\label{eq:haar_complexity}
\end{equation}
To summarize, a random state is exceedingly likely to have an exponentially large strong $\delta$-state complexity. 

The Haar measure has another desirable property. It is fair in the sense that it assigns the same (infinitesimal) weight to each pure state. Such perfectly flat probability distributions allow for turning the probabilistic statement \eqref{eq:haar-probability} into a quantitative one. From the definition of probability, $\mathrm{Pr}\left[ C_\delta (|\psi \rangle) \leq r \right]$ corresponds to the ratio of low complexity states over all states.
Thus, Eq.~\eqref{eq:haar-probability} ensures that the fraction of low complexity states remains exponentially tiny until $r \simeq q^n/\log (n)$. In other words: \emph{most pure states have exponentially large complexity.}

\subsubsection{Most high-complexity states are far apart}\label{subsec:far-apart}

In the previous subsection, we saw that concentration of measure \eqref{eq:haar_concentration} allows us to conclude that most quantum states have exponentially high state complexity. 
This argument, however, does not (yet) tell us anything about the geometric separation between high complexity states. In principle, a large fraction of high-complexity states could result from tiny perturbations of only a few well-separated high-complexity core states. In other words, high-complexity states could come in few tightly packed clusters, in which case the effective number of high-complexity regions could still be comparatively small. 

The probabilistic method \cite{alon_probabilistic_2016} allows us to prove that extreme clustering cannot occur: \emph{there are exponentially many high complexity states whose pairwise distance is almost maximal.}

We show this statement by induction based on two features of Haar random states. 
Firstly, we use the main result from the previous subsection. Choose $r \lesssim q^n/\log (n)$ such that Eq.~\eqref{eq:haar-probability} ensures
\begin{equation}
\mathrm{Pr} \left[ \CC_\delta (| \psi \rangle) \leq r \right] \leq 2.0072 | \mathsf{M}_r| \exp \left(- \frac{d(1-\delta)^2}{9 \pi^3} \right) \leq \frac{1}{2}\,.
\label{eq:probabilistic-method1}
\end{equation}
The parameter $r$ is chosen such that Haar random states exceed this complexity with probability $1/2$. 
Concentration of measure also implies that a Haar-random state is very likely to be far away from any fixed state $| \phi \rangle \! \langle \phi|$. For any $\Delta \in (0,1)$,
\begin{align}
\mathrm{Pr} \left[ \tfrac{1}{2} \left\| \ketbra{\psi} -\ketbra{\phi} \right\|_1 \leq 1-\Delta \right] = \mathrm{Pr} \left[ | \langle \psi | \phi \rangle|^2 \geq \Delta^2 \right]
\leq
3\exp \left(-\frac{\Delta^2d}{9 \pi^3}\right). \label{eq:probabilistic-method2}
\end{align}
This bound readily follows from Eq.~\eqref{eq:haar_concentration} (set $M=| \phi \rangle\! \langle \phi|$) and elementary modifications.

The first step in our inductive argument is simple. Eq.~\eqref{eq:probabilistic-method1} asserts that the probability of Haar-randomly sampling a low complexity (at most $r$) state is smaller than $1/2$. This is equivalent to stating that the probability of Haar-randomly sampling a high complexity (larger than $r$) is at least $1/2$. Importantly, this assertion implies  that such a state exists, because the probability of sampling one is strictly positive.
Choose one such state $| \phi_1 \rangle$ as the first state in our list.

To construct the second state in our list, we refine this probabilistic existence argument. 
The probability of Haar-randomly sampling a low complexity state \emph{or} a state that is too close to $| \phi_1 \rangle$ is bounded by
\begin{align}
&\mathrm{Pr} \left[ \CC_\delta (|\psi \rangle) \leq r \cup \tfrac{1}{2} \| \ketbra{\psi}-\ketbra{\phi_1}\|_1 \leq 1- \Delta \right] \nn
&\qquad \leq
\mathrm{Pr} \left[ \CC_\delta (|\psi \rangle) \leq r \right] 
+ \mathrm{Pr} \left[ \tfrac{1}{2} \| \ketbra{\psi} - \ketbra{\phi_1} \|_1 \leq 1-\Delta \right] 
\leq  \frac{1}{2}+3 \exp \left( - \frac{\Delta^2 d}{9 \pi^3}\right).
\end{align}
By contraposition, the probability of sampling a state that has high complexity \emph{and} is simultaneously far away from $| \phi_1 \rangle$ is at least $\tfrac{1}{2}-3 \exp (-\tfrac{\Delta^2d}{9 \pi^3})>0$. 
This implies the existence of such a state. Choose one such state $| \phi_2 \rangle$ and append it to the list: $\left\{ \ket{\phi_1},\ket{\phi_2}\right\}$.

We can now inductively repeat this probabilistic existence argument and iteratively append distant high-complexity states to the list $\left\{ \ket{\phi_1},\ldots,\ket{\phi_N} \right\}$. This construction only breaks down once the list cardinality $N$ counterbalances exponential suppression:
$\frac{1}{2}-3N \exp (-\tfrac{\Delta^2 d}{9 \pi^3}) \leq 0$,
or equivalently
$N \geq  \tfrac{1}{6}\exp (\tfrac{\Delta^2d}{9 \pi^3})$.
Beyond this threshold, we cannot use simple union bounds and concentration of measure to ensure existence of another list element.
Such a threshold, however, scales exponentially in the Hilbert space dimension. We conclude that the list $\left\{ \ket{\phi_1},\ldots,\ket{\phi_N}\right\}$ contains $N = \tfrac{1}{6} \exp (\tfrac{\Delta^2 d}{9 \pi^3})$ 
high complexity states whose pairwise trace distance is at least $1-\Delta$.

We conclude this subsection with providing a bit of context. Existence proofs based on strictly positive probabilities date back to Erd\H{o}s who developed them to solve several important problems in graph theory. Today, this technique is known as the \emph{probabilistic method}  and does constitute an important tool in applied math, combinatorics, and theoretical computer science  \cite{alon_probabilistic_2016}.

\subsection{Proof of \autoref{thm:main_state}}

Haar-random states result from applying a Haar-random unitary $U \in U(d)$ to an arbitrary starting state, say $| \psi_0 \rangle$. 
Now suppose that this unitary $U$ is not chosen from the Haar measure, but from an approximate $2k$-design. By definition, this ensures that the first $2k$ moments of $| \psi \rangle \! \langle \psi| = U | \psi_0 \rangle \! \langle \psi_0| U^\dagger$ accurately approximate the corresponding Haar moments.
While this is too coarse to deduce exponential concentration \eqref{eq:haar_concentration} (this would require matching behavior for \emph{all} moments), polynomial concentration arguments do apply. 
Fix a measurement $M \in \mathbb{H}_d$ and let $\bar{M} = M - \frac{\Tr(M)}{d} \mathbb{I}$ denote its traceless part. Markov's inequality then implies that for any $\tau >0$
\begin{align*}
\mathrm{Pr} \left[ \left| \Tr \left( M (\psi \rangle \! \langle \psi|-\rho_0) \right)\right| \geq \tau \right] = \mathrm{Pr} \left[ \left( \Tr \left( \bar{M} | \psi \rangle \! \langle \psi| \right) \right)^{2k} \geq \tau^{2k} \right]
\leq \tau^{-2k} \mathbb{E} \left[ \Tr \left( \bar{M} | \psi \rangle \! \langle \psi| \right)^{2k} \right]\,.
\end{align*}
The final expectation value corresponds to a moment of order $2k$. This is the largest moment that still approximately exhibits Haar-random behavior. Explicit bounds can be derived by exploiting this similarity and we refer to \autoref{prop:state_moments} below for a technical derivation:
\begin{equation}
\mathrm{Pr} \left[ \left| \Tr \left( M (\ketbra{\psi}-\rho_0) \right)\right| \geq \tau \right]
\leq (1+\epsilon) \left( \frac{2k}{\tau \sqrt{d}}\right)^{2k}\,.
\end{equation}
Qualititatively, this deviation bound is proportional to $d^{-k}$ and becomes ever more stringent as the design order $2k$ increases.  
We can now combine this tail bound with a union bound and a counting argument for the measurement set $ \mathsf{M}_r$ in a fashion analogous to the Haar random case. For any $r \in \mathbb{N}$ and any $\delta \in (0,1)$ this yields
\begin{align}
\mathrm{Pr} \left[ \mathcal{C}_\delta (|\psi \rangle) \leq r \right]
\leq & | \mathsf{M}_r | (1+\epsilon) \left( \frac{2k}{\sqrt{d} \left( 1- d^{-1} - \delta \right)} \right)^{2k} 
\leq   
2(1+\epsilon) d n^r |\mathsf{G}|^r \left( \frac{16k^2}{d(1-\delta)^2} \right)^{k}\,,
\label{eq:probability_bound_states}
\end{align}
where we have tacitly assumed $(1-\delta) \geq 2 d^{-1}$ in the last step.
Qualitatively, this probability remains tiny until 
\begin{equation}
r \simeq \frac{(k-1) n - 2k \log (k)}{\log (n)+\log (| \mathsf{G}|)}\simeq \frac{k (n- 2\log (k))}{\log (n)}\,,
\end{equation}
provided that $n \geq | \mathsf{G}|$ and $k < d/2$.
So far, this is a purely probabilistic statement. In contrast to the Haar-uniform case it is a priori not clear whether it is possible to transform it into a quantitative one. The reason for this is twofold: (i) the weights $p_j$ associated with different elements from an approximate $2k$-design are typically \emph{not} uniform. This non-uniformity extends to the distribution over the different states $| \psi_i \rangle$. (ii) collisions in the state generation: two (or more) distinct design unitaries can produce the same state.

Fortunately, the defining properties of designs ensure that these deviations cannot be too radical: the weights associated with \emph{distinct} states $| \psi_i \rangle$ must obey $q_j \leq (1+\epsilon) \binom{d+k-1}{k}^{-1}$ -- see \autoref{lem:state_weights} below. 
This extra condition does allow for drawing quantitative conclusions.
Recall that the probability of an event $E$ is the expected value of its indicator function $\1den \left\{E \right\}$. Therefore,
\begin{equation}
\mathrm{Pr} \left[ \mathcal{C}_\delta (| \psi \rangle ) >r \right]
= \sum_j q_j \1den \left\{ \mathcal{C}_\delta (|\psi \rangle) > r \right\}
\leq (1+\epsilon) \binom{d+k-1}{k}^{-1} \sum_j \1den \left\{ \mathcal{C}_\delta (| \psi \rangle) >r\right\}\,.
\end{equation}
The sum on the r.h.s.\ is simply the cardinality $N$ of the set of states $| \psi \rangle$ with $\delta$-state complexity greater than $r$ and the l.h.s. is $1- \mathrm{Pr} \left[ \mathcal{C}_\delta (| \psi \rangle)\leq r \right]$. Substituting the bound \eqref{eq:probability_bound_states} into this expression establishes the claim:
\begin{equation}
N \geq \binom{d+k-1}{k} \left( \frac{1}{1+\epsilon}- 2 d n^r |\mathsf{G}|^r \left( \frac{16k^2}{d(1-\delta)^2} \right)^{k} \right)\,.
\end{equation}

\subsection{Proof of \autoref{thm:main_circuit}}

The proof is largely analogous to the proof of \autoref{thm:main_state}.
Fix a measurement $M \in \mathbb{H}_d \otimes \mathbb{H}_d$ and an input state $| \phi \rangle \in \mathbb{C}^d \otimes \mathbb{C}^d$. Recall that the bias of distinguishing a unitary channel $\mathcal{U}: \mathbb{H}_d \to \mathbb{H}_d$ from the depolarizing channel $\mathcal{D}$ via this measurement procedure is
$
\Tr \left( M \left( \mathcal{U} \otimes \mathcal{I} - \mathcal{D} \otimes \mathcal{I} \right) ( | \phi \rangle \! \langle \phi|) \right)
$. Moreover, the depolarizing channel corresponds to the Haar average over all unitary channels: $\mathbb{E}_U (\mathcal{U})=\mathcal{D}$, see e.g.\ \autoref{lem:expectation} below.
Now suppose that the corresponding unitary $U \in U(d)$ is chosen randomly from an $\epsilon$-approximate $2k$-design. Markov's inequality yields
\begin{align}
&\mathrm{Pr} \left[ \left|\Tr \left( M \mathcal{U} \otimes \mathcal{I}(| \phi \rangle \! \langle \phi| \right)
-\Tr \left( M \mathcal{D} \otimes \mathcal{I}(| \phi \rangle \! \langle \phi| \right)
\right|\geq \tau\right] \nn
&\leq \tau^{-2k}\, \mathbb{E} \left[ \left(\Tr \left( M \mathcal{U} \otimes \mathcal{I}(| \phi \rangle \! \langle \phi| \right)
-\Tr \left( M \mathcal{D} \otimes \mathcal{I}(| \phi \rangle \! \langle \phi| \right)
\right)^{2k} \right]\,.
\end{align}
The final expectation value corresponds to the highest $2k$-design moment that still approximates Haar-random behavior. Our main technical contribution in \autoref{thm:main-technical-intro} establishes tight bounds on such Haar random moments. These generalize approximate $2k$-design ensembles $\mathcal{E}$ in a relatively straightforward fashion:
\begin{equation}
\mathbb{E}_{\mathcal{E}} \left[ \left(\Tr \left( M \mathcal{U} \otimes \mathcal{I}(| \phi \rangle \! \langle \phi| \right)
-\Tr \left( M \mathcal{D} \otimes \mathcal{I}(| \phi \rangle \! \langle \phi| \right)
\right)^{2k} \right]
\leq \frac{((2k)!)^2}{d^k} \left( C_{2k} + \frac{\epsilon}{(2k)!d^{3k}}\right)\,.
\end{equation}
See \autoref{cor:master} below for a precise statement and proof.
Next, we emphasize that the crude bound $| \mathsf{M}_r| \leq (2d^2+1) n^{2r} |\mathsf{G}|^r$ applies to circuit measurements. Combining the above concentration inequality with a union bound over all measurements $M \in \mathsf{M}_r$ ensures
\begin{align}
\mathrm{Pr} \left[ \mathcal{C}_\delta (U) \leq r \right] 
\leq 3\left(C_{2k}+\frac{\epsilon}{(2k)!d^{3k}} \right) d^2 n^{2r} |\mathsf{G}|^r \left( \frac{64 k^4}{d(1-\delta)^2}\right)^k
\,,
\label{eq:probability_bound_channels}
\end{align}
where we have tacitly assumed $(1-\delta) \geq 2d^{-1}$.
Qualitatively, this probability remains tiny until
\begin{equation}
r \lesssim \frac{(k-2)n - 4k \log (k))}{\log (n)+ \log| \mathsf{G}|}
\simeq \frac{k (n - 4\log (k))}{\log (n)}\,,
\end{equation}
provided that $n \geq | \mathsf{G}|$ and $k \leq d/2$. 
The definition of an approximate $2k$-design also imposes constraints on the weight fluctuations. \autoref{lem:weights-intro} asserts that weights associated with distinct ensemble unitaries must obey $p_j \leq (1+\epsilon) \frac{k!}{d^{2k}}$.
This approximate flatness allows us to turn the probabilistic statement from above into a quantitative one:
\begin{equation}
\mathrm{Pr} \left[ \mathcal{C}_\delta (U) >r \right]
= \sum_j p_j \1den \left\{ \mathcal{C}_\delta (U)>r \right\}
\leq (1+\epsilon) \frac{k!}{d^{2k}} \sum_j \1den \left\{ \mathcal{C}_\delta (U)>r \right\}\,.
\end{equation}
The sum on the right counts the cardinality $N$ of distinct unitaries with $\delta$-unitary complexity at least $r+1$, while the l.h.s.\ may be lower-bounded by \eqref{eq:probability_bound_channels}:
\begin{equation}
N \geq \frac{d^{2k}}{k!} \left( \frac{1}{1+\epsilon} - 3d^2 n^{2r} |\mathsf{G}|^r \left( \frac{1024 k^4}{d(1-\delta)^2}\right)^k \right)\,.
\end{equation}

\subsection{Distant and distinct design elements}
We have shown that unitary and state designs contain an exponential number ($\Omega(d^k)$) of distinct high complexity elements. But to really capture the ergodic nature of chaotic evolution over the unitary group, these distinct high complexity elements should be pairwise far apart. Here we address this subtlety and show that unitary and state designs contain an exponential number of distant high complexity unitaries or states.

\subsubsection*{Distant and distinct state design elements}
Consider an element drawn at random from an $\epsilon$-approximate spherical $k$-design $\ket{\psi}$. Eq.~\eqref{eq:probability_bound_states} gives that the probability the state has $\delta$-state complexity less than $r$, $\CC_\delta(\ket{\psi})\leq r$, is bounded to be $O(d^{-k})$ when $r\lesssim kn$. We can also show that the probability an element drawn at random from an $\epsilon$-approximate spherical $k$-design is close to a fixed reference state $\ket{\phi}$ is polynomially suppressed in $k$. Choose $\Delta \in (0,1)$ and combine $\tfrac{1}{2} \| \ketbra{\psi}-\ketbra{\phi}\|_1=\sqrt{1-|\langle \psi,\phi \rangle|^2}$ 
with Markov's inequality to conclude
\begin{align}
\Pr \left[ \tfrac{1}{2} \left\| \ketbra{\psi}-\ketbra{\phi} \right\|_1 \leq 1-\Delta \right] 
&= \Pr \left[ | \langle \psi, \phi \rangle|^2 \geq \Delta^2 \right] 
= \Pr \left[ | \langle \psi, \phi \rangle|^{2k} \geq \Delta^{2k} \right] \nn
&\leq \Delta^{-2k} \mathbb{E}_{\ket{\psi}}\left[ |\langle \psi, \phi \rangle|^{2k}\right]
\leq \frac{1+\epsilon}{\Delta^{2k}}\binom{d+k-1}{k}^{-1}\,.
\end{align}
The last inequality follows from $k$-design moment bound similar to Eq.~\eqref{eq:state-moment-intro}. We refer to
the proof of \autoref{lem:state_weights} for a detailed derivation. Qualitatively, this bound is of order $O(d^{-k})$.
We can now use a union bound to limit the probability of a random $k$-design state to have either low complexity \emph{or} to be close to the reference state,
\begin{align}
& \Pr \left[ \CC_\delta (\ket{\psi}) \leq r \; \cup \; \tfrac{1}{2} \| \ketbra{\psi}-\ketbra{\phi} \|_1 \leq 1- \Delta \right] \nn
&\quad \leq \mathrm{Pr} \left[ \CC_\delta (\ket{\psi}) \leq r \right] + \mathrm{Pr} \left[ \tfrac{1}{2} \| \ketbra{\psi}-\ketbra{\phi} \|_1 \leq 1-\Delta \right] \nn
&\quad \leq 2(1+\epsilon) dn^r |\mathsf{G}|^r \left( \frac{16k^2}{d(1-\delta)^2}\right)^k
+ \frac{1+\epsilon}{\Delta^{2k}}\binom{d+k-1}{k}^{-1}\,.
\end{align}
As long as $r \lesssim nk$, this bound is also of order $O(d^{-k})$ and, in turn, strictly smaller than one.
We know that if the probability of the state having low complexity or being close to our fixed state is strictly less than 1, then there is a nonzero probability of a design element that is of high complexity and is far away from the fixed state. Simply stated: if $\Pr [A\cup B] < 1$ then $\Pr [\bar A \cap \bar B] > 0$. 

We can iterate this procedure to construct a set of high complexity states that are pairwise separated. As long as the probability that the design element is of low complexity or is close to all elements of the set is less than one, then there exists a design element which is of high complexity and far away from all other design elements in the set. To construct the set $\{\ket{\psi_1},\ldots,\ket{\psi_N}\}$, we simply need that
\begin{equation}
\Pr \left[ \CC_\delta( \ket{\psi_N} \leq r \; \bigcup_{i=1}^{N-1} \tfrac{1}{2} \| \ketbra{\psi_N}-\ketbra{\psi_i} \|_1 \leq 1- \Delta \right] <1.
\end{equation}
A union bound then converts this requirement into the following sufficient condition on the set cardinality $N$:
\begin{equation}
N < \Delta^{2k} \binom{d+k-1}{k} \left( \frac{1}{1+\epsilon} - 2 d n^r |\mathsf{G}|^r \left( \frac{16k^2}{d(1-\delta)^2} \right)^k \right)
\end{equation}
For constant $\Delta \in (0,1)$, this threshold is exponential as long as the complexity obeys $r\lesssim k$,
\begin{equation}
N \approx O(d^k) \for \CC_\delta(\ket\psi) \leq r \approx k\,.
\end{equation}
We note the similarity of this bound to the bound derived for the number of distinct design elements. 

\subsubsection*{Distant and distinct unitary design elements}
Now we consider a unitary $U$ drawn from an $\epsilon$-approximate unitary $k$-design $\mathcal{E}$. Eq.~\eqref{eq:probability_bound_channels} bounds the probability of the unitary having $\delta$-unitary complexity less than $r$, $\CC_\delta(U)\leq r$, to be $O(d^{-2k})$ when the complexity is roughly $r\lesssim n k$. 

Randomly chosen $k$-design elements also tend to land far away from any fixed unitary. For some $V \in U(d)$ and $\Delta \in (0,1)$, Markov's inequality implies
\begin{align}
\Pr\left[ |\Tr (U^\dagger V)|^2 \geq d^2 \Delta^2 \right] &= \Pr\left[ |\Tr (U^\dagger V)|^{2k} \geq d^{2k}\Delta^{2k} \right]\nn &\leq \frac{\mathbb{E}_{\mathcal{E}} \big[ |\Tr (U^\dagger V)|^{2k}\big]}{d^{2k}\Delta^{2k}} \leq \frac{1+\epsilon}{\Delta^{2k}} \frac{k!}{d^{2k}}\,,
\end{align}
where the last inequality follows from a $k$-design moment bound. We refer to the proof of \autoref{lem:weights-intro} for a detailed derivation. Next, we apply a trick from the proof of \autoref{lem:unitary_complexity_stronger}:
$|\Tr (U^\dagger V)|^2 \geq d^2 \Delta^2$ is a necessary condition for $\|\mathcal{U}-\mathcal{V}\|_\diamond <1-\Delta$. This allows us to conclude
\begin{equation}
\Pr\big[ \| \mathcal{U} - \mathcal{V} \|_\diamond \leq 1-\Delta \big] \leq (1+\epsilon) \frac{k!}{d^{2k}} \frac{1}{\Delta^{2k}}\,.
\end{equation}
Qualitatively, this is of order $O(d^{-2k})$.

We now have all the ingredients in place to repeat the argument from the state case.
The probability of sampling a unitary that has either low complexity \emph{or} is close to any reference unitary $V$ is
\begin{align}
\Pr\left[ \CC_\delta(U)\leq r \cup \| \mathcal{U} - \mathcal{V} \|_\diamond 
\leq 1-\Delta \right]
&\leq 3 (1+\epsilon) d^2 n^{2r} |\mathsf{G}|^r \left(\frac{1024 k^4}{d(1-\delta)^2}\right)^k + \frac{1+\epsilon}{\Delta^{2k}}\frac{k!}{d^{2k}}\,,
\end{align}
according to a union bound. This is on the order of $O(d^{-2k})<1$ as long as the complexity $r\lesssim n k$. 
By contraposition, this ensures that there exists a design element $U_1$ that has both high complexity \emph{and} is far away from $V$. 
We can use this insight to iteratively construct a set of $N$ high-complexity design unitaries with large pairwise distances. Explicitly, to construct a set of unitaries $\{U_1,\ldots,U_N\}$, we need that 
\begin{equation}
\Pr \left[ \CC_\delta(U_N)\leq r \bigcup_{i=1}^{N-1} \| \mathcal{U}_N - \mathcal{U}_i\|_\diamond\leq 1-\Delta \right] < 1\,.
\end{equation}
A union bound relates this condition to a sufficient upper bound on the set cardinality $N$:
\begin{equation}
N < \Delta^{2k} \frac{d^{2k}}{k!} \left( \frac{1}{1+\epsilon} - 3d^2 n^{2r} |\mathsf{G}|^r \left( \frac{1024k^4}{d(1-\delta)^2}\right)^k \right)\,.
\end{equation}
This threshold is exponential as long as the complexity obeys $r\lesssim k$:
\begin{equation}
N \approx O(d^{2k}) \for \CC_\delta(\ket\psi) \leq r \approx k\,.
\end{equation}

\section{Conceptual background and contributions}

\subsection{Distinguishing states and channels}

This conceptual section will review one fundamental question in probability theory, as well as two 
quantum generalizations. We refer to \cite{Watrous2018,Kueng2019} for details.
The underlying question is: \emph{what is the best strategy to distinguish two (biased) coins based on a single toss?}
More precisely, we consider the following game: there are two identically-looking coins with different biases towards coming up heads when being tossed. These biases are known to the player.
A referee then picks one of these coins uniformly at random and hands it to the player. The player is allowed to perform a single toss. Based on the result she must guess which coin she obtained and wins if this guess was correct.

\subsubsection{Distinguishing classical probability distributions}

Consider two (discrete) $d$-variate random variables. Then, we may represent the associated probability distributions by $d$-dimensional vectors $p,q\in \mathbb{R}^d$ which are entry-wise positive ($p_i,q_i \geq 0$) and whose entries sum up to one.
Likewise, a collection of events $E_1,\ldots,E_m$ can be also  represented by vectors $e_1,\ldots,e_m \in \mathbb{R}^d$ that are entry-wise non-negative 
and obey the following normalization condition: $\sum_{i=1}^m e_i = \vec{1}$. Here, $\vec{1}=(1,\ldots,1)^T \in \mathbb{R}^d$ denotes the all-ones vector.
The probability of observing the event associated with index $i$ is
\begin{equation}
\mathrm{Pr} \left[ i \right] = \langle e_i, p \rangle\,.
\label{eq:probability_rule}
\end{equation}
The properties of probability and event vectors then assure $\mathrm{Pr} \left[ i \right] \geq 0$ and $\sum_{i=1}^m \mathrm{Pr} \left[ i \right] =1$. 
Let us now return to the motivating question: what is the best strategy to distinguish two random variables, characterized by known probability vectors $p$ and $q$ in the single-shot limit? This is a binary question and without loss of generality we can restrict our attention to binary events. Let $e_1$ denote the event that leads us to guess that we observed the first random variable. The complementary event $e_2 = \vec{1}-e_1$ is then fully characterized as well. 
Under the additional assumption that either random variable is handed to us with equal prior probability, the probability of success becomes
\begin{align}
p_{\mathrm{cl}}
&= \frac{1}{2} \mathrm{Pr} \left[ 1| p \right] + \frac{1}{2} \mathrm{Pr} \left[ 2 | q \right] 
= \frac{1}{2} \left( \langle e_1, p \rangle + \langle e_2, q \rangle \right) \nn
&= \frac{1}{2} \left( \langle e_1, p - q \rangle + \langle \vec{1}, q \rangle \right) 
= \frac{1}{2} + \frac{1}{2} \langle e_1, p-q \rangle\,.
\end{align}
This expression may now be optimized over all possible events $e_1$ in order to determine the optimal guessing strategy.
The only constraints on $e_1$ are non-negativity and normalization. Together, they demand $0 \leq e_1 \leq \vec{1}$, where the inequality signs are to be understood component-wise. The resulting optimization problem is a \emph{linear program} \cite{Boyd04,Bar02:Course-Convexity}
\begin{align}
&\textrm{maximize} \quad \tfrac{1}{2} + \langle e_1, p - q \rangle \\
&\textrm{subject to} \quad \vec{1} \geq e_1 \geq 0\,. \nonumber
\end{align}
and can be solved in a computationally tractable way. In fact, this problem is simple enough to solve analytically. The optimal $e_1$ is the indicator function for $p_i \geq q_i$, i.e.\ $e_i = \1den \left\{ p_i \geq q_i \right\}$. This is the \emph{maximum likelihood estimator} from statistics. Opt for the distribution that is most likely to produce the outcome that has been observed. This choice achieves an optimal success probability of
\begin{equation}
p_{\mathrm{cl}}^\sharp = \frac{1}{2} + \frac{1}{4} \| p -q \|_{\ell_1}\,.
\end{equation}
Note that a success probability of $1/2$ can be trivially achieved by mere guessing. The remaining factor (multiplied by two)
\begin{equation}
\beta^\sharp_{\mathrm{cl}} = \frac{1}{2} \| p - q \|_{\ell_1}= \frac{1}{2} \sum_{i=1}^d \left| p_i - q_i \right|
\end{equation}
is called the \emph{bias} and corresponds to the \emph{total variational distance} between $p$ and $q$.

\subsubsection{Distinguishing quantum states} \label{sub:distinguishing_states}

It is useful to think of quantum states $\rho$ as matrix generalizations of probability vectors. 
Similarly, (POVM) measurements with $m$ outcomes are characterized by a collection of psd matrices $\left\{M_i \right\}_{i=1}^m \in \mathbb{H}_d$ that sum up to the identity matrix $\mathbb{I}$. Born's rule states that the probability of observing certain outcomes is
\begin{equation}
\mathrm{Pr} \left[ i \right] = \Tr \left( M_i \rho \right)\quad \textrm{for all} \quad 1 \leq i \leq m.
\label{eq:born}
\end{equation}
This may be viewed as a non-commutative analogue of the classical probability rule in Eq.~\eqref{eq:probability_rule}. 
One may also adapt the distinguishability game to the quantum setting: what is the probability of correctly distinguishing two quantum states $\rho,\sigma$ by performing a single measurement? 
Once more, this is a binary question. We can without loss restrict attention to 2-outcome measurements: $M_1$ and $M_2 = \mathbb{I} - M_1$. We associate the first outcome with opting for $\rho$ while the second outcome flags $\sigma$. Similar to the classical case, the probability of success is
\begin{equation}
p_{\mathrm{qs}} = \frac{1}{2} + \frac{1}{2} \left( M_1, \rho - \sigma \right)
\label{eq:prob_success}
\end{equation}
which corresponds to a bias of $\beta_{\mathrm{qs}} = \left( M_1, \rho- \sigma\right)$.
We may now optimize over all possible measurements $M_1$ to obtain the best bias possible:
\begin{align}
\beta^\sharp_{\mathrm{qs}} = \;&\textrm{maximize} \quad \left( M_1, \rho - \sigma \right) \\
&\textrm{subject to} \quad \mathbb{I} \succeq M_1 \succeq 0\,. \nonumber
\end{align}
The constraint denotes the positive semidefinite order ($A \succeq B$ if and only if $A-B$ is positive semidefinite). This is a semidefinite program \cite{Boyd04,Bar02:Course-Convexity} that is simple enough to solve analytically.
The optimal measurement $M_1$ corresponds to the orthogonal projection onto the positive range of $\rho - \sigma$. The associated optimal bias is
\begin{equation}
\beta^\sharp_{\mathrm{qs}} = \frac{1}{2}\| \rho - \sigma \|_1\,,
\end{equation}
which is the \emph{trace distance} of the density matrices $\rho$ and $\sigma$. 
This result is known as \emph{Holevo-Helstrom Theorem} \cite{Holevo1973,Helstrom1976}. 

\begin{example}
Choose $\rho = | \psi \rangle \! \langle \psi|$ and $\sigma = \rho_0 = \tfrac{1}{d}\mathbb{I}$. Then, the (unique) optimal measurement is $M_1 = | \psi \rangle \! \langle \psi|$ and achieves a bias of
\begin{equation}
\beta^\sharp_{\mathrm{qs}}=\frac{1}{2}\| | \psi \rangle \! \langle \psi|  - \rho_0 \|_1 = 1-\frac{1}{d}\,.
\label{eq:trace_distance}
\end{equation}
\end{example}

\subsubsection{Distinguishing quantum channels} \label{sub:distinguishing_channels} 

\emph{Quantum channels} decribe evolutions of quantum mechanical systems. They are linear maps $\mathcal{A}: \mathbb{H}_d \to \mathbb{H}_{d'}$ that map density operators to density operators of potentially different dimension $d'$.

Suppose that we wish to distinguish two channels, say $\mathcal{A}$ and $\mathcal{B}$ based on a single channel use. I.e.\ input a concrete quantum state and perform a measurement on the outcome state.
This indicates more freedom to maximize the probability of correct distinction by opitmizing over potential input states and measurements of the channel output. 
The laws of quantum mechanics allow for further improving this strategy. It is possible to entangle the input state with a quantum memory: $\rho_{\mathrm{in}} \in \mathbb{H}_d \otimes \mathbb{H}_d$. We then apply the channel to the first quantum system, while the second one is left unchanged in the memory. A final two-outcome measurement $M_1 \in \mathbb{H}_{d'} \otimes \mathbb{H}_d$  on both output and memory state potentially reveals additional information. 
The outcome state depends on the channel in question. A priori there are two possibilities. Either $\rho_{\mathrm{out}}= \mathcal{A} \otimes \mathcal{I}(\rho_{\mathrm{in}})$, or $\rho_{\mathrm{out}}=\mathcal{B} \otimes \mathcal{I}(\rho_{\mathrm{in}})$. Here, $\mathcal{I}(X) = X$ denotes the identity channel acting trivially on the memory.
The probability of correctly distinguishing these states -- and thus the underlying channels -- with a singe measurement $M_1 \in \mathbb{H}_{d'} \otimes \mathbb{H}_d$ becomes
\begin{equation}
p_{\mathrm{qc}} = \frac{1}{2} + \Tr\big( M_1 \left( \mathcal{A} \otimes \mathcal{I} (\rho_{\mathrm{in}}) - \mathcal{B} \otimes \mathcal{I} (\rho_{\mathrm{in}}) \right) \big)\,.
\end{equation}
We may now optimize over all degrees of freedom to maximize the value of $p_{\mathrm{qc}}$. Optimizing the measurement $M_1$ results in a bias that is proportional to the trace distance of the outcome states.
Because of convexity, optimization over potential input states can without loss of generality be restricted to pure states:
\begin{equation}
\beta^\sharp_{\mathrm{qc}} = \frac{1}{2}\max_{| \psi \rangle \! \langle \psi|} \big\| \mathcal{A} \otimes \mathcal{I}(| \psi \rangle \! \langle \psi|) - \mathcal{B} \otimes \mathcal{I}(| \psi \rangle \! \langle \psi|) \big\|_1\,.
\label{eq:diamond_distance}
\end{equation}
This optimal bias is called the \emph{diamond distance} between channels $\mathcal{A}$ and $\mathcal{B}$ \cite{Kitaev1997}. 
This metric is more complicated than their vector and matrix counterparts and does highlight genuine quantum advantages.
It can be difficult to compute it analytically,  but does admit a computationally tractable reformulation (SDP) \cite{Watrous2009,Aroya2010,Watrous2013}. 

\begin{example}
Consider a unitary channel $\mathcal{U}(\rho) = U \rho U^\dagger \in \mathbb{H}_d$ and the completely depolarizing channel $\mathcal{D} (\rho) = \frac{\Tr(\rho)}{d} \mathbb{I} \in \mathbb{H}_d$.
Then,
\begin{equation}
\frac{1}{2} \left\| \mathcal{U}-\mathcal{D} \right\|_\diamond = 1 - \frac{1}{d^2} \label{eq:unitary_depolarizing}
\end{equation}
and optimal strategies are based on maximally entangling the input with the memory: Let $| \Omega \rangle = \frac{1}{\sqrt{d}} \sum_{i=1}^d|i \rangle \otimes |i \rangle \in \mathbb{C}^d \otimes \mathbb{C}^d$ be the maximally entangled (Bell) state. Set $\rho_{\mathrm{in}} = \ketbra{\Omega}$ and measure $M_1 = (U^\dagger \otimes \mathbb{I}) \ketbra{\Omega} (U \otimes \mathbb{I})$.
\end{example}

It is easy to check that this strategy achieves the diamond distance in Eq.~\eqref{eq:unitary_depolarizing}. Proving optimality is less trivial. For instance, this claim follows from relating the diamond distance to another norm that is easier to compute. We refer to \cite[Theorem 7]{kliesch_improving_2016} and \cite{michel_comments_2018} for details.

\subsection{Conceptual contributions}

\subsubsection{Cornering ``easy'' unitary transformations}

Fix $d=q^n$.
The evolution of a closed, $d$-dimensional quantum mechanical system is unitary: 
$\mathcal{U}(\rho) = U \rho U^\dagger$ with $U \in U(d)$.
While evolutions may represent natural processes, they can also be engineered to perform certain tasks, such as quantum computing. 
Scalability of quantum computing hinges on the important observation that complicated evolutions (quantum gate architectures) can be decomposed into sequences of simple building blocks. 
A universal gate set $\mathsf{G} \subset U(q^2)$ acting on two (neighboring) qudits forms such a basic set of building blocks. 
For technical reasons, we shall assume that $\mathsf{G}$ contains the identity (doing nothing), as well as inverses: $g \in \mathsf{G}$ implies $g^\dagger \in \mathsf{G}$. 
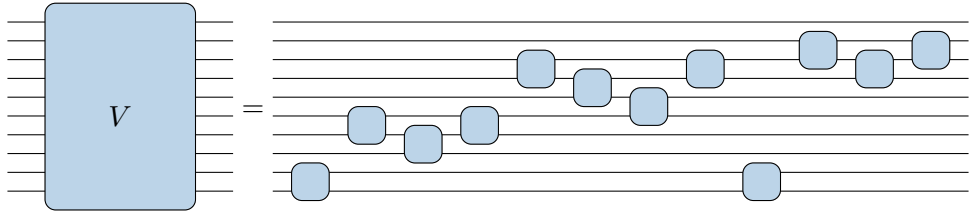
\begin{figure}
\begin{center}
\begin{equation*}
\begin{tikzpicture}[baseline, scale=0.5]
 \foreach \y in {-2,-1.5,...,2.5}
	\draw (-3,\y) -- (3,\y);
\draw[rounded corners, fill=white] (-2,-2.5) rectangle (2,3);
\draw[rounded corners, fill=Blue,opacity=0.3] (-2,-2.5) rectangle (2,3);
\node at (0,0) {$V$};
\end{tikzpicture}
=
\begin{tikzpicture}[baseline,scale=0.5]
 \foreach \y in {-2,-1.5,...,2.5}
	\draw (0,\y) -- (18.5,\y);
\foreach \x / \y in {1/9, 2.5/6, 4/7, 5.5/6, 7/3, 8.5/4, 10/5, 11.5/3, 13/9,14.5/2,16/3,17.5/2} 
{
\draw[fill=white,rounded corners] (\x-0.5, 3.25-0.5*\y) rectangle (\x+0.5,2.25-0.5*\y);
\draw[fill=Blue, opacity=0.3, rounded corners] (\x-0.5, 3.25-0.5*\y) rectangle (\x+0.5,2.25-0.5*\y);
}
\end{tikzpicture}
\end{equation*}
\end{center}
\caption{\emph{Illustration of elementary gate decompositions.} 
A unitary $V$ on $n=10$ qudits is comprised of $12$ geometrically local 2-qudit gates at random positions, i.e.\ $\mathrm{size}(V)=12$.}
\label{fig:circuit-decomposition}
\end{figure}

Universality then means that any unitary $U \in U(d)$ may be accurately approximated by a finite sequence of $r$ unitaries chosen from $\mathsf{G}$. We refer to \autoref{fig:circuit-decomposition} for an illustrative example.
Such decompositions into sequences of elementary gates provide us with a notion of simplicity. Intuitively, a quantum cicuit $V$ is simple if it may be generated by a $\mathsf{G}$-local circuit of short size. In contrast do depth, size counts the total number of elementary gates in a circuit.
For $r \in \mathbb{N}$ we define
\begin{equation}
\mathsf{G}_r := \left\{ V \in U(d): \textrm{$V$ is generated by a $\mathsf{G}$-local circuit of size $\leq r$} \right\}\,.
\end{equation}
We set $\mathsf{G}_0 = \left\{ \mathbb{I} \right\}$ and the following inclusion relation follows from $\mathbb{I} \in \mathsf{G}$:
\begin{equation}
\mathsf{G}_0 \subseteq \mathsf{G}_1 \subseteq \cdots \subseteq \mathsf{G}_r\,.
\end{equation}
The cardinality of $\mathsf{G}_r$ may be bounded by a simple counting argument:
\begin{equation}
\left| \mathsf{G}_r \right| \leq \left( n |\mathsf{G}| \right)^r = \log_q (d)^r |\mathsf{G}|^r\,.
\end{equation}
The fact that $\mathsf{G}$ is a universal gate set ensures that $\mathsf{G}_r$ becomes dense in $U(d)$ provided that $r \to \infty$.
A priori $\mathsf{G}_r$ depends on the particular choice of universal gate set $\mathsf{G}$. However, the Solayev-Kitaev theorem also asserts that other universal gate sets can be accurately compiled at the cost of a constant overhead only \cite{dawson_kitaev_2005}.

\subsubsection{Cornering ``easy'' measurements} \label{sub:easy_measurement}

The conceptual question underlying our definition of complexity is binary. Are we facing a pure state (unitary channel), or a maximally mixed state (depolarizing channel)? 
This allows us to restrict attention to two-outcome measurements, where we associate one outcome with each possibility.

Two-outcome measurements are described by a single matrix: $(M,\mathbb{I}-M)$ that obeys $\mathbb{I} \succeq  M \succeq 0$.
Measuring a quantum state $\rho \in \mathbb{H}_d$ results in two potential outcomes, say ``yes'' and ''no''. The probability of observing either is characterized by Born's rule \eqref{eq:born}:
\begin{equation}
\mathrm{Pr} \left[\textrm{``yes''} \right] = \Tr \left( M \rho \right) \and
\mathrm{Pr} \left[ \textrm{``no''} \right] = \Tr \left( (\mathbb{I}-M) \rho \right)=
1 - \mathrm{Pr} \left[ \textrm{``yes''} \right]\,.
\end{equation}
A \emph{projective} two-outcome measurement is one for which $M$ is an orthogonal projection:
\begin{equation}
M = V P_l V \with P_l = \sum_{i=1}^l |i \rangle \! \langle i|,\quad \textrm{and} \quad V \in U(d)\,.
\label{eq:projective_simple}
\end{equation}
Here $l \in \left[ d \right]$ characterizes the rank of the measurement $M$ and $V$ is a unitary basis change to the eigenbasis of $M$.
\emph{Naimark's theorem}, see e.g.\ \cite{Paulsen2002,Watrous2018}, provides a powerful connection between arbitrary two-outcome measurements $M$ and projective measurements of the form \eqref{eq:projective_simple}. 
Every two-outcome measurement on $\rho \in \mathbb{H}_d$ corresponds to a projective measurement on $\rho \otimes |a \rangle \! \langle a| \in \mathbb{H}_d \otimes \mathbb{H}_2$, where $|a \rangle \! \langle a| \in \mathbb{H}_2$ is an ancilla system prepared in a pure state $| a \rangle \in \mathbb{C}^2$. Pictorially (see Sec.~\ref{sub:wiring} for an introduction of wiring diagrams),

\begin{equation}
\begin{tikzpicture}[scale=0.55, baseline=-1mm]
\draw[rounded corners] (-0.5,-0.5) rectangle (0.5,0.5);
\node at (0,0) {$M$};
\draw[thick] (-1,0) -- (-0.5,0);
\draw[thick] (0.5,0) -- (1,0);
\end{tikzpicture}
=
\begin{tikzpicture}[scale=0.55,baseline=-4mm]
\draw[rounded corners] (-0.5,-1.5) rectangle (0.5,0.5);
\node at (0,-1) {$P_l$};
\draw[thick] (-2,0) -- (-0.5,0);
\draw[thick] (0.5,0) -- (2,0);
\draw[thick] (-1,-1) -- (-0.5,-1);
\draw[thick] (0.5,-1) -- (1,-1);
\draw[rounded corners] (-1.5,-1.5) rectangle (-1,-0.5);
\node at (-1.25,-1) {$a$};
\draw[rounded corners] (1,-1.5) rectangle (1.5,-0.5);
\node at (1.25,-1) {$a$};
\end{tikzpicture}
\end{equation}

\ni Based on this reformulation of general 2-outcome measurements, we model limited resources in the following way:
\begin{enumerate}
\item The ancilla state $|a \rangle \in \mathbb{C}^2$ is corresponds to a (fixed) simple state, e.g. $|a \rangle = |0 \rangle$.
\item The unitary $V \in U(2d)$ must be feasible to implement. More concretely we assume that it is comprised of at most $r$ 2-qubit gates chosen from a (fixed) universal gate set $\mathsf{G} \subset U(q^2)$.
\item The projective measurement $P_l = \sum_{i=1}^l |i \rangle \! \langle i|$ is diagonal in the computational basis. 
\end{enumerate}

\ni For fixed $r \in \mathbb{N}$ (circuit size for $V$), this framework defines the following class of measurements:
\begin{equation}
\mathsf{M}_{r} = \left\{ \Tr_2 \left( \mathbb{I} \otimes |a \rangle \! \langle a|\; V P_{l'} V^\dagger\right):\; V \in \mathsf{G}_r,\; l' \in \left[2d \right] \right\} \subset \mathbb{H}_d\,.
\end{equation}
Here, $\Tr_2: \mathbb{H}_d \otimes \mathbb{H}_2 \to \mathbb{H}_d$ denotes the partial trace.
By construction, this set is finite and obeys
\begin{equation}
\left| \mathsf{M}_{r} \right| \leq  2d \left|\mathsf{G}_r \right| \leq 2d \log_q (d)^r |\mathsf{G}|^r\,.
\label{eq:measurement_cardinality}
\end{equation}
For the class of measurements applied to channel (and memory) outputs, we slightly modify this definition. We include a single Bell measurement $| \Omega \rangle \! \langle \Omega| \in \mathbb{H}_d^{\otimes 2}$
with $| \Omega \rangle = \frac{1}{\sqrt{d}} \sum_{i=1}^d |i \rangle \otimes |i \rangle$ in the definition:
\begin{equation*}
\mathsf{M}_r = \left\{ \Tr_2 \left( \mathbb{I} \otimes |a \rangle \! \langle a|\; V P_{l'} V^\dagger\right):\; V \in \mathsf{G}_r,\; l' \in \left[2d^2 \right] \right\} 
\cup \left\{ V | \Omega \rangle \! \langle \Omega|V^\dagger:\; V \in \mathsf{G}_r \right\} \subset \mathbb{H}_d^{\otimes 2}\,.
\end{equation*}
This modification simplifies exposition and is comparatively benign. 
A simple counting argument reveals
\begin{equation}
|\mathsf{M}_r| \leq (2d^2+1) | \mathsf{G}_r| \leq (2d^2+1) d^2 2 \log_q (d)^r |\mathsf{G}|^r\,.
\end{equation}

\section{Technical background and contributions}

\subsection{Notation and basic facts from matrix analysis}

Endow the vector space $\mathbb{C}^d$ with the standard inner product $\braket{x|y} $. 
A pure quantum state is a vector $\psi \in \mathbb{C}^d$ normalized to (Euclidean) unit length, i.e.\ $\langle \psi,\psi \rangle =1$. We succinctly denote this by identifying normalized vectors with kets:
\begin{equation}
\ket{\psi}  \quad\textrm{denotes}\quad \psi \in \mathbb{C}^d \quad \textrm{with} \quad  \langle \psi| \psi \rangle = 1\,.
\end{equation}
Let $\mathbb{H}_d$ denote the space of Hermitian $d \times d$ matrices. 
This is a real-valued subspace of the space of all (complex-valued) $d \times d$ matrices $\mathbb{M}_d$.
Fix an orthonormal basis $|1 \rangle,\dots,|d \rangle$ of $\mathbb{C}^d$. Then, the trace of a matrix $X$ is $\Tr(X) = \sum_{i=1}^d \langle i| X |i \rangle$.
The trace is cyclic, i.e.\ $\Tr(XY) = \Tr(YX)$ and forms the basis for defining the Schatten-$p$ norms. In particular,
\begin{align*}
&\| X \|_1 = \Tr(|X|),\; |X| = \sqrt{X^2} & \textrm{(trace norm)} \\
&\| X \|_2 = \sqrt{\Tr(X^2)} & \textrm{(Frobenius norm)} \\
&\| X \|_\infty = \max_{\ket y} \left| \langle y | X |y \rangle \right| & \textrm{(operator norm)}\,.
\end{align*}
Schatten-norms obey the following order relations:
\begin{equation}
\| X \|_\infty \leq \| X \|_2 \leq \| X \|_1 \quad \textrm{and} \quad \| X \|_1 \leq \sqrt{d} \| X \|_2 \leq d \| X \|_\infty\quad \textrm{for all} \quad X \in \mathbb{H}_d\,.
\end{equation}
A variant of H\"older's inequality applies to traces of inner products, see e.g.\ \cite[Ex. IV.2.12]{Bhatia1997}:
\begin{equation}
\left|\Tr(XY) \right|\leq \| X \|_1 \| Y \|_\infty \quad \textrm{for all} \quad X,Y \in \mathbb{H}_d\,.
\label{eq:hoelder}
\end{equation}
The trace corresponds to a full index contraction. Partial contractions are possible for tensor products and \emph{partial traces} are concrete examples. For $X,Y \in \mathbb{H}_d$ define 
\begin{equation}
\Tr_1 \left( X \otimes Y \right) = \Tr(X) Y \quad \textrm{and} \quad \Tr_2 \left( X \otimes Y \right) = \Tr(Y) X\,,
\end{equation}
and extend this definition linearly to the tensor product $\mathbb{H}_d^{\otimes 2} \simeq \mathbb{H}_{d^2}$. This definition naturally extends to tensor products of higher order. 
The following tight bound connects partial traces and operator norms:
\begin{equation}
\max \left\{ \left\| \Tr_1 (X) \right\|_\infty, \left\| \Tr_2 (X) \right\|_\infty \right\} \leq d \| X \|_\infty \quad \textrm{for all} \quad X \in \mathbb{H}_d^{\otimes 2}\,.
\label{eq:partial_operator_norm}
\end{equation}
A matrix $X \in \mathbb{H}_d$ is positive semidefinite (psd) if $\langle y|X|y \rangle \geq 0$ for all $y \in \mathbb{C}^d$. We denote this feature by $X \succeq 0$. Positive semidefiniteness is preserved under partial traces:
\begin{equation}
X \in \mathbb{H}_d^{\otimes 2}, \;X \succeq 0\quad \textrm{implies} \quad \Tr_1 (X) \succeq 0,\; \Tr_2 (X) \succeq 0\,.
\end{equation}
The trace norm of psd matrices is particularly simple: $\|X \|_1 = \Tr(X)$ whenever $X \succeq 0$.

\subsection{Convex geometry and optimization}

The main technical contributions of this paper are based on bounds that follow from a fundamental argument in convex optimization. Comprehensive references for convex geometry and optimization include \cite{Roc70:Convex-Analysis,Bar02:Course-Convexity}.
A function $f: \mathbb{H}_d \to \mathbb{R}$ is \emph{convex} if
\begin{equation}
f\left( \tau X+ (1-\tau) Y \right) \leq \tau f (X) + (1-\tau) f (Y) \quad \textrm{for all} \quad X,Y \in \mathbb{H}_d, \; \tau \in \left[0,1 \right]\,.
\end{equation}
Linear transformations in the argument preserve this feature.
Similarly, a set $\mathcal{K} \subseteq \mathbb{H}_d$ is \emph{convex} if 
\begin{equation}
X,Y \in \mathcal{K} \quad \textrm{imply} \quad \tau X + (1-\tau) Y \in \mathcal{K} \quad \textrm{for all} \quad \tau \in \left[0,1\right]\,.
\end{equation}
Let $\mathcal{K} \subseteq \mathbb{H}_d$ be a convex set. A point $X \in \mathcal{K}$ is an \emph{extreme point} if $Y,Z \in \mathcal{K}$ and $X= \tau Y + (1-\tau) Z$ for some $\tau \in (0,1)$ necessarily imply $Y=Z=X$. 
Extreme points form the boundary of a convex set.

\begin{example}
The set of all quantum states in $\mathbb{H}_d$ is the convex hull (i.e.\ the set of all convex combinations) of pure states:
\begin{equation}
\left\{ \rho \in \mathbb{H}_d:\; \Tr(\rho)=1,\; \rho \succeq 0 \right\} = \mathrm{conv} \left\{ \ketbra{\psi}: \; \ket{\psi} \in \mathbb{C}^d \right\}\,.
\end{equation}
All extreme points are pure states.
\end{example}

\begin{fact}[convex functions achieve their maximum at an extreme point] \label{fact:convex_maximum}
Let $\mathcal{K} \subseteq \mathbb{H}_d$ be a convex set and let $f: \mathcal{K} \to \mathbb{R}$ be a convex function. 
Then, there exists an extreme point $X_\sharp$ of $\mathcal{K}$ such that
\begin{equation}
\max_{X \in \mathcal{K}} f(X) \leq f(X_\sharp)\,.
\end{equation}
\end{fact}

This result justifies the presentation of the diamond distance in Eq.~\eqref{eq:diamond_distance}. The function $X \mapsto \left\| \mathcal{A} \otimes \mathcal{I} (X) - \mathcal{B}\otimes \mathcal{I}(X) \right\|_1$ is convex (norms are convex and the channel acts like a linear transformation of the argument) and pure states are the extreme points of the set of all quantum states. Hence,
\begin{equation}
\max_{\rho} \left\| \mathcal{A} \otimes \mathcal{I} (\rho) - \mathcal{B}\otimes \mathcal{I}(\rho) \right\|_1
= \max_{| \psi \rangle \! \langle \psi|}  \left\| \mathcal{A} \otimes \mathcal{I} (\ketbra{\psi}) - \mathcal{B}\otimes \mathcal{I}(\ketbra{\psi}) \right\|_1\,.
\end{equation}

The following technical result will prove highly valuable for establishing bounds on very general Haar moments.

\begin{lemma} \label{lem:convex_function}
Fix $A \in \mathbb{H}_d$ psd ($A 
\succeq 0)$. Then, the function $h(X) = \Tr \left( XAXA\right)$ is nonnegative and convex for all $X \in \mathbb{H}_d$.
\end{lemma}

\begin{proof}
Apply an eigenvalue decomposition: $A = U (\sum_{i=1}^d \alpha_i |i \rangle \! \langle i|) U^\dagger$. The assumption that $A$ is psd ensures $\alpha_1,\ldots,\alpha_d \geq 0$. Next, fix $X \in \mathbb{H}_d$ arbitrary, set $\tilde{X} = U^\dagger X U$ and compute
\begin{equation}
\Tr(XAXA) = \sum_{i,j=1}^d \alpha_i \alpha_j \left| \langle i| \tilde{X} |j \rangle \right|^2 \geq 0\,.
\end{equation}
This establishes non-negativity of $h(X)$.
For convexity, fix $X,Y \in \mathbb{H}_d$ and $\tau \in \left[0,1\right]$. Set $\bar{\tau} =1-\tau$ and note that $\tau \bar{\tau} = \tau - \tau^2 = \bar{\tau} - \bar{\tau}^2 \geq 0$.
Nonnegativity moreover implies $h(X-Y)\geq 0$ and we can readily deduce convexity:
\begin{align}
h (\tau X + \bar{\tau} Y )
&= \tau^2 \Tr(XAXA) + 2 \tau \bar{\tau} \Tr(XAYA) + \bar{\tau}^2 \Tr(YAYA) \nn
&= \tau h(X) - \tau \bar{\tau} \left( \Tr(XAXA) - 2 \Tr (XAYA) + \Tr(YAYA) \right) + \bar{\tau} h(Y) \nn
&= \tau h(X)- \tau \bar{\tau} h(X-Y) + \bar{\tau} h(Y) \leq \tau h(X) + \bar{\tau}h(Y)\,.
\end{align}
\end{proof}

\subsection{Wiring calculus} \label{sub:wiring}

Wiring diagrams, sometimes also known as tensor network diagrams, provide a graphical way for computing contractions between tensors. 
Here we only provide a brief overview and refer to the recent survey \cite{bridgeman_handwaving_2017} and lecture notes \cite{Kueng2019} for a detailed introduction. The wiring formalism associates a box with every tensor and a line emanating from the box with every index.
Connected lines represent contracted indices. More precisely, we place contravariant indices of a tensor on the left of the box and covariant ones on the right. \autoref{tab:wiring_basics} contains all the essential rules necessary for the scope of this work.
\begin{table}
\centering
\taburulecolor{gray}
\begin{tabu}{|l|c|c|}
\hline
\begin{tikzpicture}[baseline,scale=0.5]
\node at (0,0)
{ket vector};
\end{tikzpicture}
&
 \begin{tikzpicture}[baseline,scale=0.5]
 \node at (0,0)
{$| \psi \rangle \in \mathbb{C}^d$};
\end{tikzpicture}
& 
\begin{tikzpicture}[scale=0.5,baseline]
\draw[white] (0,-0.7) -- (0,0.7);
\draw[rounded corners] (-0.5,-0.5) rectangle (0.5,0.5);
\node at (0,0) {$\psi$};
\draw[thick] (-1,0) -- (-0.5,0);
\end{tikzpicture} \\
\hline \hline 
\begin{tikzpicture}[baseline,scale=0.5]
\node at (0,0)
{bra vector};
\end{tikzpicture}
&
 \begin{tikzpicture}[baseline,scale=0.5]
\node at (0,0) {$\langle \phi | \in \left(\mathbb{C}^d \right)^* \simeq \mathbb{C}^d$};
\end{tikzpicture}
&
\begin{tikzpicture}[scale = 0.5,baseline]
\draw[white] (0,-0.7) -- (0,0.7);
\draw[rounded corners] (-0.5,-0.5) rectangle (0.5,0.5);
\node at (0,0) {$\phi$};
\draw[thick] (0.5,0) -- (1,0);
\end{tikzpicture} \\
\hline \hline
\begin{tikzpicture}[baseline,scale=0.5]
\node at (0,0)
{inner product (contraction)};
\end{tikzpicture}
&
\begin{tikzpicture}[baseline,scale=0.5]
\node at (0,0)
{$\langle \phi| \psi \rangle$};
\end{tikzpicture}

&
\begin{tikzpicture}[scale =0.5,baseline]
\draw[white] (0,-0.7) -- (0,0.7);
\draw[rounded corners] (-1.25,-0.5) rectangle (-0.25,0.5);
\node at (-0.75,0) {$\phi$};
\draw[thick] (-0.25,0) -- (0.25,0);
\draw[rounded corners] (0.25,-0.5) rectangle (1.25,0.5);
\node at (0.75,0) {$\psi$};
\end{tikzpicture} \\
\hline \hline 
\begin{tikzpicture}[baseline,scale=0.5]
\node at (0,0)
{matrix};
\end{tikzpicture}
&
\begin{tikzpicture}[baseline,scale=0.5]
\node at (0,0)
{$A \in \mathbb{M}_d$};
\end{tikzpicture}
 &
\begin{tikzpicture}[scale=0.5,baseline]
\draw[white] (0,-0.7) -- (0,0.7);
\draw[rounded corners] (-0.5,-0.5) rectangle (0.5,0.5);
\node at (0,0) {$A$};
\draw[thick] (0.5,0) -- (1,0);
\draw[thick] (-1,0) -- (-0.5,0);
\end{tikzpicture} \\
\hline \hline 
\begin{tikzpicture}[baseline,scale=0.5]
\node at (0,0)
{matrix product of $A,B \in \mathbb{M}_d$};
\end{tikzpicture}
&
\begin{tikzpicture}[baseline,scale=0.5]
\node at (0,0)
{$AB \in \mathbb{M}_d$};
\end{tikzpicture}
&
\begin{tikzpicture}[scale=0.5,baseline]
\draw[white] (0,-0.7) -- (0,0.7);
\draw[thick] (-1.75,0) -- (-1.25,0);
\draw[rounded corners] (-1.25,-0.5) rectangle (-0.25,0.5);
\node at (-0.75,0) {$A$};
\draw[thick] (-0.25,0) -- (0.25,0);
\draw[rounded corners] (0.25,-0.5) rectangle (1.25,0.5);
\node at (0.75,0) {$B$};
\draw[thick] (1.25,0) -- (1.75,0);
\end{tikzpicture} \\
\hline \hline 
\begin{tikzpicture}[baseline,scale=0.5]
\node at (0,0.25)
{matrix trace (contraction)};
\end{tikzpicture}
&
\begin{tikzpicture}[baseline,scale=0.5]
\node at (0,0.25)
{$\Tr(A) \in \mathbb{C}$};
\end{tikzpicture}
&
\begin{tikzpicture}[scale=0.5,baseline]
\draw[white] (0,-0.7) -- (0,1.2);
\draw[rounded corners] (-0.5,-0.5) rectangle (0.5,0.5);
\node at (0,0) {$A$};
\draw[thick,rounded corners] (0.5,0) -- (1,0) -- (1,1) -- (-1,1) -- (-1,0) -- (-0.5,0);
\end{tikzpicture} \\
\hline \hline 
\begin{tikzpicture}[baseline,scale=0.5]
\node at (0,0)
{tensor product (vectors)};
\end{tikzpicture}
&
\begin{tikzpicture}[baseline,scale=0.5]
\node at (0,0)
{$| \psi \rangle \otimes | \phi \rangle \in (\mathbb{C}^d)^{\otimes 2}$};
\end{tikzpicture}
&
\begin{tikzpicture}[scale = 0.5,baseline]
\draw[white] (0,-1.45) -- (0,1.45);
\draw[rounded corners] (-0.5,0.25) rectangle (0.5,1.25);
\node at (0,0.75) {$\psi$};
\draw[thick] (-1,0.75) -- (-0.5,0.75);
\draw[rounded corners] (-0.5,-1.25) rectangle (0.5,-0.25);
\node at (0,-0.75) {$\phi$};
\draw[thick] (-1,-0.75) -- (-0.5,-0.75);
\end{tikzpicture} \\
\hline \hline 
\begin{tikzpicture}[baseline,scale=0.5]
\node at (0,0)
{tensor product (matrices)};
\end{tikzpicture}
& 
\begin{tikzpicture}[baseline,scale=0.5]
\node at (0,0)
{$A \otimes B \in \mathbb{H}_d^{\otimes 2}$};
\end{tikzpicture}
&
\begin{tikzpicture}[scale = 0.5,baseline]
\draw[white] (0,-1.45) -- (0,1.45);
\draw[rounded corners] (-0.5,0.25) rectangle (0.5,1.25);
\node at (0,0.75) {$A$};
\draw[thick] (-1,0.75) -- (-0.5,0.75);
\draw[thick] (0.5,0.75) -- (1,0.75);
\draw[rounded corners] (-0.5,-1.25) rectangle (0.5,-0.25);
\node at (0,-0.75) {$B$};
\draw[thick] (-1,-0.75) -- (-0.5,-0.75);
\draw[thick] (0.5,-0.75) -- (1,-0.75);
\end{tikzpicture} \\
\hline
\end{tabu}
\caption{\emph{Basic building blocks of wiring calculus.}} \label{tab:wiring_basics}
\end{table}

Importantly lines can be bent at will without changing the value of an equation.\footnote{Technically, this is only true for bending lines an even number of times. A single bend corresponds to transposition, which is basis dependent and not equivalent to conjugation. This subtlety, however, will rarely feature in our arguments.}
For instance, let $\rho =\ketbra{\psi} \in \mathbb{H}_d$ be a pure quantum state and suppose that $M \in\mathbb{H}_d$ is measurement. We can then represent Born's rule pictographically as
\begin{equation}
\Tr \left( M \rho \right)
=
\begin{tikzpicture}[scale=0.7,baseline=-1mm]
\draw[rounded corners] (-1.25,-0.5) rectangle (-0.25,0.5);
\node at (-0.75,0) {$M$};
\draw[thick] (-0.25,0) -- (0.25,0);
\draw[rounded corners] (0.25,-0.5) rectangle (1.25,0.5);
\node at (0.75,0) {$\rho$};
\draw[thick,rounded corners] (1.25,0) -- (1.75,0)--(1.75,1) -- (-1.75,1) -- (-1.75,0) -- (-1.25,0);
\end{tikzpicture}
=
\begin{tikzpicture}[scale=0.7,baseline=-1mm]
\draw[rounded corners] (-1.25,-0.5) rectangle (-0.25,0.5);
\node at (-0.75,0) {$M$};
\draw[thick] (-0.25,0) -- (0.25,0);
\draw[rounded corners] (0.25,-0.5) rectangle (0.75,0.5);
\node at (0.5,0) {$\psi$};
\draw[rounded corners] (1,-0.5) rectangle (1.5,0.5);
\node at (1.25,0) {$\psi$};
\draw[thick, rounded corners] (1.5,0) -- (2,0) -- (2,1) -- (-1.75,1) -- (-1.75,0) -- (-1.25,0);
\end{tikzpicture}
=
\begin{tikzpicture}[scale=0.7,baseline=-1mm]
\draw[rounded corners] (-0.5,-0.5) rectangle (0.5,0.5);
\node at (0,0) {$M$};
\draw[thick] (0.5,0) -- (1,0);
\draw[thick] (-1,0) -- (-0.5,0);
\draw[rounded corners] (1,-0.5) rectangle (1.5,0.5);
\node at (1.25,0) {$\psi$};
\draw[rounded corners] (-1.5,-0.5) rectangle (-1,0.5);
\node at (-1.25,0) {$\psi$};
\end{tikzpicture}
= \braket{\psi |M | \psi}\,.
\end{equation}

\ni Partial traces also assume a simple form. For $X \in \mathbb{H}_d \otimes \mathbb{H}_d$
\begin{equation}
\Tr_1 (X) = 
\begin{tikzpicture}[baseline=-1mm,scale=0.55]
\draw[thick, rounded corners] (0,0.5) -- (1.15,0.5) -- (1.15,1.25) -- (-1.15,1.25) -- (-1.15,0.5) -- (0,0.5);
\draw[thick] (-1.15,-0.5) -- (1.15,-0.5);
\draw[rounded corners,fill=white] (-0.75,-0.9) rectangle (0.75,0.9);
\node at (0,0) {$X$};
\end{tikzpicture}
\quad \textrm{and} \quad 
\Tr_2 (X) = 
\begin{tikzpicture}[baseline=-1mm,scale=0.55]
\draw[thick, rounded corners] (0,-0.5) -- (1.15,-0.5) -- (1.15,-1.25) -- (-1.15,-1.25) -- (-1.15,-0.5) -- (0,-0.5);
\draw[thick] (-1.15,0.5) -- (1.15,0.5);
\draw[rounded corners,fill=white] (-0.75,-0.9) rectangle (0.75,0.9);
\node at (0,0) {$X$};
\end{tikzpicture}\,.
\end{equation}
Wiring calculus is exceptionally well suited to keep track of \emph{flip operators}. Define $\mathbb{F} |i \rangle | \otimes |j \rangle = |j \rangle \otimes |i \rangle$ via its action on computational basis elements and extend this definition linearly to $\mathbb{C}^d \otimes \mathbb{C}^d$. Then,
\begin{equation}
\begin{tikzpicture}[baseline=-1mm,scale=0.55]
\draw[thick] (-1.15,0.5) -- (1.15,0.5);
\draw[thick] (-1.15,-0.5) -- (1.15,-0.5);
\draw[rounded corners,fill=white] (-0.75,-0.9) rectangle (0.75,0.9);
\node at (0,0) {$\mathbb{F}$};
\end{tikzpicture}
=
\begin{tikzpicture}[baseline=-1mm,scale=0.55]
\draw[thick, rounded corners] (-3,0.5) -- (-2,0.5) -- (-1.5,-0.5) --(-0.5,-0.5);
\draw[thick, rounded corners] (-3,-0.5) -- (-2,-0.5) -- (-1.5,0.5) -- (-0.5,0.5);
\end{tikzpicture} \,.
\end{equation}
\emph{Vectorization} is a linear map  vec: $\mathbb{M}_d \to \mathbb{C}^d \otimes \mathbb{C}^d$ defined by its action on computational basis elements
\begin{equation}
|\mathrm{vec}\left(|i \rangle \! \langle j| \right) \rangle := |i \rangle \otimes |j \rangle\,,
\end{equation}
and linearly extended to all of $\mathbb{M}_d$. 
In wiring calculus, $|\phi \rangle = |\mathrm{vec}(\Phi) \rangle$ corresponds to bending the right (covariant) index of a matrix $A$ to the left (into a contravariant one):
\begin{equation}
\begin{tikzpicture}[baseline=-1mm,scale=0.55]
\draw[rounded corners] (0,-0.9) rectangle (0.5,0.9);
\node at (0.25,0) {$\phi$};
\draw[thick] (-0.5,0.5) -- (0,0.5);
\draw[thick] (-0.5,-0.5) -- (0,-0.5);
\end{tikzpicture}
=
\begin{tikzpicture}[baseline=-1mm,scale=0.55]
\draw[thick,rounded corners] (-0.5,0.5) -- (1.5,0.5) -- (1.5,-0.5) -- (-0.5,-0.5);
\draw[rounded corners, fill=white] (0.05,-0.95) rectangle (0.95,-0.05);
\node at (0.5,-0.5) {$\Phi$};
\end{tikzpicture}
\quad \textrm{and} \quad
\begin{tikzpicture}[baseline=-1mm,scale=0.55]
\draw[rounded corners] (0,-0.9) rectangle (0.5,0.9);
\node at (0.25,0) {$\phi$};
\draw[thick] (0.5,0.5) -- (1,0.5);
\draw[thick] (0.5,-0.5) -- (1,-0.5);
\end{tikzpicture}
=
\begin{tikzpicture}[baseline=-1mm,scale=0.55]
\draw[thick,rounded corners] (0.5,0.5) -- (-1.5,0.5) -- (-1.5,-0.5) -- (0.5,-0.5);
\draw[rounded corners, fill=white] (-0.95,-0.95) rectangle (-0.05,-0.05);
\node at (-0.5,-0.5) {$\Phi^\dagger$};
\end{tikzpicture}~.
\end{equation}
It is easy to see that vectorization is an isometry:
\begin{equation}
\langle \phi | \phi \rangle
=
\begin{tikzpicture}[baseline=-1mm,scale=0.55]
\draw[rounded corners] (-0.5,-0.9) rectangle (0,0.9);
\node at (-0.25,0) {$\phi$};
\draw[rounded corners] (1,-0.9) rectangle (1.5,0.9);
\node at (1.25,0) {$\phi$};
\draw[thick] (0,0.5) -- (1,0.5);
\draw[thick] (0,-0.5) -- (1,-0.5);
\end{tikzpicture}
=
\begin{tikzpicture}[baseline=-1mm,scale=0.55]
\draw[thick, rounded corners] (0,0.5) -- (1.75,0.5) -- (1.75,-0.5) -- (-1.75,-0.5) -- (-1.75,0.5) -- (0,0.5);
\draw[rounded corners, fill=white] (-1.2,-0.95) rectangle (-0.3,-0.05);
\node at (-0.75,-0.5) {$\Phi^\dagger$};
\draw[rounded corners, fill=white] (0.3,-0.95) rectangle (1.2,-0.05);
\node at (0.74,-0.5) {$\Phi$};
\end{tikzpicture}
= \Tr \left( \Phi^\dagger \Phi \right) = \| \Phi \|_2^2\,.
\end{equation}

\subsection{Random unitaries and \ktitle-designs} \label{sub:designs}
Here we introduce a few essential concepts from quantum information theory, including a discussion of random unitaries and the notion of a design. First, recall that the Haar measure is the unique left/right invariant measure on the unitary group $U(d)$. We will often be interested in moments of the Haar ensemble. Consider an operator $X$ acting on the $k$-fold Hilbert space $(\C^d)^{\otimes k}$, the $k$-fold channel, or $k$-fold twirl, of the operator with respect to the Haar measure on the unitary group is
\begin{equation}
\CT^{(k)}_{U}(X) = \int dU\, U^{\otimes k} (X) U^\dagger{}^{\otimes k}\,.
\end{equation}
Similarly, we can average an operator over an ensemble of unitaries $\CE = \{p_i,U_i\}$, a weighted subset of the full unitary group. The $k$-fold channel with respect to $\CE$ is
\begin{equation}
\CT^{(k)}_{\CE}(X) = \sum_i p_i U_i^{\otimes k} (X) U_i^\dagger{}^{\otimes k}\,,
\end{equation}
here written for a discrete ensemble, but such an ensemble might be discrete or continuous. 

\paragraph{Unitary $k$-designs}
We will often be interested in how well an average over an ensemble captures an average over the full unitary group, \ie how random the ensemble is with respect to the Haar measure on $U(d)$. 
A {\it unitary $k$-design} is an ensemble of unitaries $\CE=\{ p_i,U_i\}$, for which the $k$-fold twirl equals its Haar-random counterpart:
\begin{equation}
\CT^{(k)}_{\CE}(X) = \CT^{(k)}_{U}(X)\quad \textrm{for all} \quad X \in \mathbb{H}_d^{\otimes k}.
\end{equation}
This means that the ensemble $\CE$ exactly captures the first $k$ moments of the Haar ensemble. Unitary operator bases, such as the $n$-qubit Pauli group, form an exact 1-design. But very little is known about the construction of exact designs for higher $k$, with the notable exception of $k=3$ and the $n$-qubit Clifford group \cite{zhu_multiqubit_2017,webb_clifford_2015,Kueng15}. We will return to this point when discussing approximate designs. 

\paragraph{Schur-Weyl duality}
Many of the important analytic expressions for Haar averages rely on {\it Schur-Weyl duality} \cite{Fulton91, christandl_phd_2006}, a deep connection between irreducible representations (irreps) of the unitary group $U(d)$ and the symmetric group $S_k$. First, when thinking about $k$-fold Hilbert spaces, there is a useful set of operators that acts on this space, namely permutations of the $k$ copies. A permutation operator $P_\sigma$ acts on the computational basis of $(\C^d)^{\otimes k}$ as
\begin{equation}
P_\sigma \ket{i_1,\ldots,i_k} = \ket{i_{\sigma^{-1}(1)},\ldots,i_{\sigma^{-1}(k)}}\,.
\end{equation}
This action can be extended linearly to all of $(\mathbb{C}^d)^{\otimes k}$.
Schur-Weyl duality is the statement that an operator acting on $(\C^d)^{\otimes k}$ commutes with all $k$-fold unitaries $U^{\otimes k}$ if and only if it is a linear combination of permutation operators
\begin{equation}
U^{\otimes k}X U^\dagger{}^{\otimes k} = X \quad\longleftrightarrow\quad X = \sum_{\sigma\in S_k} c_\sigma P_\sigma\,.
\end{equation}
Many of the exact expressions for Haar moments and random unitary averages in the following subsection follow directly from this powerful result. 

\subsection{Haar-integration over the unitary group}
We now introduce the general formalism for integrating arbitrary moments of random unitaries over the full unitary group with respect to the Haar measure, often referred to as Weingarten calculus. The exact expression \cite{Collins02,Collins04} for integrating the $k$-th moment of $U(d)$ is
\begin{equation}
\int dU\, U_{i_1 j_1}\ldots U_{i_k j_k} U^\dagger_{\ell_1 m_1} \ldots U^\dagger_{\ell_k m_k} = \sum_{\sigma,\tau \in S_k} \delta_\sigma(\vec \imath\,|\vec m) \delta_\tau(\vec \jmath\,|\vec \ell\,) \Wg(\sigma^{-1}\tau, d)\,,
\label{eq:Uint}
\end{equation}
where we sum over elements of the permutation group $S_k$ and define a contraction of indices with respect to a permutation $\sigma \in S_k$ as
\begin{equation}
\delta_\sigma(\vec \imath\,|\vec \jmath\,) := \prod_{s=1}^k \delta_{i_s j_{\sigma(s)}} = \delta_{i_1 j_{\sigma(1)}} \ldots \delta_{i_k j_{\sigma(k)}}\,.
\end{equation}
Mixed moments of $U$ and $U^\dagger$, \ie averages of $U^{\otimes k}\otimes U^\dagger{}^{\otimes k'}$ with $k\neq k'$, vanish identically. 

It will often be convenient to interpret the index contraction $\delta_\sigma(\vec \imath\,|\vec \jmath\,)$ as a permutation operator acting on the computational basis of the $k$-fold space,
\begin{equation}
\delta_\sigma(\vec \imath\,|\vec \jmath\,) = P_\sigma\,.
\end{equation}
For instance, two examples of contractions for $k=4$ are
\begin{equation}
\delta_{\{2,1,4,3\}} (\vec \imath\,|\vec \jmath\,) = 
\begin{tikzpicture}[baseline=-1mm,scale=0.55]
\draw[thick, rounded corners] (-1,1.5) -- (-0.5,1.5) -- (0.5,0.5) -- (1,0.5);
\draw[thick, rounded corners] (-1,0.5) -- (-0.5,0.5) -- (0.5,1.5) -- (1,1.5);
\draw[thick, rounded corners] (-1,-0.5) -- (-0.5,-0.5) -- (0.5,-1.5) -- (1,-1.5);
\draw[thick, rounded corners] (-1,-1.5) -- (-0.5,-1.5) -- (0.5,-0.5) -- (1,-0.5);
\node[anchor=east] at (-0.8,1.5) {$i_1$};
\node[anchor=east] at (-0.8,0.5) {$i_2$};
\node[anchor=east] at (-0.8,-0.5) {$i_3$};
\node[anchor=east] at (-0.8,-1.5) {$i_4$};
\node[anchor=west] at (0.8,1.5) {$j_1$};
\node[anchor=west] at (0.8,0.5) {$j_2$};
\node[anchor=west] at (0.8,-0.5) {$j_3$};
\node[anchor=west] at (0.8,-1.5) {$j_4$};
\end{tikzpicture}
\and 
\delta_{\{2,3,4,1\}} (\vec \imath\,|\vec \jmath\,) = 
\begin{tikzpicture}[baseline=-1mm,scale=0.55]
\draw[thick, rounded corners] (-1,1.5) -- (-0.5,1.5) -- (0.5,0.5) -- (1,0.5);
\draw[thick, rounded corners] (-1,0.5) -- (-0.5,0.5) -- (0.5,-0.5) -- (1,-0.5);
\draw[thick, rounded corners] (-1,-0.5) -- (-0.5,-0.5) -- (0.5,-1.5) -- (1,-1.5);
\draw[thick, rounded corners] (-1,-1.5) -- (-0.5,-1.5) -- (0.5,1.5) -- (1,1.5);
\node[anchor=east] at (-0.8,1.5) {$i_1$};
\node[anchor=east] at (-0.8,0.5) {$i_2$};
\node[anchor=east] at (-0.8,-0.5) {$i_3$};
\node[anchor=east] at (-0.8,-1.5) {$i_4$};
\node[anchor=west] at (0.8,1.5) {$j_1$};
\node[anchor=west] at (0.8,0.5) {$j_2$};
\node[anchor=west] at (0.8,-0.5) {$j_3$};
\node[anchor=west] at (0.8,-1.5) {$j_4$};
\end{tikzpicture}\,.
\end{equation}

The weight associated to a given contraction is called the the Weingarten function, $\Wg(\sigma, d)$. It is a function on elements of $S_k$ and admits an expansion in terms of characters of the symmetric group
\begin{equation}
\Wg(\sigma, d) = \frac{1}{k!} \sum_{\lambda\vdash k} \frac{f_\lambda \chi_\lambda(\sigma)}{c_\lambda(d)},
\label{eq:WgU}
\end{equation}
where we sum we sum over the integer partitions of $k$ that label the irreps of $S_k$; $\chi_\lambda(\sigma)$ is an irreducible character of $\lambda$, and $f_\lambda$ is the dimension of the irrep $\lambda$. The polynomial in the denominator is defined as
\begin{equation}
c_\lambda(d) = \prod_{(i,j)\in \lambda} (d+j-1)\,,
\end{equation}
where we take a product over the coordinates $(i,j)$ of the Young diagram of $\lambda$. Writing $\lambda$ as an integer partition of $k$, with elements $\lambda_i$, the product is taken over $i$ from 1 to $\ell(\lambda)$, the length of the partition, and $j$ from 1 to $\lambda_i$. The expression for the Weingarten function in Eq.~\eqref{eq:WgU}, is valid for $k\geq d$ by restricting the sum over partitions of length $\ell(\lambda)\leq d$ (such that the polynomial $c_\lambda(d)$ in the denominator is free of zeroes). 

The Weingarten functions only depend on the cycle type of the permutation, where the cycle type of $\sigma\in S_k$ is an integer partition of $k$. We end this brief exposition by listing the first few unitary Weingarten functions, labeled by cycle type. For $k=1$, $\Wg((1),d) = \frac{1}{d}$, and for $k=2$, we have
\begin{equation}
\Wg((1,1),d) = \frac{1}{d^2-1}\,, \and \Wg((2),d) = -\frac{1}{d(d^2-1)}\,.
\end{equation}

\paragraph{$k$-fold twirl over $U(d)$}
The $k$-fold twirl of an operator over the unitary group can be written using Eq.~\eqref{eq:Uint} as
\begin{equation}
\CT^{(k)}_{U} (X) = \mathbb{E}_{U} \left[ U^{\otimes k} (X) U^\dagger{}^{\otimes k} \right] = \sum_{\sigma,\tau \in S_k} \Wg(\sigma^{-1}\tau, d) P_\sigma \Tr(X P_\tau)\,.
\label{eq:Utwirl}
\end{equation}
This expression equivalently follows from noting that, by the invariance of the Haar measure, the $k$-fold twirl $\CT^{(k)}_{U} $ is invariant both under $k$-fold unitary conjugation and under $k$-fold conjugation of $X$. 

We also note that the $k$-fold twirl of a permutation operator is $\CT^{(k)}_{U} (P_\rho) = P_\rho$. Eq.~\eqref{eq:Utwirl}, then gives that $\Wg(\sigma^{-1}\tau, d) \Tr(P_\tau P_\rho) = \delta_{\sigma,\rho}$. Viewed as a matrix equation, the matrix of Weingarten functions $\Wg_{(k)}$ is the pseudoinverse of the $k! \times k!$ matrix $G_{(k)}$ of inner products of permutation operators $P_\sigma$ (the Gram matrix of $P_\sigma$'s). The elements of $G_k $ are the inner products between permutation operators, $\Tr(P_\sigma P_\tau) = d^{\ell (\sigma^{-1} \tau)}$, where $\ell (\sigma^{-1} \tau)$ simply counts the number of closed cycles in the permutation product (equivalently, the length of the cycle type of the product):
\begin{equation}
\Wg_{(k)} =G_{(k)}^{-1} \with \Wg_{(k)} = \big(\Wg(\sigma^{-1} \tau, d)\big)_{\sigma,\tau\in S_k} \and G_{(k)} = \big(\Tr(P_\sigma P_\tau)\big)_{\sigma,\tau\in S_k} \,.
\label{eq:Wginv}
\end{equation}
The matrix inverse exists for $k\leq d$. Although elegant, this derivation of the Weingarten functions quickly becomes intractable as we need to invert a $k!\times k!$ matrix. The representation theoretic definition in Eq.~\eqref{eq:WgU} is straightforward to use in computing high moments. 

\paragraph{Wiring diagrams for the first few Haar moments}
To set up the calculations that will follow in the next section, we explicitly write out the wiring diagrams in the first two moments, detailing the index contractions one must take. For $k=1$, we simply have
\begin{align}
\mathbb{E}_U \left[
\begin{tikzpicture}[baseline=-1mm,scale=0.55]
\draw[thick, rounded corners] (-3,0.5) -- (-0.5,0.5) -- (-0.5,-0.75) -- (-3,-0.75);
\draw[thick, rounded corners] (3,0.5) -- (0.5,0.5) -- (0.5,-0.75) -- (3,-0.75);
\draw[rounded corners, fill=white] (-2,0) rectangle (-1,1);
\draw[rounded corners, fill=white] (1,0) rectangle (2,1);
\node at (-1.5,0.5) {$U$};
\node at (1.5,0.5) {$U^\dagger$};
\end{tikzpicture}
\right]
=
\sum_{\sigma,\tau \in S_1} \Wg (\sigma^{-1}\tau,d)
\begin{tikzpicture}[baseline=-2mm,scale=0.55]
\draw[thick] (-1.5,0.5) -- (1.5,0.5);
\draw[thick] (-1.5,-0.75) -- (1.5,-0.75);
\draw[rounded corners, fill=white] (-0.5,0) rectangle (0.5,1);
\draw[rounded corners, fill=white] (-0.5,-1.25) rectangle (0.5,-0.25);
\node at (0,0.5) {$P_\sigma$};
\node at (0,-0.75) {$P_\tau$};
\end{tikzpicture}
~=~
\frac{1}{d}~
\begin{tikzpicture}[baseline=-1mm,scale=0.55]
\foreach \y in {0.5,-0.5}
{
\draw[thick, rounded corners] (-1,\y)--(1,\y);
}
\end{tikzpicture}
\end{align}
For $k=2$, we sum over elements of $S_2$, separately permuting the internal and external indices as 
\begin{align}
&\mathbb{E}_U \left[
\begin{tikzpicture}[baseline,scale=0.55]
\foreach \y in {-0.5,1.5}
{
\draw[thick, rounded corners] (-3,\y) -- (-0.5,\y) -- (-0.5,\y-0.5);
\draw[thick, rounded corners] (0.5,\y-0.5) -- (0.5,\y) -- (3,\y);
\draw[rounded corners,fill=white] (-2,\y-0.5) rectangle (-1,\y+0.5);
\node at (-1.5,\y) {$U$};
\draw[rounded corners,fill=white] (1,\y-0.5) rectangle (2,\y+0.5);
\node at (1.5,\y) {$U^\dagger$};
}
\foreach \y in {-1.5,0.5}
{
\draw[thick, rounded corners] (-3,\y) -- (-0.5,\y) -- (-0.5,\y+0.5);
\draw[thick, rounded corners] (0.5,\y+0.5) -- (0.5,\y) -- (3,\y);
}
\end{tikzpicture}
\right]
=
\sum_{\sigma,\tau \in S_2} \Wg (\sigma^{-1} \tau,d)~
\begin{tikzpicture}[baseline=-1mm,scale=0.55]
\draw[thick] (-3,1.5) -- (3,1.5);
\draw[thick] (-3,-1.5) -- (3,-1.5);
\draw[thick, rounded corners] (-3,0.5) -- (-2,0.5) -- (-1.5,-0.5) -- (1.5,-0.5) -- (2,0.5) -- (3,0.5);
\draw[thick, rounded corners] (-3,-0.5) -- (-2,-0.5) -- (-1.5,0.5) -- (1.5,0.5) -- (2,-0.5) -- (3,-0.5);
\draw[rounded corners,fill=white] (-0.75,0.25) rectangle (0.75,1.75);
\node at (0,1) {$P_\sigma$};
\draw[rounded corners,fill=white] (-0.75,-1.75) rectangle (0.75,-0.25);
\node at (0,-1) {$P_\tau$};
\end{tikzpicture} \vspace*{3pt}\nn
&\qquad\qquad = \frac{1}{d^2-1}
\left(~
\begin{tikzpicture}[baseline,scale=0.55]
\draw[thick, rounded corners] (-1.75,1.5) -- (-0.5,1.5);
\draw[thick, rounded corners] (-1.75,0.5) -- (-1.25,0.5) -- (-0.75,-0.5) -- (-0.5,-0.5);
\draw[thick, rounded corners] (-1.75,-0.5) -- (-1.25,-0.5) -- (-0.75,0.5) -- (-0.5,0.5);
\draw[thick, rounded corners] (-1.75,-1.5) -- (-0.5,-1.5);
\foreach \y in {1.5,0.5,-0.5,-1.5}
{
\draw[thick, rounded corners] (-0.5,\y) -- (0.5,\y);
}
\draw[thick, rounded corners] (0.5,1.5) -- (1.75,1.5);
\draw[thick, rounded corners] (0.5,0.5) -- (0.75,0.5) -- (1.25,-0.5) -- (1.75,-0.5);
\draw[thick, rounded corners] (0.5,-0.5) -- (0.75,-0.5) -- (1.25,0.5) -- (1.75,0.5);
\draw[thick, rounded corners] (0.5,-1.5) -- (1.75,-1.5);
\end{tikzpicture}
~+~
\begin{tikzpicture}[baseline,scale=0.55]
\draw[thick, rounded corners] (-1.75,1.5) -- (-0.5,1.5);
\draw[thick, rounded corners] (-1.75,0.5) -- (-1.25,0.5) -- (-0.75,-0.5) -- (-0.5,-0.5);
\draw[thick, rounded corners] (-1.75,-0.5) -- (-1.25,-0.5) -- (-0.75,0.5) -- (-0.5,0.5);
\draw[thick, rounded corners] (-1.75,-1.5) -- (-0.5,-1.5);
\draw[thick, rounded corners] (-0.5,1.5) -- (-0.25,1.5) -- (0.25,0.5) -- (0.5,0.5);
\draw[thick, rounded corners] (-0.5,0.5) -- (-0.25,0.5) -- (0.25,1.5) -- (0.5,1.5);
\draw[thick, rounded corners] (-0.5,-0.5) -- (-0.25,-0.5) -- (0.25,-1.5) -- (0.5,-1.5);
\draw[thick, rounded corners] (-0.5,-1.5) -- (-0.25,-1.5) -- (0.25,-0.5) -- (0.5,-0.5);
\draw[thick, rounded corners] (0.5,1.5) -- (1.75,1.5);
\draw[thick, rounded corners] (0.5,0.5) -- (0.75,0.5) -- (1.25,-0.5) -- (1.75,-0.5);
\draw[thick, rounded corners] (0.5,-0.5) -- (0.75,-0.5) -- (1.25,0.5) -- (1.75,0.5);
\draw[thick, rounded corners] (0.5,-1.5) -- (1.75,-1.5);
\end{tikzpicture}
~- \frac{1}{d}~
\begin{tikzpicture}[baseline,scale=0.55]
\draw[thick, rounded corners] (-1.75,1.5) -- (-0.5,1.5);
\draw[thick, rounded corners] (-1.75,0.5) -- (-1.25,0.5) -- (-0.75,-0.5) -- (-0.5,-0.5);
\draw[thick, rounded corners] (-1.75,-0.5) -- (-1.25,-0.5) -- (-0.75,0.5) -- (-0.5,0.5);
\draw[thick, rounded corners] (-1.75,-1.5) -- (-0.5,-1.5);
\draw[thick, rounded corners] (-0.5,1.5) -- (-0.25,1.5) -- (0.25,0.5) -- (0.5,0.5);
\draw[thick, rounded corners] (-0.5,0.5) -- (-0.25,0.5) -- (0.25,1.5) -- (0.5,1.5);
\draw[thick, rounded corners] (-0.5,-0.5) -- (0.5,-0.5);
\draw[thick, rounded corners] (-0.5,-1.5) -- (0.5,-1.5);
\draw[thick, rounded corners] (0.5,1.5) -- (1.75,1.5);
\draw[thick, rounded corners] (0.5,0.5) -- (0.75,0.5) -- (1.25,-0.5) -- (1.75,-0.5);
\draw[thick, rounded corners] (0.5,-0.5) -- (0.75,-0.5) -- (1.25,0.5) -- (1.75,0.5);
\draw[thick, rounded corners] (0.5,-1.5) -- (1.75,-1.5);
\end{tikzpicture}
~- \frac{1}{d}~
\begin{tikzpicture}[baseline,scale=0.55]
\draw[thick, rounded corners] (-1.75,1.5) -- (-0.5,1.5);
\draw[thick, rounded corners] (-1.75,0.5) -- (-1.25,0.5) -- (-0.75,-0.5) -- (-0.5,-0.5);
\draw[thick, rounded corners] (-1.75,-0.5) -- (-1.25,-0.5) -- (-0.75,0.5) -- (-0.5,0.5);
\draw[thick, rounded corners] (-1.75,-1.5) -- (-0.5,-1.5);
\draw[thick, rounded corners] (-0.5,1.5) -- (1.5,1.5);
\draw[thick, rounded corners] (-0.5,0.5) -- (0.5,0.5);
\draw[thick, rounded corners] (-0.5,-0.5) -- (-0.25,-0.5) -- (0.25,-1.5) -- (0.5,-1.5);
\draw[thick, rounded corners] (-0.5,-1.5) -- (-0.25,-1.5) -- (0.25,-0.5) -- (0.5,-0.5);
\draw[thick, rounded corners] (0.5,1.5) -- (1.75,1.5);
\draw[thick, rounded corners] (0.5,0.5) -- (0.75,0.5) -- (1.25,-0.5) -- (1.75,-0.5);
\draw[thick, rounded corners] (0.5,-0.5) -- (0.75,-0.5) -- (1.25,0.5) -- (1.75,0.5);
\draw[thick, rounded corners] (0.5,-1.5) -- (1.75,-1.5);
\end{tikzpicture} ~\right) \vspace*{3pt} \nn
&\qquad\qquad = \frac{1}{d^2-1} 
\left(~
\begin{tikzpicture}[baseline,scale=0.55]
\foreach \y in {1.5,0.5,-0.5,-1.5}
{
\draw[thick] (-1,\y) -- (1,\y);
}
\end{tikzpicture}
~+~
\begin{tikzpicture}[baseline,scale=0.55]
\draw[thick, rounded corners] (-1,1.5) -- (-0.5,1.5) -- (0.5,-0.5) -- (1,-0.5);
\draw[thick, rounded corners] (-1,0.5) -- (-0.5,0.5) -- (0.5,-1.5) -- (1,-1.5);
\draw[thick, rounded corners] (-1,-0.5) -- (-0.5,-0.5) -- (0.5,1.5) -- (1,1.5);
\draw[thick, rounded corners] (-1,-1.5) -- (-0.5,-1.5) -- (0.5,0.5) -- (1,0.5);
\end{tikzpicture}
~- \frac{1}{d}~
\begin{tikzpicture}[baseline,scale=0.55]
\draw[thick, rounded corners] (-1,1.5) -- (-0.5,1.5) -- (0.5,-0.5) -- (1,-0.5);
\draw[thick, rounded corners] (-1,0.5) --(1,0.5);
\draw[thick, rounded corners] (-1,-0.5) -- (-0.5,-0.5) -- (0.5,1.5) -- (1,1.5);
\draw[thick, rounded corners] (-1,-1.5) -- (1,-1.5);
\end{tikzpicture}
~- \frac{1}{d}~
\begin{tikzpicture}[baseline,scale=0.55]
\draw[thick, rounded corners] (-1,1.5) -- (1,1.5);
\draw[thick, rounded corners] (-1,0.5) -- (-0.5,0.5) -- (0.5,-1.5) -- (1,-1.5);
\draw[thick, rounded corners] (-1,-0.5) --  (1,-0.5);
\draw[thick, rounded corners] (-1,-1.5) -- (-0.5,-1.5) -- (0.5,0.5) -- (1,0.5);
\end{tikzpicture}~
\right)\,.
\end{align}

\paragraph{Moments of traces}
We can use the formalism introduced above to compute a few simple expressions averaged over the unitary group, which will be of use in later sections. Consider the $2k$-th moment of the trace of a random unitary, $|\Tr (U)|^{2k}$, which we integrate over the unitary group as
\begin{equation}
\mathbb{E}_{U} \big[ |\Tr (U)|^{2k}\big] = \sum_{\sigma, \tau \in S_k} \Wg(\sigma^{-1}\tau,d)\Tr(P_\sigma P_\tau)\,,
\end{equation}
with $ \Tr(P_\sigma P_\tau) = d^{\ell(\sigma\tau)}$.
View this as a matrix equation, and recall that for $k\leq d$ the Weingarten functions are the inverse of the inner products Eq.~\eqref{eq:Wginv}.  Then, we simply have the trace of the identity matrix, a sum over $S_k$:
\begin{equation}
\mathbb{E}_{U} \left[ |\Tr (U)|^{2k} \right] = k!\,. \label{eq:haar_trace}
\end{equation}
This quantity is essentially the same as the frame potential \cite{gross_evenly_2007}, a quantity which quantifies the 2-norm distance between an ensemble of unitaries $\mathcal{E}$ and the Haar ensemble. The frame potential for any ensemble is lower bounded by this Haar value.

\paragraph{Averages of pure states}
Consider a Haar random state $\ket \psi = U \ket 0$, with $\ket 0 \in \C^d$ and $U \in U(d)$, and take the $k$-fold average with respect to the unitary group. Then,
\begin{equation}
\CT^{(k)}_{U}(\ketbra{\psi}^{\otimes k}\big) = \sum_{\sigma,\tau\in  S_k} \Wg(\sigma^{-1}\tau,d) P_\sigma \Tr(P_\tau \ketbra{\psi}^{\otimes k}) = \sum_{\sigma,\tau\in  S_k} \Wg(\sigma^{-1}\tau,d) P_\sigma\,,
\label{eq:chstate}
\end{equation}
as permuting and contracting the pure state moments is the same for any permutation. This also follows from Schur-Weyl duality by noting that the $k$-fold average is invariant under $k$-fold unitary conjugation and may thus be expressed as a sum of permutations. Fixing $\sigma$ above, the sum over $\tau$ just gives the sum over Weingarten functions, which is
\begin{equation}
\sum_{\tau\in S_k} \Wg(\tau, d) = \frac{1}{k!} \binom{k+d-1}{k}^{-1}\,.
\label{eq:Wgsum}
\end{equation}
Equivalently, we can fix this coefficient by taking the trace of Eq.~\eqref{eq:chstate}. Thus we find the the $k$-fold average of a pure state is
\begin{equation}
\CT^{(k)}_{U}(\ketbra{\psi}^{\otimes k}\big) = \binom{k+d-1}{k}^{-1}\Pi_{\rm sym}\,,
\end{equation}
where $\Pi_{\rm sym} = \frac{1}{k!}\sum_{\sigma\in S_k} P_\sigma $ is the projector onto the symmetric subspace and $\binom{k+d-1}{k}$ is the corresponding dimension.

A similar calculation is to consider the moments of the expectation value of a conjugated operator $\braket{\psi|U^\dagger M U|\psi}$, where $\ket\psi\in \C^d$ and a a Hermitian operator $M \in \mathbb{H}_d$. We find
\begin{equation}
\mathbb{E}_{U} \left[ |\!\braket{\psi | U^\dagger M U|\psi}\!|^{k} \right] = \sum_{\sigma, \tau \in S_k} \Wg(\sigma^{-1}\tau,d) \Tr(P_\sigma \ketbra{\psi} )\Tr (P_\tau M^{\otimes k})\,.
\end{equation}
Again, as permuting and contracting tensor products of a pure state just gives one, for any $\tau$ the $\sigma$ sum is just a sum over Weingarten functions. Using  Eq.~\eqref{eq:Wgsum} and recalling the definition of the projector onto the symmetric subspace, we conclude
\begin{equation}
\mathbb{E}_{U} \left[ \braket{\psi | U^\dagger M U|\psi}^{k} \right] = \binom{d+k-1}{k}^{-1} \Tr \left( \Pi_{\rm sym}M^{\otimes k}\right)\,.
\label{eq:state_moment}
\end{equation}

\subsection{Approximate \ktitle-designs and bounds on weight distributions} \label{sub:designs_technical}

Weingarten calculus is a powerful tool. It characterizes twirls over the diagonal representation of the unitary group for arbitrary tensor powers $k \in \mathbb{N}$. In turn, this formula allows for computing moments of random variables that involve Haar random unitaries.  These then can be used to establish \emph{generic} features, such as concentration of measure.
However, full control of \emph{all} moments comes at a price. It is excessively difficult to sample unitaries directly from the Haar measure. Simple dimension counting highlights that circuits of exponential size are required to implement a Haar random unitary circuit on $n$ qudits. 

The notion of $k$-designs introduced in Sec.~\ref{sub:designs} addresses this issue by allowing one to interpolate between Haar-random ($k=\infty$) and highly structured ($k=1$) ensembles.
Unfortunately, very few explicit constructions of $k$-designs are known. 
This lack of efficient constructions can be overcome by relaxing the defining property of a $k$-design.

\begin{definition}[approximate $k$-design] \label{def:approximate_design}
Fix $k \in \mathbb{N}$ and $\epsilon >0$. A unitary ensemble $\CE = \left\{p_i,U_i \right\}_{i=1}^N$ is an $\epsilon$-approximate (unitary) $k$-design if the associated twirling channel
$\mathcal{T}_\CE^{(k)}(X) = \sum_{i=1}^n p_i\, U_i^{\otimes k} X (U_i^\dagger)^{\otimes k}$ obeys
\begin{equation}
\left\| \mathcal{T}^{(k)}_\CE - \mathcal{T}^{(k)}_{U} \right\|_\diamond \leq \frac{k!}{d^{2k}} \epsilon\,.
\end{equation}
Here, $\mathcal{T}^{(k)}_{U}$ denotes the twirl over the full unitary group \eqref{eq:Utwirl} (with respect to the Haar measure).
\end{definition}

This definition readily extends to ensembles of infinite cardinality. Several different definitions of approximate $k$-designs can be found in the literature. By and large these differ in terms of the metric that is used to quantify closeness. We have defined an approximate design up to additive error, but have chosen $\epsilon$ to scale with $d$ in a manner that mimics relative error, similar to the strong definition of a design used in \cite{brandao_local_2016}. This will also simplify exposition considerably. 

The approximate $k$-design property imposes severe restrictions on associated distribution of weights and the ensemble size.

\begin{lemma}[Restatement of \autoref{lem:weights-intro}] \label{lem:weights-restatement}
Let $\CE = \left\{ p_i, U_i \right\}_{i=1}^N$ be an $\epsilon$-approximate $k$ design for $U(d)$. Then,
\begin{equation}
\max_{1 \leq j \leq N} p_j \leq (1+\epsilon) \frac{k!}{d^{2k}} \and N \geq \frac{d^{2k}}{(1+\epsilon)k!}\,.
\end{equation}
\end{lemma}

Lower bounds on approximate $k$-design cardinality are known, see e.g.\ \cite[Lemma~26]{brandao_local_2016} for a similar result. We are not aware of any weight bounds  in the literature.

We also consider orbits of approximate $k$-designs $\CE = \left\{p_i, U_i \right\}_{i=1}^N$. Fix $|x \rangle \in \mathbb{C}^d$ arbitrary and define $| y_i \rangle = U_i |x \rangle$ for $i \in \left[N \right]$. Doing so results in a weighted set 
of unit vectors. 
These sets are called approximate complex-projective $k$-designs \cite{Renes2004,Ambainis2007}. They approximately reproduce the first $k$ moments of the uniform distribution on the complex unit sphere. 
Lower bounds on the cardinality of exact spherical $k$-designs are known, see e.g.\ \cite{scott_tight_2006}, but we are not aware of any statement that bounds the associated weights. 

\begin{lemma} \label{lem:state_weights}
Let $\left\{q_i, |y_i \rangle \right\}_{i=1}^{N'} \subset \mathbb{C}^d$ be the weighted set of \emph{distinct} states contained in an orbit of an $\epsilon$-approximate $k$-design. Then,
\begin{equation}
\max_{j \in \left[N' \right]} q_j \leq (1+\epsilon) \binom{d+k-1}{k}^{-1} \quad \textrm{and} \quad N' \geq \frac{1}{1+\epsilon}\binom{d+k-1}{k}\,.
\end{equation}
\end{lemma}

The emphasis on distinct states is justified. Two or more distinct unitaries can give rise to the same state.

\begin{proof}[Proof of \autoref{lem:weights-restatement}]

Fix $j \in [N]=\left\{1,\ldots,N \right\}$ and use Eq.~\eqref{eq:haar_trace} to conclude
\begin{equation}
\sum_{i=1}^N p_i \left| \Tr (U_j^\dagger U_i) \right|^{2k} = \mathbb{E}_\CE \left[ | \Tr(U_j^\dagger U ) |^{2k} \right]
\leq k! + \underset{\Delta}{\underbrace{\mathbb{E}_\CE \left[ | \Tr(U_j^\dagger U ) |^{2k} \right]- \mathbb{E}_{U}\left[ | \Tr(U_j^\dagger U ) |^{2k} \right]}}.
\label{eq:unitary_design_bound}
\end{equation}
The approximate $k$-design property implies that the mismatch on the r.h.s.\ remains small. Let $| \Omega \rangle = \frac{1}{\sqrt{d}} \sum_{i=1}^d |i \rangle \otimes |i \rangle$ denote the maximally entangled state. Then, $\Tr(U) = d \langle \Omega | U \otimes \mathbb{I} | \Omega \rangle$ and we apply \autoref{def:approximate_design} to bound
\begin{align}
\Delta &= \EE_\CE \left[ | \Tr(U_j^\dagger U ) |^{2k} \right]- \mathbb{E}_{U}\left[ | \Tr(U_j^\dagger U ) |^{2k} \right] \nn
&= d^{2k} \langle \Omega|^{\otimes k} \left( \EE_\CE \left[ \left( (U \otimes \mathbb{I} ) | \Omega \rangle \! \langle \Omega | (U \otimes \mathbb{I})^\dagger \right)^{\otimes k} \right] - \mathbb{E}_{U} \left[ \left( (U \otimes \mathbb{I} ) | \Omega \rangle \! \langle \Omega | (U \otimes \mathbb{I})^\dagger \right)^{\otimes k} \right] \right) | \Omega \rangle^{\otimes k} \nn
&\leq d^{2k} \left\| \EE_\CE \left[ \left( \mathcal{U} \otimes \mathcal{I} (| \Omega \rangle \! \langle \Omega|)\right)^{\otimes k} \right]-
\EE_{U} \left[ \left( \mathcal{U} \otimes \mathcal{I} (| \Omega \rangle \! \langle \Omega|)\right)^{\otimes k} \right] \right\|_\infty \nn
&\leq d^{2k} \left\| \mathcal{T}^{(k)}_\CE - \mathcal{T}^{(k)}_{U} \right\|_\diamond \leq \epsilon k!\,.
\end{align}
Combining both arguments implies
$\sum_{i=1}^N p_i \left| \Tr (U_j^\dagger U_i) \right|^{2k}  \leq (1+\epsilon ) k! $. This allows us to conclude
\begin{equation}
(1+\epsilon) k! \geq \sum_{i=1}^N p_i \left| \Tr \left( U_j^\dagger U_i \right) \right|^{2k}
= \sum_{i \neq j} p_i \left| \Tr \left( U_j^\dagger U_i \right) \right|^{2k} + p_j \left| \Tr \left( U_j^\dagger U_j \right) \right|^{2k} \geq p_j d^{2k}
\end{equation}
for $j \in \left[N\right]$ arbitrary. The lower bound on the cardinality $N$ is an immediate consequence of this weight restriction:
\begin{equation}
1 = \sum_{i=1}^N p_i \leq \sum_{i=1}^N (1+\epsilon) \frac{k!}{d^{2k}} = N (1+\epsilon) \frac{k!}{d^{2k}}\,.
\end{equation}
\end{proof}

\begin{proof}[Proof of \autoref{lem:state_weights}]
The argument is very similar to the proof of \autoref{lem:weights-restatement}.
Fix $j \in \left[N'\right]$, set $M=\ketbra{y_j}$ and use Eq.~\eqref{eq:state_moment} to conclude
\begin{align}
\label{eq:state_design_bound}
&\sum_{i=1}^{N'} q_i \left| \langle y_j, y_i \rangle \right|^{2k} =\sum_{i=1}^N p_i \left| \langle y_j| U_i |x \rangle \right|^2 
= \mathbb{E}_\CE \left[ \langle x|U M U^\dagger |x \rangle\right] \\
&= \binom{d+k-1}{k}^{-1} \Tr \left( \Pi_{\mathrm{sym}} M^{\otimes k} \right) + \underset{\Delta}{\underbrace{\Tr \left( M^{\otimes k} \left( \mathbb{E}_\CE \left[ (U \ketbra{x}U^\dagger)^{\otimes k}\right] - \mathbb{E}_{U} \left[ (U \ketbra{x} U^\dagger)^{\otimes k} \right]\right) \right)}}\,. \nonumber
\end{align}
Next, observe that the Haar average obeys$\Tr \left( \Pi_{\mathrm{sym}} M^{\otimes k} \right) = \Tr \left( \Pi_{\mathrm{sym}} \ketbra{y_j}^{\otimes k} \right) =1$.
The approximate $k$-design property in addition implies that the deviation from this ideal value remains small. Matrix Hoelder asserts
\begin{align}
\Delta &= \Tr \left( M^{\otimes k} \left( \mathbb{E}_\CE \left[ (U \ketbra{x}U^\dagger)^{\otimes k}\right] - \mathbb{E}_{U} \left[ (U \ketbra{x} U^\dagger)^{\otimes k} \right]\right) \right) \nn
&\leq \| M^{\otimes k} \|_\infty  \left\| \mathcal{T}^{(k)}_\CE \left( (|x \rangle \! \langle x|)^{\otimes k} \right) -\mathcal{T}^{(k)}_{U} \left( (|x \rangle \! \langle x|)^{\otimes k} \right) \right\|_1 \nn
&\leq \| M \|_\infty^k \left\| \mathcal{T}^{(k)}_\CE - \mathcal{T}^{(k)}_{U}\right\|_\diamond \leq \epsilon \frac{k!}{d^{2k}} \leq \binom{d+k-1}{k}^{-1} \epsilon\,,
\end{align}
because $\|M\|_\infty = \| \ketbra{y_j}\|_\infty =1$.
This allows us to conclude
\begin{equation}
(1+\epsilon) \binom{d+k-1}{k}^{-1} \geq \sum_{i=1}^{N'} q_i \left| \langle y_j, y_i \rangle \right|^{2k} = q_j | \langle y_j, y_j r\rangle|^{2k} + \sum_{i \neq j} q_j | \langle y_j, y_i \rangle|^{2k} \geq q_j
\end{equation}
for any $j \in \left[N'\right]$.
Both weight and cardinality bound readily follow from this assertion.
\end{proof}

\subsection{A general moment bound for Haar random unitaries}

\begin{theorem}[Detailed restatement of \autoref{thm:main-technical-intro}] \label{thm:master}
Fix $\ket{\phi} \in (\mathbb{C}^d )^{\otimes 2}$ and $M \in \mathbb{H}_d^{\otimes 2}$ such that $\mathbb{I} \succeq M \succeq 0$.
Set
\begin{equation}
S_U(M, \phi) := \Tr \big( M \mathcal{U} \otimes \mathcal{I} ( \ketbra{\phi}) \big) 
=
\begin{tikzpicture}[baseline=-1mm,scale=0.55]
\draw[thick] (-2.5,0.5) -- (2.5,0.5);
\draw[thick] (-2.5,-0.5) -- (2.5,-0.5);
\draw[rounded corners,fill=white] (-0.75,-0.9) rectangle (0.75,0.9);
\node at (0,0) {$M$};
\draw[rounded corners,fill=white] (-2.125,0) rectangle (-1.125,1);
\node at (-1.625,0.5) {$U^\dagger$};
\draw[rounded corners,fill=white] (1.125,0) rectangle (2.125,1);
\node at (1.625,0.5) {$U$};
\draw[rounded corners] (-3,-0.9) rectangle (-2.5,0.9);
\node at (-2.75,0) {$\phi$};
\draw[rounded corners] (2.5,-0.9) rectangle (3,0.9);
\node at (2.75,0) {$\phi$};
\end{tikzpicture}\,,
\end{equation}
where $U \in U(d)$ is chosen uniformly from the Haar measure. Then,
\begin{equation}
\mu (M,\phi) := \mathbb{E}_U \left[ S_U (M,\phi) \right]
 = \frac{1}{d}
\begin{tikzpicture}[baseline,scale=0.55]
\draw[thick, rounded corners] (-1.55,0.5) -- (-1.9,0.5) -- (-1.9,1.25) -- (-0.1,1.25) -- (-0.1,0.5) -- (-0.95,0.5);
\draw[thick, rounded corners] (0.5,0.5) -- (0.1,0.5) -- (0.1,1.25) -- (2.4,1.25) -- (2.4,0.5) -- (2,0.5);
\draw[thick, rounded corners] (0, -0.5) -- (2.4,-0.5) -- (2.4,-1.25) -- (-1.9,-1.25) -- (-1.9,-0.5) -- (0,-0.5);
\draw[rounded corners,fill=white] (0.5,-0.9) rectangle (2,0.9);
\node at (1.25,0) {$M$};
\draw[rounded corners, fill=white] (-0.95,-0.9) rectangle (-0.45,0.9);
\node at (-0.7,0) {$\phi$};
\draw[rounded corners, fill=white] (-1.55,-0.9) rectangle (-1.05,0.9);
\node at (-1.3,0) {$\phi$};
\end{tikzpicture}
= \Tr \left( M \mathcal{D} \otimes \mathcal{I}(| \phi \rangle \! \langle \phi|) \right)\,,
\end{equation}
where $\mathcal{D}(X) = \frac{\Tr(X)}{d} \mathbb{I}$ is the depolarizing channel. Moreover, the following bounds apply to all centered moments of order $k=1,\ldots,d^{2/3}$:
\begin{equation}
\mathbb{E}_U \left[  \left( S_U (M,\phi) - \mu (M,\phi)\right)^{k} \right]
\leq  \frac{C_k (k!)^2}{d^{k/2}}\,.
\end{equation}
Here, $C_k=\tfrac{1}{k+1}\binom{2k}{k}$ is the $k$-th Catalan number. 
\end{theorem}

\subsection{Moment bounds for approximate designs}

\begin{corollary} \label{cor:master}
With the same assumptions in \autoref{thm:master}, but suppose that $U \in U(d)$ is chosen from an  $\epsilon$-approximate unitary $k$-design $\CE$. 
Then,
\begin{equation}
\mathbb{E}_\CE
\left[\Bigg(^{\vphantom{k}}\right.
\underset{S_U (M,\phi)}{\underbrace{
\begin{tikzpicture}[baseline=-1mm,scale=0.55]
\draw[thick] (-2.5,0.5) -- (2.5,0.5);
\draw[thick] (-2.5,-0.5) -- (2.5,-0.5);
\draw[rounded corners,fill=white] (-0.75,-0.9) rectangle (0.75,0.9);
\node at (0,0) {$M$};
\draw[rounded corners,fill=white] (-2.125,0) rectangle (-1.125,1);
\node at (-1.625,0.5) {$U^\dagger$};
\draw[rounded corners,fill=white] (1.125,0) rectangle (2.125,1);
\node at (1.625,0.5) {$U$};
\draw[rounded corners] (-3,-0.9) rectangle (-2.5,0.9);
\node at (-2.75,0) {$\phi$};
\draw[rounded corners] (2.5,-0.9) rectangle (3,0.9);
\node at (2.75,0) {$\phi$};
\end{tikzpicture}
}}
-
\underset{\mu (M,\phi)}{\underbrace{
\frac{1}{d}
\begin{tikzpicture}[baseline=-1mm,scale=0.55]
\draw[thick, rounded corners] (-1.55,0.5) -- (-1.9,0.5) -- (-1.9,1.25) -- (-0.1,1.25) -- (-0.1,0.5) -- (-0.95,0.5);
\draw[thick, rounded corners] (0.5,0.5) -- (0.1,0.5) -- (0.1,1.25) -- (2.4,1.25) -- (2.4,0.5) -- (2,0.5);
\draw[thick, rounded corners] (0, -0.5) -- (2.4,-0.5) -- (2.4,-1.25) -- (-1.9,-1.25) -- (-1.9,-0.5) -- (0,-0.5);
\draw[rounded corners,fill=white] (0.5,-0.9) rectangle (2,0.9);
\node at (1.25,0) {$M$};
\draw[rounded corners, fill=white] (-0.95,-0.9) rectangle (-0.45,0.9);
\node at (-0.7,0) {$\phi$};
\draw[rounded corners, fill=white] (-1.55,-0.9) rectangle (-1.05,0.9);
\node at (-1.3,0) {$\phi$};
\end{tikzpicture}
}}\,
\left. \Bigg)^k \right]
\leq  \frac{(k!)^2}{d^{k/2}} \left( C_k + \frac{\epsilon}{k! d^{3k/2}} \right)\,.
\end{equation}
\end{corollary}

\begin{proof}
We can rewrite random variable and (Haar) expectation as
\begin{equation}
S_U (M,\phi) = \Tr \big( M \mathcal{U} \otimes \mathcal{I} ( \ketbra{\phi} \big) \quad \textrm{and} \quad 
\mu (M,\phi) 
= \Tr \left( \frac{\mathbb{I}}{d} \otimes \Tr_1 (M) \mathcal{U} \otimes \mathcal{I} (\ketbra{\phi}) \right)\,.
\end{equation}
Combine them to obtain
\begin{equation}
\bar{S}_U (M,\phi) = S_U (M,\phi) - \mu (M,\phi) = \Tr \left(\tilde{M} \mathcal{U} \otimes \mathcal{I} (\ketbra{\phi} ) \right)\,,
\end{equation}
where $\tilde{M} = M - \frac{1}{d} \mathbb{I} \otimes \Tr_1 (M) \in \mathbb{H}_d \otimes \mathbb{H}_d$ is a traceless difference of two psd matrices.
Next, fix $k \in \mathbb{N}$ and compare the $k$-th centered moment to its Haar-averaged counter-part:
\begin{equation}
\mathbb{E}_\CE \left[ \bar{S}_U (M,\phi)^k \right]
\leq \mathbb{E}_{U} \left[ \bar{S}_U (M,\phi)^k \right]
+ \underset{\Delta}{\underbrace{\left( \mathbb{E}_\CE \left[ \bar{S}_U (M,\phi)^k \right] - \mathbb{E}_{U} \left[ \bar{S}_U (M,\phi)^k \right] \right)}}\,.
\end{equation}
The first contribution is bounded by \autoref{thm:master} and the approximate $k$-design property (\autoref{def:approximate_design}) ensures that the mismatch $\Delta$ remains controlled:
\begin{align}
\Delta &=
\Tr \left( \tilde{M}^{\otimes k} \left( \mathbb{E}_\CE \left[ \left( \mathcal{U}\otimes \mathcal{I} \right)^{\otimes k} \right] - \mathbb{E}_{U} \left[ \left(\mathcal{U} \otimes \mathcal{I}\right)^{\otimes k} \right] \right) \left( ( \ketbra{\phi})^{\otimes k} \right) \right) \nn
&\leq \| \tilde{M}^{\otimes k} \|_\infty \left\| \left(\mathbb{E}_\CE \left[ \mathcal{U}^{\otimes k} \otimes \mathcal{I} \right] - \mathbb{E}_{U} \left[ \mathcal{U}^{\otimes k} \otimes \mathcal{I} \right]\right) \left( (\ketbra{\phi})^{\otimes k} \right) \right\|_1 \nn
&\leq \| \tilde{M} \|_\infty^k \left\| \mathbb{E}_\CE \left[ \mathcal{U}^{\otimes k} \right] - \mathbb{E}_{U} \left[ \mathcal{U}^{\otimes k} \right] \right\|_\diamond
= \| \tilde{M} \|_\infty^k \left\| \mathcal{T}^{(k)}_\CE - \mathcal{T}^{(k)}_{U} \right\|_\diamond 
\leq \| \tilde{M} \|_\infty^k \frac{k!}{d^{2k}} \epsilon\,.
\end{align}
Finally, use the fact that $\tilde{M}$ is the difference of two psd matrices to conclude
\begin{equation}
\| \tilde{M} \|_\infty \leq \max \left\{ \| M \|_\infty, \| \frac{1}{d} \mathbb{I} \otimes \Tr_1 (M) \|_\infty \right\}
= \max \left\{ \| M \|_\infty, \frac{1}{d} \| \Tr_1 (M) \|_\infty \right\} \leq 1\,,
\end{equation}
where we have also used Eq.~\eqref{eq:partial_operator_norm}.
\end{proof}

\begin{corollary}[Moments of $k$-design orbits]\label{prop:state_moments}
For $|x \rangle \in \mathbb{C}^d$ and a measurement $M \in \mathbb{H}_d$ ($\mathbb{I} \succeq M \succeq 0)$ define
\begin{equation}
\bar{Q}_U(M,x) = \langle x| U^\dagger M U |x \rangle - \frac{\Tr(M)}{d}\,,
\end{equation}
where $U$ is sampled from an $\epsilon$-approximate $k$-design. Then,
\begin{equation}
\mathbb{E}_\CE \left[ \bar{Q}_U (M,x)^k \right] \leq \binom{d+k-1}{k}^{-1} \left( d^{k/2} + \epsilon \right) \leq (1+\epsilon) \left(\frac{k^2}{d}\right)^{k/2}\,.
\end{equation}
\end{corollary}

\begin{proof}
Let $\bar{M}=M-\frac{\Tr(M)}{d} \mathbb{I}$ denote the traceless part of $M$ and note that this reformulation cannot increase the operator norm:$\| \bar{M} \|_\infty \leq \|M\|_\infty \leq 1$.
Moreover,
\begin{equation}
\mathbb{E}_\CE \left[ \bar{Q}_U (M,x)^k \right] \leq \mathbb{E}_{U} \left[ \bar{Q}_U (M,x)^k \right] + \underset{\Delta}{\underbrace{ \mathbb{E}_\CE \left[ \bar{Q}_U (M,x)^k \right]- \mathbb{E}_{U} \left[ \bar{Q}_U (M,x)^k \right]}}
\end{equation}
and $\Delta \leq \| \bar{M} \|_\infty^k \binom{d+k-1}{k}^{-1} \epsilon$ follows from arguments that are analogous to the ones presented in the proof of \autoref{lem:state_weights}. Next, apply Eq.~\eqref{eq:state_moment} to the remaining Haar expectation:
\begin{equation}
\mathbb{E}_{U} \left[ \bar{Q}_U (\bar{M},x) \right]
= \mathbb{E}_{U} \left[ \bra{x} U^\dagger M U \ket{x}^k \right]
= \binom{d+k-1}{k}^{-1} \Tr \left( \Pi_{\rm sym} \bar{M}^{\otimes k} \right)\,.
\end{equation}
This trace can be bounded using $\mathrm{tr}(\bar{M})=0$, $\mathrm{tr}(\bar{M}^l) \leq \mathrm{tr}(\bar{M}^2)^{l/2}$ for $l \geq 2$ and $\mathrm{tr}(\bar{M}^2) = \| \bar{M} \|_2^2 \leq \| M \|_2^2$,  see e.g.\ \cite[Lemma~17]{kueng_low_2017}:
\begin{equation}
\Tr \left( \Pi_{\rm sym} \bar{M}^{\otimes k} \right) \leq  \| \bar{M} \|_2^k \leq \| M \|_2^k \leq  d^{k/2}  \|M \|_\infty^k \leq d^{k/2}\,.
\end{equation}
\end{proof}

\subsection{Proof of the general moment bound}

This section is devoted to proving the general moment bound presented in \autoref{thm:master}.

\paragraph{Reformulation and basic norm bounds}

Fix $M \in \mathbb{H}_d \otimes \mathbb{H}_d$ psd with $\| M \|_\infty \leq 1$ and a state $\ket\phi \in \mathbb{C}^d \otimes \mathbb{C}^d$. Use the vectorization correspondence $| \phi \rangle = \mathrm{vec}(\Phi)$ with $\Phi \in \mathbb{M}_{d \times d}$ to rewrite the random variable defined in \autoref{thm:master}:
\begin{equation}
S_U (M,\phi)=
\begin{tikzpicture}[baseline=-1mm,scale=0.55]
\draw[thick] (-2.5,0.5) -- (2.5,0.5);
\draw[thick] (-2.5,-0.5) -- (2.5,-0.5);
\draw[rounded corners,fill=white] (-0.75,-0.9) rectangle (0.75,0.9);
\node at (0,0) {$M$};
\draw[rounded corners,fill=white] (-2.125,0) rectangle (-1.125,1);
\node at (-1.625,0.5) {$U^\dagger$};
\draw[rounded corners,fill=white] (1.125,0) rectangle (2.125,1);
\node at (1.625,0.5) {$U$};
\draw[rounded corners] (-3,-0.9) rectangle (-2.5,0.9);
\node at (-2.75,0) {$\phi$};
\draw[rounded corners] (2.5,-0.9) rectangle (3,0.9);
\node at (2.75,0) {$\phi$};
\end{tikzpicture}
=
\begin{tikzpicture}[baseline=-1mm,scale=0.55]
\draw[thick,rounded corners] (-2,0.5) -- (2.5,0.5) -- (2.5,-0.5) -- (-2.5,-0.5) -- (-2.5,0.5) -- (-2,0.5);
\draw[rounded corners,fill=white] (-0.75,-0.9) rectangle (0.75,0.9);
\node at (0,0) {$M$};
\draw[rounded corners,fill=white] (-2.125,0.05) rectangle (-1.125,0.95);
\node at (-1.625,0.5) {$U^\dagger$};
\draw[rounded corners, fill=white] (-2.125,-0.95) rectangle (-1.125,-0.05);
\node at (-1.625,-0.5) {$\Phi^\dagger$};
\draw[rounded corners,fill=white] (1.125,0.05) rectangle (2.125,0.95);
\node at (1.625,0.5) {$U$};
\draw[rounded corners,fill=white] (1.125,-0.95) rectangle (2.125,-0.05);
\node at (1.625,-0.5) {$\Phi$};
\end{tikzpicture}
= 
\begin{tikzpicture}[baseline=-1mm,scale=0.55]
\draw[thick,rounded corners] (-2,0.5) -- (2.5,0.5) -- (2.5,-0.5) -- (-2.5,-0.5) -- (-2.5,0.5) -- (-2,0.5);
\draw[rounded corners,fill=white] (-0.75,-0.9) rectangle (0.75,0.9);
\node at (0,0) {$M_\Phi$};
\draw[rounded corners,fill=white] (-2.125,0.05) rectangle (-1.125,0.95);
\node at (-1.625,0.5) {$U^\dagger$};
\draw[rounded corners,fill=white] (1.125,0.05) rectangle (2.125,0.95);
\node at (1.625,0.5) {$U$};
\end{tikzpicture}\,.
\end{equation}
Here, we have implicitly defined $M_\Phi := (\mathbb{I} \otimes \Phi^\dagger) M (\mathbb{I} \otimes \Phi)$.
Also, recall that vectorization is an isometry, i.e.\ $\| \Phi \|_2 = \braket{\phi|\phi}=1$.
The following auxiliary result bounds the 2-norm of $M_\Phi$ and its partial contractions.

\begin{lemma} \label{lem:norm_bound}
Fix a psd matrix $M \in \mathbb{H}_d^{\otimes 2}$ with $\|M \|_\infty \leq 1$ and a matrix $\Phi \in \mathbb{M}_d$ obeying $\| \Phi \|_2 = 1$. Then, $M_\Phi = (\mathbb{I} \otimes \Phi^\dagger) M (\mathbb{I} \otimes \Phi) \in \mathbb{H}_d^{\otimes 2}$ obeys
\begin{equation}
\| \Tr_1 (M_\Phi) \|_2 \leq d
\quad \textrm{and} \quad 
\| M_\Phi \|_2 \leq \sqrt{d}\,, 
\quad \textrm{as well as} \quad 
\| \Tr_2 (M_\Phi) \|_2 \leq \sqrt{d}\,.
\end{equation}
\end{lemma}

\begin{proof}
Observe 
\begin{equation*}
\| \Tr_1 (M_\Phi) \|_2^2
= \!
\begin{tikzpicture}[baseline,scale=0.55]
\draw[thick, rounded corners] (-2,0.5) -- (-3.5,0.5) -- (-3.5,1.25) -- (-1,1.25) -- (-1,0.5) -- (-2,0.5);
\draw[thick, rounded corners] (2,0.5) -- (3.5,0.5) -- (3.5,1.25) -- (1,1.25) -- (1,0.5) -- (2,0.5);
\draw[thick, rounded corners] (0,-0.5) -- (-5,-0.5) -- (-5,-1.25) -- (5,-1.25) -- (5,-0.5) -- (0,-0.5);
\draw[rounded corners, fill=white] (-0.95,-0.95) rectangle (-0.05,-0.05);
\node at (-0.5,-0.5) {$\Phi^\dagger$};
\draw[rounded corners, fill=white] (0.05,-0.95) rectangle (0.95,-0.05);
\node at (0.5,-0.5) {$\Phi$};
\draw[rounded corners, fill=white] (-3,-0.9) rectangle (-1.5,0.9);
\node at (-2.25,0) {$M$};
\draw[rounded corners, fill =white] (1.5,-0.9) rectangle (3,0.9);
\node at (2.25,0) {$M$};
\draw[rounded corners, fill=white] (-4.45,-0.95) rectangle (-3.55,-0.05);
\node at (-4,-0.5) {$\Phi$};
\draw[rounded corners,fill=white] (3.55,-0.95) rectangle (4.45,-0.05);
\node at (4,-0.5) {$\Phi^\dagger$};
\end{tikzpicture}\!
= \Tr \left( \Phi^\dagger \Phi \Tr_1 (M) \Phi^\dagger \Phi \Tr_1 (M) \right)=:h_1 (\Phi^\dagger \Phi)\,.
\end{equation*}
The function $X \mapsto h_1(X)$ is convex, according to \autoref{lem:convex_function} ($M \succeq 0$ implies $\Tr_1 (M) \succeq 0$). Moreover, $\rho = \Phi^\dagger \Phi \in \mathbb{H}_d$ is guaranteed to be a quantum state: $\rho =\Phi^\dagger \Phi \succeq 0$ and $\Tr(\rho) = \| \Phi \|_2^2 =1$.
 The extreme points of the convex set of all quantum states are pure states. The convex function $h_1$ achieves its maximum value at such an extreme point (\autoref{fact:convex_maximum}) and we infer
\begin{equation}
h_1 (\Phi^\dagger \Phi) \leq \max_{\rho \textrm{ state}} h_1 (\rho) = \max_{\ket\psi}\, h_1 (| \psi \rangle \! \langle \psi|) = \max_{\ket\psi}\, \langle \psi| \Tr_1 (M) | \psi \rangle^2 = \| \Tr_1 (M) \|_\infty^2\,.
\end{equation}
Apply Eq.~\eqref{eq:partial_operator_norm} to conclude the first estimate:
$
\| \Tr_1 (M) \|_\infty^2 \leq d^2 \| M \|_\infty \leq d^2.
$
The second bound can be derived in a similar fashion. Observe,
\begin{equation}
\left\| M_\Phi \right\|_2^2
=
\begin{tikzpicture}[baseline,scale=0.55]
\draw[thick, rounded corners] (0,0.5) -- (-3.5,0.5) -- (-3.5,1.25) -- (3.5,1.25) -- (3.5,0.5) -- (0,0.5);
\draw[thick, rounded corners] (0,-0.5) -- (-5,-0.5) -- (-5,-1.25) -- (5,-1.25) -- (5,-0.5) -- (0,-0.5);
\draw[rounded corners, fill=white] (-0.95,-0.95) rectangle (-0.05,-0.05);
\node at (-0.5,-0.5) {$\Phi^\dagger$};
\draw[rounded corners, fill=white] (0.05,-0.95) rectangle (0.95,-0.05);
\node at (0.5,-0.5) {$\Phi$};
\draw[rounded corners, fill=white] (-3,-0.9) rectangle (-1.5,0.9);
\node at (-2.25,0) {$M$};
\draw[rounded corners, fill =white] (1.5,-0.9) rectangle (3,0.9);
\node at (2.25,0) {$M$};
\draw[rounded corners, fill=white] (-4.45,-0.95) rectangle (-3.55,-0.05);
\node at (-4,-0.5) {$\Phi$};
\draw[rounded corners,fill=white] (3.55,-0.95) rectangle (4.45,-0.05);
\node at (4,-0.5) {$\Phi^\dagger$};
\end{tikzpicture}
= \Tr \left( \Phi^\dagger \Phi \otimes \mathbb{I} M \mathbb{I} \otimes \Phi^\dagger \Phi M \right) = h_2 (\Phi^\dagger \Phi)\,.
\end{equation}
The function $h_2 (X)$ is again convex, because $X\mapsto \mathbb{I} \otimes X$ is a linear transformation and $M \succeq 0$.
Moreover, $\rho = \Phi^\dagger \Phi$ is again a quantum state. We infer
\begin{equation}
h_2 (\Phi^\dagger \Phi)
\leq \max_{\rho \textrm{ state}} h_2 (\rho) = \max_{\ket{\psi}} h_2 (| \psi \rangle \! \langle \psi|) = \max_{\ket{\psi}}\left\| \Tr_2 \left( \mathbb{I} \otimes | \psi \rangle \! \langle \psi| M \right) \right\|_2^2\,,
\end{equation}
because convex functions achieve their maximum at the boundary of convex sets (\autoref{fact:convex_maximum}).
Applying the relation $\| X \|_2 \leq \sqrt{d} \| X \|_\infty$ for Schatten norms in $\mathbb{H}_d$, we conclude
\begin{align}
\| M_\Phi \|_2^2 \leq &d\max_{\ket{\psi}} \| \Tr_2 \left( \mathbb{I} \otimes | \psi \rangle \! \langle \psi| M \right) \|_\infty^2  \nonumber \\ 
=& d \left(\max_{\ket{\psi},\ket{x}} (\langle x | \otimes \langle \psi|) M (| x \rangle  \otimes | \psi \rangle)\right)^2 \leq d \| M \|_\infty^2\,.
\end{align}
The final bound can be established directly. Set $\rho = \Phi^\dagger \Phi$ and observe
\begin{equation}
\| \Tr_2 (M_\Phi) \|_2^2
= \begin{tikzpicture}[baseline,scale=0.55]
\draw[thick, rounded corners] (0,0.5) -- (-3.5,0.5) -- (-3.5,1.25) -- (3.5,1.25) -- (3.5,0.5) -- (0,0.5);
\draw[thick, rounded corners] (-2,-0.5) -- (-6,-0.5) -- (-6,-1.25) -- (-1,-1.25) -- (-1,-0.5) -- (-2,-0.5);
\draw[thick, rounded corners] (2,-0.5) -- (6,-0.5) -- (6,-1.25) -- (1,-1.25) -- (1,-0.5) -- (2,-0.5);
\draw[rounded corners, fill=white] (-3,-0.9) rectangle (-1.5,0.9);
\node at (-2.25,0) {$M$};
\draw[rounded corners, fill =white] (1.5,-0.9) rectangle (3,0.9);
\node at (2.25,0) {$M$};
\draw[rounded corners, fill=white] (-4.45,-0.95) rectangle (-3.55,-0.05);
\node at (-4,-0.5) {$\Phi$};
\draw[rounded corners, fill=white] (-5.45,-0.95) rectangle (-4.55,-0.05);
\node at (-5,-0.5) {$\Phi^\dagger$};
\draw[rounded corners, fill=white] (3.55,-0.95) rectangle (4.45,-0.05);
\node at (4,-0.5) {$\Phi^\dagger$};
\draw[rounded corners, fill=white] (4.55,-0.95) rectangle (5.45,-0.05);
\node at (5,-0.5) {$\Phi$};
\end{tikzpicture}
= \left\| \Tr_2 \left( \mathbb{I} \otimes \rho M \right) \right\|_2^2.
\end{equation}
Apply $\| X \|_2 \leq \sqrt{d} \|X \|_\infty$ to simplify further:
\begin{align}
\| \Tr_2 (\mathbb{I} \otimes \rho M ) \|_2
&\leq \sqrt{d} \| \Tr_2 (\mathbb{I} \otimes \rho M) \|_\infty
\leq \sqrt{d} \max_{\ket{x}} \big| \langle x| \Tr_2 (\mathbb{I} \otimes \rho M ) |x \rangle \big| \nn
&= \sqrt{d} \max_{\ket{x}} \big| \Tr \left( |x \rangle \! \langle x| \otimes \rho M \right) \big|\,.
\end{align}
Finally, use matrix Hoelder \eqref{eq:hoelder} to infer the advertised bound:
\begin{equation}
\sqrt{d} \max_{\ket{x}} \left| \Tr \left( |x \rangle \! \langle x| \otimes \rho M \right) \right|
\leq \sqrt{d} \max_{\ket{x}} \| |x \rangle\! \langle x| \otimes \rho \|_1 \| M \|_\infty = \sqrt{d} \| M \|_\infty\,.
\end{equation}
\end{proof}

\paragraph{Expectation value and centering}

The following result is well known in the literature, see e.g.\ \cite{emerson05}. We include a self-contained derivation based on wiring diagrams for the sake of completeness.

\begin{lemma}[Averaging unitary channels produces the depolarizing channel] \label{lem:expectation}
Fix a psd matrix $M \in \mathbb{H}_d^{\otimes 2}$ and $\ket\phi \in \mathbb{C}^d \otimes \mathbb{C}^d$. Let $\mathcal{U}(X) = UXU^\dagger$ be a Haar-random unitary channel. Then,
\begin{equation}
\mathbb{E}_U \left[ \Tr \left( M \mathcal{U} \otimes \mathcal{I} (| \phi \rangle \! \langle \phi| ) \right) \right] = \Tr \left(M \mathcal{D} \otimes \mathcal{I} (| \phi \rangle \! \langle \phi|) \right)
\quad \textrm{with} \quad \mathcal{D}(\rho) = \tfrac{\Tr (\rho)}{d} \mathbb{I}\,.
\end{equation}
\end{lemma}

\begin{proof}
Averaging over a single unitary $U$ and its adjoint decouples the register in question. Combine this with the reformulation from the previous paragraph to conclude
\begin{align}
\mathbb{E}_U \left[ \Tr \left( M \mathcal{U} \otimes \mathcal{I} (| \phi \rangle \! \langle \phi|) \right)\right] &=
\mathbb{E} \left[ \begin{tikzpicture}[baseline=-1mm,scale=0.55]
\draw[thick] (-2.5,0.5) -- (2.5,0.5);
\draw[thick] (-2.5,-0.5) -- (2.5,-0.5);
\draw[rounded corners,fill=white] (-0.75,-0.9) rectangle (0.75,0.9);
\node at (0,0) {$M$};
\draw[rounded corners,fill=white] (-2.125,0) rectangle (-1.125,1);
\node at (-1.625,0.5) {$U^\dagger$};
\draw[rounded corners,fill=white] (1.125,0) rectangle (2.125,1);
\node at (1.625,0.5) {$U$};
\draw[rounded corners] (-3,-0.9) rectangle (-2.5,0.9);
\node at (-2.75,0) {$\phi$};
\draw[rounded corners] (2.5,-0.9) rectangle (3,0.9);
\node at (2.75,0) {$\phi$};
\end{tikzpicture}
\right] 
= \mathbb{E} \left[
\begin{tikzpicture}[baseline=-1mm,scale=0.55]
\draw[thick,rounded corners] (-2,0.5) -- (2.5,0.5) -- (2.5,-0.5) -- (-2.5,-0.5) -- (-2.5,0.5) -- (-2,0.5);
\draw[rounded corners,fill=white] (-0.75,-0.9) rectangle (0.75,0.9);
\node at (0,0) {$M_\Phi$};
\draw[rounded corners,fill=white] (-2.125,0.05) rectangle (-1.125,0.95);
\node at (-1.625,0.5) {$U^\dagger$};
\draw[rounded corners,fill=white] (1.125,0.05) rectangle (2.125,0.95);
\node at (1.625,0.5) {$U$};
\end{tikzpicture}
\right] \nn
&=
\frac{1}{d}\,
\begin{tikzpicture}[baseline=-1mm,scale=0.55]
\draw[thick, rounded corners] (0,0.5) -- (1.15,0.5) -- (1.15,1.25) -- (-1.15,1.25) -- (-1.15,0.5) -- (0,0.5);
\draw[thick, rounded corners] (0,-0.5) -- (1.15,-0.5) -- (1.15,-1.25) -- (-1.15,-1.25) -- (-1.15,-0.5) -- (0,-0.5);
\draw[rounded corners,fill=white] (-0.75,-0.9) rectangle (0.75,0.9);
\node at (0,0) {$M_\Phi$};
\end{tikzpicture} 
=
\frac{1}{d}\,
\begin{tikzpicture}[baseline=-1mm,scale=0.55]
\draw[thick, rounded corners] (0.5,0.5) -- (0.1,0.5) -- (0.1,1.25) -- (2.4,1.25) -- (2.4,0.5) -- (2,0.5);
\draw[thick, rounded corners] (0, -0.5) -- (2.4,-0.5) -- (2.4,-1.25) -- (-2.4,-1.25) -- (-2.4,-0.5) -- (0,-0.5);
\draw[rounded corners,fill=white] (0.5,-0.9) rectangle (2,0.9);
\node at (1.25,0) {$M$};
\draw[rounded corners, fill=white] (-0.95,-0.95) rectangle (-0.05,-0.05);
\node at (-0.5,-0.5) {$\Phi^\dagger$};
\draw[rounded corners, fill=white] (-1.95,-0.95) rectangle (-1.05,-0.05);
\node at (-1.5,-0.5) {$\Phi$};
\end{tikzpicture}
= \frac{1}{d}\,\begin{tikzpicture}[baseline=-1mm,scale=0.55]
\draw[thick, rounded corners] (-1.55,0.5) -- (-1.9,0.5) -- (-1.9,1.25) -- (-0.1,1.25) -- (-0.1,0.5) -- (-0.95,0.5);
\draw[thick, rounded corners] (0.5,0.5) -- (0.1,0.5) -- (0.1,1.25) -- (2.4,1.25) -- (2.4,0.5) -- (2,0.5);
\draw[thick, rounded corners] (0, -0.5) -- (2.4,-0.5) -- (2.4,-1.25) -- (-1.9,-1.25) -- (-1.9,-0.5) -- (0,-0.5);
\draw[rounded corners,fill=white] (0.5,-0.9) rectangle (2,0.9);
\node at (1.25,0) {$M$};
\draw[rounded corners, fill=white] (-0.95,-0.9) rectangle (-0.45,0.9);
\node at (-0.7,0) {$\phi$};
\draw[rounded corners, fill=white] (-1.55,-0.9) rectangle (-1.05,0.9);
\node at (-1.3,0) {$\phi$};
\end{tikzpicture}\,.
\end{align}
The connection to the depolarizing channel readily follows from $\mathcal{D} \otimes \mathcal{I}(| \phi \rangle \! \langle \phi|) = \frac{\mathbb{I}}{d} \otimes \Tr_2 (| \phi \rangle \! \langle \phi|)$.
\end{proof}

\begin{corollary}[Reformulation of the centered random variable] \label{cor:centering}
Fix $\ket\phi \in \mathbb{C}^d \otimes \mathbb{C}^d$ (state) and $M \in \mathbb{H}_d^{\otimes 2}$ such that $\mathbb{I} \succeq M \succeq 0$ (measurement). For channels $\mathcal{U}(X) = U XU^\dagger$ and  $\mathcal{D}(X) = \frac{\Tr(X)}{d} \mathbb{I}$ define
\begin{equation*}
S_U (M,\phi) = \Tr \left( M \mathcal{U}\otimes \mathcal{I} (| \phi \rangle \! \langle \phi|) \right), \quad \textrm{as well as} \quad 
\mu (M,\phi) =\Tr \left( M \mathcal{D}\otimes \mathcal{I} (| \phi \rangle \! \langle \phi|) \right).
\end{equation*}
Then, we may rewrite the difference of these variables as
\begin{equation}
\bar{S}_U (M,\phi) = S_U (M,\phi) - \mu (M,\Phi) = 
\begin{tikzpicture}[baseline=-1mm,scale=0.55]
\draw[thick,rounded corners] (-2,0.5) -- (2.5,0.5) -- (2.5,-0.5) -- (-2.5,-0.5) -- (-2.5,0.5) -- (-2,0.5);
\draw[rounded corners,fill=white] (-0.75,-0.9) rectangle (0.75,0.9);
\node at (0,0) {$\bar{M}_\Phi$};
\draw[rounded corners,fill=white] (-2.125,0.05) rectangle (-1.125,0.95);
\node at (-1.625,0.5) {$U^\dagger$};
\draw[rounded corners,fill=white] (1.125,0.05) rectangle (2.125,0.95);
\node at (1.625,0.5) {$U$};
\end{tikzpicture}\,,
\end{equation}
where $\bar{M}_\Phi = M_\Phi - \frac{\Tr(M_\Phi)}{d} \mathbb{I} \in \mathbb{H}_d^{\otimes 2}$ is the traceless part of $M_\Phi$ (i.e.\ $\Tr(\bar{M}_\Phi)=0$).
\end{corollary}

\ni This reformulation immediately follows from the proof of \autoref{lem:expectation}, provided that we rewrite
\begin{equation}
\mu (M,\phi) =\frac{1}{d} \Tr(M_\Phi) = \frac{\Tr (M_\Phi)}{d^2}\, \begin{tikzpicture}[baseline=-1mm,scale=0.55]
\draw[thick,rounded corners] (0,0.5) -- (1.75,0.5) -- (1.75,-0.5) -- (-1.75,-0.5) -- (-1.75,0.5) -- (0,0.5);
\draw[rounded corners,fill=white] (-1.45,0.05) rectangle (-0.55,0.95);
\node at (-1,0.5) {$U^\dagger$};
\draw[rounded corners,fill=white] (0.55,0.05) rectangle (1.45,0.95);
\node at (1,0.5) {$U$};
\end{tikzpicture}\,.
\end{equation}

\paragraph{Bounds on centered moments}

\begin{lemma}
With the same assumptions and notation as in \autoref{cor:centering}, suppose that $U \in U(d)$ is chosen uniformly from the Haar measure. Then, for any $k \leq d^{2/3}$
\begin{equation}
\mathbb{E}_{U} \left[ \bar{S}_U (M,\phi)^k \right] \leq C_k \frac{(k!)^2}{d^{k/2}}\,,
\end{equation}
where $C_k=\tfrac{1}{k-1}\binom{2k}{k}$ is the $k$-th Catalan number.
\end{lemma}

\begin{proof}
It is instructive to first analyze and understand the second moment:
\begin{align}
\mathbb{E}_U \left[ \bar{S}_U (M,\phi)^2 \right]
=\mathbb{E}_U \left[
\begin{array}{c}
\begin{tikzpicture}[baseline,scale=0.55]
\draw[thick,rounded corners] (-2,0.5) -- (2.5,0.5) -- (2.5,-0.5) -- (-2.5,-0.5) -- (-2.5,0.5) -- (-2,0.5);
\draw[rounded corners,fill=white] (-0.75,-0.9) rectangle (0.75,0.9);
\node at (0,0) {$\bar{M}_\Phi$};
\draw[rounded corners,fill=white] (-2.125,0.05) rectangle (-1.125,0.95);
\node at (-1.625,0.5) {$U^\dagger$};
\draw[rounded corners,fill=white] (1.125,0.05) rectangle (2.125,0.95);
\node at (1.625,0.5) {$U$};
\end{tikzpicture}
\\
\begin{tikzpicture}[baseline,scale=0.55]
\draw[thick,rounded corners] (-2,0.5) -- (2.5,0.5) -- (2.5,-0.5) -- (-2.5,-0.5) -- (-2.5,0.5) -- (-2,0.5);
\draw[rounded corners,fill=white] (-0.75,-0.9) rectangle (0.75,0.9);
\node at (0,0) {$\bar{M}_\Phi$};
\draw[rounded corners,fill=white] (-2.125,0.05) rectangle (-1.125,0.95);
\node at (-1.625,0.5) {$U^\dagger$};
\draw[rounded corners,fill=white] (1.125,0.05) rectangle (2.125,0.95);
\node at (1.625,0.5) {$U$};
\end{tikzpicture}
\end{array}
\right]
=
 \mathbb{E}_U \left[
\begin{tikzpicture}[baseline,scale=0.55]
\foreach \y in {-0.5,1.5}
{
\draw[thick, rounded corners] (-2,\y) -- (-0.5,\y) -- (-0.5,\y-0.5);
\draw[thick, rounded corners] (0.5,\y-0.5) -- (0.5,\y) -- (3,\y);
\draw[rounded corners,fill=white] (-2,\y-0.5) rectangle (-1,\y+0.5);
\node at (-1.5,\y) {$U$};
\draw[rounded corners,fill=white] (1,\y-0.5) rectangle (2,\y+0.5);
\node at (1.5,\y) {$U^\dagger$};
}
\foreach \y in {-1.5,0.5}
{
\draw[thick, rounded corners] (-2,\y) -- (-0.5,\y) -- (-0.5,\y+0.5);
\draw[thick, rounded corners] (0.5,\y+0.5) -- (0.5,\y) -- (3,\y);
}
\draw[rounded corners] (3,0.1) rectangle (4.5,1.9);
\node at (3.75,1) {$\bar{M}_\Phi$};
\draw[rounded corners] (3,-1.9) rectangle (4.5,-0.1);
\node at (3.75,-1) {$\bar{M}_\Phi$};
\draw[thick, rounded corners] (4.5,1.5) -- (5,1.5) -- (5,2.25) -- (-2.5,2.25) -- (-2.5,1.5) -- (-2,1.5);
\draw[thick, rounded corners] (4.5,-1.5) -- (5,-1.5) -- (5,-2.25) -- (-2.5,-2.25) -- (-2.5,-1.5) -- (-2,-1.5);
\draw[thick, rounded corners] (4.5,0.5) -- (5.5,0.5) -- (5.5,2.5) -- (-3,2.5)  -- (-3,0.5) -- (-2,0.5);
\draw[thick, rounded corners] (4.5,-0.5) -- (5.5,-0.5) -- (5.5,-2.5) -- (-3,-2.5) -- (-3,-0.5) -- (-2,-0.5);
\end{tikzpicture}
\right].
\end{align}
For $k=2$ there are two permutations: the identity permutation $\iden = \{1,2\}$ and swap (or flip) $S = \{2,1\}$. 
This results in $(k!)^2=4$ different contributions to the formula: $(\iden,\iden)$, $(S,S)$, $(S,\iden)$ and $(\iden,S)$ contribute each. 
The associated Weingarten functions are $\Wg((1),d) = \frac{1}{d^2-1}$ and $\Wg ((2),d) = -\frac{1}{d(d^2-1)}$.
Ignoring the common factor $\frac{1}{d^2-1}$, the individual contributions become
\begin{align}
&\begin{tikzpicture}[baseline=-1mm,scale=0.55]
\foreach \y in {1.5,0.5,-0.5,-1.5}
{
\draw[thick, rounded corners] (-1,\y) -- (1,\y);
}
\draw[rounded corners] (1,0.1) rectangle (2.5,1.9);
\node at (1.75,1) {$\bar{M}_\Phi$};
\draw[rounded corners] (1,-1.9) rectangle (2.5,-0.1);
\node at (1.75,-1) {$\bar{M}_\Phi$};
\draw[thick, rounded corners] (2.5,1.5) -- (3,1.5) -- (3,2.25) -- (-1.25,2.25) -- (-1.25,1.5) -- (-1,1.5);
\draw[thick, rounded corners] (2.5,-1.5) -- (3,-1.5) -- (3,-2.25) -- (-1.25,-2.25) -- (-1.25,-1.5) -- (-1,-1.5);
\draw[thick, rounded corners] (2.5,0.5) -- (3.5,0.5) -- (3.5,2.5) -- (-1.5,2.5) -- (-1.5,0.5) -- (-1,0.5);
\draw[thick, rounded corners] (2.5,-0.5) -- (3.5,-0.5) -- (3.5,-2.5) -- (-1.5,-2.5) -- (-1.5,-0.5) -- (-1,-0.5);
\end{tikzpicture}
+
\begin{tikzpicture}[baseline=-1mm,scale=0.55]
\draw[thick, rounded corners] (-1,1.5) -- (-0.5,1.5) -- (0.5,-0.5) -- (1,-0.5);
\draw[thick, rounded corners] (-1,0.5) -- (-0.5,0.5) -- (0.5,-1.5) -- (1,-1.5);
\draw[thick, rounded corners] (-1,-0.5) -- (-0.5,-0.5) -- (0.5,1.5) -- (1,1.5);
\draw[thick, rounded corners] (-1,-1.5) -- (-0.5,-1.5) -- (0.5,0.5) -- (1,0.5);
\draw[rounded corners] (1,0.1) rectangle (2.5,1.9);
\node at (1.75,1) {$\bar{M}_\Phi$};
\draw[rounded corners] (1,-1.9) rectangle (2.5,-0.1);
\node at (1.75,-1) {$\bar{M}_\Phi$};
\draw[thick, rounded corners] (2.5,1.5) -- (3,1.5) -- (3,2.25) -- (-1.25,2.25) -- (-1.25,1.5) -- (-1,1.5);
\draw[thick, rounded corners] (2.5,-1.5) -- (3,-1.5) -- (3,-2.25) -- (-1.25,-2.25) -- (-1.25,-1.5) -- (-1,-1.5);
\draw[thick, rounded corners] (2.5,0.5) -- (3.5,0.5) -- (3.5,2.5) -- (-1.5,2.5) -- (-1.5,0.5) -- (-1,0.5);
\draw[thick, rounded corners] (2.5,-0.5) -- (3.5,-0.5) -- (3.5,-2.5) -- (-1.5,-2.5) -- (-1.5,-0.5) -- (-1,-0.5);
\end{tikzpicture}
- \frac{1}{d}
\begin{tikzpicture}[baseline=-1mm,scale=0.55]
\draw[thick, rounded corners] (-1,1.5) -- (-0.5,1.5) -- (0.5,-0.5) -- (1,-0.5);
\draw[thick, rounded corners] (-1,0.5) --(1,0.5);
\draw[thick, rounded corners] (-1,-0.5) -- (-0.5,-0.5) -- (0.5,1.5) -- (1,1.5);
\draw[thick, rounded corners] (-1,-1.5) -- (1,-1.5);
\draw[rounded corners] (1,0.1) rectangle (2.5,1.9);
\node at (1.75,1) {$\bar{M}_\Phi$};
\draw[rounded corners] (1,-1.9) rectangle (2.5,-0.1);
\node at (1.75,-1) {$\bar{M}_\Phi$};
\draw[thick, rounded corners] (2.5,1.5) -- (3,1.5) -- (3,2.25) -- (-1.25,2.25) -- (-1.25,1.5) -- (-1,1.5);
\draw[thick, rounded corners] (2.5,-1.5) -- (3,-1.5) -- (3,-2.25) -- (-1.25,-2.25) -- (-1.25,-1.5) -- (-1,-1.5);
\draw[thick, rounded corners] (2.5,0.5) -- (3.5,0.5) -- (3.5,2.5) -- (-1.5,2.5) -- (-1.5,0.5) -- (-1,0.5);
\draw[thick, rounded corners] (2.5,-0.5) -- (3.5,-0.5) -- (3.5,-2.5) -- (-1.5,-2.5) -- (-1.5,-0.5) -- (-1,-0.5);
\end{tikzpicture}
-
\frac{1}{d}
\begin{tikzpicture}[baseline=-1mm,scale=0.55]
\draw[thick, rounded corners] (-1,1.5) -- (1,1.5);
\draw[thick, rounded corners] (-1,0.5) -- (-0.5,0.5) -- (0.5,-1.5) -- (1,-1.5);
\draw[thick, rounded corners] (-1,-0.5) --  (1,-0.5);
\draw[thick, rounded corners] (-1,-1.5) -- (-0.5,-1.5) -- (0.5,0.5) -- (1,0.5);
\draw[rounded corners] (1,0.1) rectangle (2.5,1.9);
\node at (1.75,1) {$\bar{M}_\Phi$};
\draw[rounded corners] (1,-1.9) rectangle (2.5,-0.1);
\node at (1.75,-1) {$\bar{M}_\Phi$};
\draw[thick, rounded corners] (2.5,1.5) -- (3,1.5) -- (3,2.25) -- (-1.25,2.25) -- (-1.25,1.5) -- (-1,1.5);
\draw[thick, rounded corners] (2.5,-1.5) -- (3,-1.5) -- (3,-2.25) -- (-1.25,-2.25) -- (-1.25,-1.5) -- (-1,-1.5);
\draw[thick, rounded corners] (2.5,0.5) -- (3.5,0.5) -- (3.5,2.5) -- (-1.5,2.5) -- (-1.5,0.5) -- (-1,0.5);
\draw[thick, rounded corners] (2.5,-0.5) -- (3.5,-0.5) -- (3.5,-2.5) -- (-1.5,-2.5) -- (-1.5,-0.5) -- (-1,-0.5);
\end{tikzpicture}
\nn
&=
\begin{tikzpicture}[baseline=-1mm,scale=0.55]
\draw[rounded corners] (-2,-0.9) rectangle (-0.5,0.9);
\node at (-1.25,0) {$\bar{M}_\Phi$};
\draw[rounded corners] (0.5,-0.9) rectangle (2,0.9);
\node at (1.25,0) {$\bar{M}_\Phi$};
\draw[thick,rounded corners] (-0.5,0.5) -- (-0.25,0.5) -- (-0.25,1.25) -- (-2.25,1.25) -- (-2.25,0.5) -- (-2,0.5);
\draw[thick, rounded corners] (2,0.5) -- (2.25,0.5) -- (2.25,1.25) -- (0.25,1.25) -- (0.25,0.5) -- (0.5,0.5);
\draw[thick,rounded corners] (-0.5,-0.5) -- (-0.25,-0.5) -- (-0.25,-1.25) -- (-2.25,-1.25) -- (-2.25,-0.5) -- (-2,-0.5);
\draw[thick, rounded corners] (2,-0.5) -- (2.25,-0.5) -- (2.25,-1.25) -- (0.25,-1.25) -- (0.25,-0.5) -- (0.5,-0.5);
\end{tikzpicture}
+
\begin{tikzpicture}[baseline=-1mm,scale=0.55]
\draw[rounded corners] (-2,-0.9) rectangle (-0.5,0.9);
\node at (-1.25,0) {$\bar{M}_\Phi$};
\draw[rounded corners] (0.5,-0.9) rectangle (2,0.9);
\node at (1.25,0) {$\bar{M}_\Phi$};
\draw[thick, rounded corners] (2,0.5) -- (2.25,0.5) -- (2.25,1.25) -- (-2.25,1.25) -- (-2.25,0.5) -- (-2,0.5);
\draw[thick] (-0.5,0.5) -- (0.5,0.5);
\draw[thick, rounded corners] (2,-0.5) -- (2.25,-0.5) -- (2.25,-1.25) -- (-2.25,-1.25) -- (-2.25,-0.5) -- (-2,-0.5);
\draw[thick] (-0.5,-0.5) -- (0.5,-0.5);
\end{tikzpicture}
- \frac{1}{d}
\begin{tikzpicture}[baseline=-1mm,scale=0.55]
\draw[rounded corners] (-2,-0.9) rectangle (-0.5,0.9);
\node at (-1.25,0) {$\bar{M}_\Phi$};
\draw[rounded corners] (0.5,-0.9) rectangle (2,0.9);
\node at (1.25,0) {$\bar{M}_\Phi$};
\draw[thick, rounded corners] (2,0.5) -- (2.25,0.5) -- (2.25,1.25) -- (-2.25,1.25) -- (-2.25,0.5) -- (-2,0.5);
\draw[thick] (-0.5,0.5) -- (0.5,0.5);
\draw[thick,rounded corners] (-0.5,-0.5) -- (-0.25,-0.5) -- (-0.25,-1.25) -- (-2.25,-1.25) -- (-2.25,-0.5) -- (-2,-0.5);
\draw[thick, rounded corners] (2,-0.5) -- (2.25,-0.5) -- (2.25,-1.25) -- (0.25,-1.25) -- (0.25,-0.5) -- (0.5,-0.5);
\end{tikzpicture}
- \frac{1}{d}
\begin{tikzpicture}[baseline=-1mm,scale=0.55]
\draw[rounded corners] (-2,-0.9) rectangle (-0.5,0.9);
\node at (-1.25,0) {$\bar{M}_\Phi$};
\draw[rounded corners] (0.5,-0.9) rectangle (2,0.9);
\node at (1.25,0) {$\bar{M}_\Phi$};
\draw[thick,rounded corners] (-0.5,0.5) -- (-0.25,0.5) -- (-0.25,1.25) -- (-2.25,1.25) -- (-2.25,0.5) -- (-2,0.5);
\draw[thick, rounded corners] (2,0.5) -- (2.25,0.5) -- (2.25,1.25) -- (0.25,1.25) -- (0.25,0.5) -- (0.5,0.5);
\draw[thick, rounded corners] (2,-0.5) -- (2.25,-0.5) -- (2.25,-1.25) -- (-2.25,-1.25) -- (-2.25,-0.5) -- (-2,-0.5);
\draw[thick] (-0.5,-0.5) -- (0.5,-0.5);
\end{tikzpicture} 
\label{eq:master-aux}
\end{align}
Each term is a full contraction that is also called a tensor network \cite{bridgeman_handwaving_2017,kliesch_guaranteed_2017}.  
There are three possible constituents for each tensor network: $\bar{M}_\Phi$, $\Tr_2 (\bar{M}_\Phi)$, and $\Tr_1 (\bar{M}_\Phi)$. 
Importantly, no full self-contractions can contribute to the overall sum, because $\bar{M}_\phi$ is traceless. This ensures that networks with self-contractions -- like the first term --  evaluate to zero.
Moreover, \autoref{lem:norm_bound} bounds the 2-norm of each elementary constituent:
\begin{equation}
\left\|
\begin{tikzpicture}[baseline=-1mm,scale=0.55]
\draw[thick] (-1.25,0.5) -- (1.25,0.5);
\draw[thick] (-1.25,-0.5) -- (1.25,-0.5);
\draw[rounded corners,fill=white] (-0.75,-0.9) rectangle (0.75,0.9);
\node at (0,0) {$\bar{M}_\Phi$};
\end{tikzpicture}
\right\|_2 \leq \sqrt{d}\,, \qquad
\left\|
\begin{tikzpicture}[baseline=-2mm,scale=0.55]
\draw[thick,rounded corners] (0,-0.5) -- (1.25,-0.5) -- (1.25,-1.25) -- (-1.25,-1.25) -- (-1.25,-0.5) -- (0,-0.5);
\draw[thick] (-1.25,0.5) -- (1.25,0.5);
\draw[rounded corners,fill=white] (-0.75,-0.9) rectangle (0.75,0.9);
\node at (0,0) {$\bar{M}_\Phi$};
\end{tikzpicture}
\right\|_2 \leq \sqrt{d}\,, \qquad
\left\|
\begin{tikzpicture}[baseline,scale=0.55]
\draw[thick,rounded corners] (0,0.5) -- (1.25,0.5) -- (1.25,1.25) -- (-1.25,1.25) -- (-1.25,0.5) -- (0,0.5);
\draw[thick] (-1.25,-0.5) -- (1.25,-0.5);
\draw[rounded corners,fill=white] (-0.75,-0.9) rectangle (0.75,0.9);
\node at (0,0) {$\bar{M}_\Phi$};
\end{tikzpicture}
\right\|_2 \leq d\,.
\label{eq:norm_bound_technical}
\end{equation}
The final bound is considerably larger than the rest. However, the corresponding contribution in the sum \eqref{eq:master-aux} is also suppressed by an additional dimension factor. This is not a coincidence: term 3 can only arise if the cycle classes of ($\sigma$,$\tau$) differ from each other. This feature reflects itself in the Weingarten function. 
For the second moment, we thus obtain the following simple bound (ignoring signs):
\begin{equation}
\mathbb{E}_U \left[ \bar{S}(M,\phi)^2 \right] \leq \frac{0+d+d/d+d^2/d}{d^2-1} =\frac{2d+1}{d^2} \leq 4 d^{-1}\,.
\end{equation}
It immediately follows from upper-bounding individual terms using Eq.~\eqref{eq:norm_bound_technical}. 

This general strategy also applies to higher moments. Fix $k \geq 3$ arbitrary. Then, Weingarten calculus implies 
\begin{equation}
\mathbb{E}_U \left[ \bar{S}_U (M,\phi)^k \right] = \sum_{\sigma,\tau \in S_k} \Wg_d (\sigma,\tau) N_{\sigma,\tau} \left( \bar{M}_\Phi, \Tr_2 (\bar{M}_\Phi), \Tr_1 (\bar{M}_\Phi)\right)\,.
\end{equation}
Here, each $N_{\sigma,\tau} (\cdot)$ indicates a tensor network diagram that combines (at most) three elementary building blocks according to rules that are dictated by the permutations $\tau$ and $\sigma$:
\begin{equation}
N_{\sigma,\tau} = 
\begin{tikzpicture}[baseline=-1mm,scale=0.55]
\draw[thick,rounded corners] (0,0.6) -- (2.6,0.6) -- (2.6,1.35) -- (-1.225,1.35) -- (-1.225,0.6) -- (0,0.6);
\draw[thick,rounded corners] (0,-0.6) -- (2.6,-0.6) -- (2.6,-1.35) -- (-1.225,-1.35) -- (-1.225,-0.6) -- (0,-0.6);
\draw[rounded corners,fill=white] (-0.85,-1) rectangle (0.85,1);
\node at (0,0) {$\bar{M}_\Phi^{\otimes k}$};
\draw[rounded corners,fill=white] (1.225,0.15) rectangle (2.225,1.05);
\node at (1.725,0.6) {$\sigma$};
\draw[rounded corners,fill=white] (1.225,-0.15) rectangle (2.225,-1.05);
\node at (1.725,-0.6) {$\tau$};
\end{tikzpicture}\,.
\end{equation}
We can without loss restrict summation to tensor networks without self-contractions, because $\Tr (\bar{M}_\Phi)=0$ ensures that such contributions vanish identically. 
Next, we apply a powerful general bound to individual tensor networks. \cite[Proposition~18]{kliesch_guaranteed_2017} states that the value of a tensor network (without self-contractions) is bounded by the product of 2-norms of the individual constituents. For any $\sigma,\tau$ this implies
\begin{equation}
|N_{\sigma,\tau}|= \left| N_{\sigma,\tau} \left( \bar{M}_\Phi, \Tr_2 (\bar{M}_\Phi), \Tr_1 (\bar{M}_\Phi)\right) \right|
\leq \| \bar{M}_{\Phi} \|_2^{\nu_1} \| \Tr_2 (\bar{M}_\Phi) \|_2^{\nu_2} \| \Tr_1 (\bar{M}_\Phi) \|_2^{\nu_3}\;,
\end{equation}
where $\nu_1,\nu_2,\nu_3 \in [k]$ denote the number of times each basic building block occurred in the network. Clearly, $\nu_1 + \nu_2 + \nu_3 = k$ and we can combine this with Eq.\ \eqref{eq:norm_bound_technical} to conclude
\begin{equation}
|N_{\sigma,\tau}|\leq d^{\nu_1/2} d^{\nu_2/2} d^{\nu_3} = d^{k/2+\nu_3/2}.
\end{equation}
The final contribution $d^{\nu_3/2}$ is always counter-balanced by the Weingarten function, \ie the dangerous terms are always suppressed by powers of $1/d$. As we discussed, the Weingarten functions $\Wg(\sigma, d)$ only depend on the cycle type of the permutation $\sigma$. The asymptotic behavior is $\Wg(\sigma,d) \sim 1/d^{2k-\ell(\sigma)}$, where $\ell$ is the length of the cycle type, \ie the number of cycles in the permutation. The leading order terms are those for which the cycle type is $(1,1,\ldots,1)$, the partition of $2k$ into 1's. For $\Wg(\sigma^{-1}\tau,d)$ this corresponds to terms with $\sigma=\tau$. Returning to the problem at hand, we contract the upper indices of $\bar M_\Phi$ with respect to $\sigma$ and the lower indices with $\tau$. The leading order terms, are the terms for which we act on upper and lower indices the same. In order to generate terms in the tensor network contraction of $M$'s containing a dangerous contribution $\Tr_1 (\bar{M}_\Phi)$ , the lengths of the cycle types of the two permutations must differ by at least one, in order to generate a contraction, a length one cycle, upstairs: 
\begin{align}
\mathbb{E}_U \left[ \bar{S}_U (M,\phi)^k \right]
&\leq \sum_{\tau,\sigma \in S_k} | \Wg (\sigma^{-1}\tau,d) | N_{\sigma,\tau}\nn
&\leq \sum_{\sigma \in S_k} \Wg((1,\ldots,1),d) d^{k/2} + \sum_{\tau\neq \sigma \in S_k} \Wg(\sigma^{-1}\tau,d)d^{k/2+\nu_3/2}\,.
\end{align}

Although, the $\Tr_1 (\bar{M}_\Phi)$ terms will only contribute at subleading order, they appear with a larger contribution in powers of $d$. Thus, to rigorously upper bound the expression, we need bounds on the Weingarten functions as well as on the number of terms $\nu_3$ which appear in a given tensor network $N_{\sigma, \tau}$.

Precise upper bounds on the Weingarten functions are known \cite{Collins04,CollinsMat17}. For our purposes, it will be convenient to use the (slightly weaker) bound in \cite{montanaro2013weak}, which states that for $k\leq d^{2/3}$
\begin{equation}
|\Wg(\sigma,d)| \leq \frac{3}{2} \frac{C_{k-1}}{d^{2k-\ell(\sigma)}}
\label{eq:Wgbound}\,,
\end{equation}
where $C_k$ is the $k$-th Catalan number.

Now we establish that $\nu_3(\sigma,\tau)$, the number of dangerous terms $\Tr_1(\bar M_\Phi)$ terms in a given $N_{\sigma,\tau}$, is bounded by the distance between the permutations $\sigma$ and $\tau$ as $\nu_3(\sigma,\tau)\leq 2d(\sigma,\tau)$. First we note a few facts about the symmetric group. $d(\sigma,\tau)$ is defined as the minimal number of transpositions needed to take $\sigma$ to $\tau$, and defines a distance between the permutations. Specifically, $d(\sigma,\tau)$ is a metric on the Cayley graph of the symmetric group with the generating set of transpositions. The length of the cycle type of a permutation $\sigma \in S_k$ is related to the number of transpositions needed to build $\sigma$ from the identity permutation as $\ell(\sigma) = k-d(\sigma,\iden)$. Furthermore, a transposition changes the number of cycles in a permutation by exactly one.

The terms $\Tr_1(\bar M_\Phi)$ only appear when the permutation $\sigma$ has a fixed point where $\tau$ does not, \ie there is a contraction on the upstairs indices of $\bar M_\Phi$, meaning $\sigma$ has a length one cycle at a point where $\tau$ does not.
As $d(\sigma,\tau)$ is the minimal number of transpositions required to take $\sigma$ to $\tau$, and a transposition can only change the number of cycles by exactly 1, then for every two dangerous terms the distance between the permutations $\sigma$ and $\tau$ must increase by at least one. This shows that $\nu_3(\sigma,\tau)$ is bounded as
\begin{equation}
\nu_3(\sigma,\tau) \leq 2 d(\sigma,\tau) = 2(k-\ell(\sigma^{-1}\tau))\,.
\end{equation}

Returning to the general moment bound, we can apply the bound on Weingarten functions in Eq.~\eqref{eq:Wgbound} and the bound on $\nu_3$ to show that
\begin{align}
\mathbb{E}_U \left[ \bar{S}_U (M,\phi)^k \right]
&\leq \sum_{\tau,\sigma \in S_k} | \Wg (\sigma^{-1}\tau,d) | N_{\sigma,\tau} \leq \sum_{\tau, \sigma \in S_k} \frac{3}{2} C_{k-1} d^{\ell(\sigma^{-1}\tau) - 2k + k/2 +\nu_3/2}\nn
&\leq \sum_{\tau, \sigma \in S_k} \frac{3}{2} C_{k-1} d^{-k/2} \leq C_{k} (k!)^2 d^{-k/2}\,,
\end{align}
which establishes the claim. 

\end{proof}

\subsection{\ensuremath{\varepsilon}-coverings of local random circuits}
\label{sec:epcover}
We would like to extend our results in Sec.~\ref{sub:local} on complexity growth to local random circuits, where the gates are chosen Haar-randomly from $U(q^2)$. Obviously, the ensemble of size $T$ circuits is continuous and statements about the number of states of a certain complexity become less meaningful. Nevertheless, we can consider an $\varepsilon$-covering of the ensemble of local RQCs in order to make concrete statements about complexity growth. 

We say that a set of unitaries $\mathsf{V}$ is an $\varepsilon$-covering of a set of unitaries $\mathsf{U}$ if for all $U\in \mathsf{U}$ there is some $V\in \mathsf{V}$ such that $\|U(\cdot)U^\dagger - V(\cdot)V^\dagger\|_\diamond \leq \varepsilon$.

Consider the set of local random circuits of size $T$, where again we act on $n$ local qudits with local dimension $q$ and with local gates chosen Haar-randomly from $U(q^2)$. Following Lemma 27 from \cite{brandao_local_2016}, we can bound the size of an $\varepsilon$-covering of the set $\CE_{\rm RQC}$ size $T$ local RQCs. Approximating each local gate to accuracy $\varepsilon/T$, we construct a covering in diamond norm of each gate with size $\leq \big(10 T/\varepsilon\big)^{q^4}$. For the $n^T$ choices of gates in the circuit, we conclude that there exists an $\varepsilon$-covering $\widetilde \CE_{\rm RQC} $ of size $T$ RQCs with cardinality
\begin{equation}
|\widetilde \CE_{\rm RQC}| \leq n^T \Big(\frac{10T}{\varepsilon}\Big)^{T q^4}\,.
\label{eq:epnet}
\end{equation}

Furthermore, if an ensemble $\CE$ forms an $\epsilon$-approximate unitary $k$-design, then the $\varepsilon$-covering of $\CE$ will form an $\epsilon'$-approximate unitary design with $\epsilon' = \epsilon + 2d^{2k}\varepsilon$ (from Prop.~8 in \cite{brandao_local_2016}). Using the lower bound on the cardinality of an approximate design in \autoref{lem:weights-intro} and the upper bound on the cardinality of an $\varepsilon$-covering of size $T$ local random circuits in Eq.~\ref{eq:epnet}, means that for $\tilde \CE_{\rm RQC} $ to form an approximate design, we must have
\begin{equation}
\frac{1}{1+\epsilon'}\frac{d^{2k}}{k!} \leq |\widetilde \CE_{\rm RQC}| \leq n^T \Big(\frac{10T}{\varepsilon}\Big)^{T q^4}\,.
\end{equation}
This gives a lower bound on the size for local random circuits to form $k$-designs
\begin{equation}
T \geq \frac{2kn \log q}{q^4 \log k}\,.
\end{equation}
Therefore, an optimal random circuit implementation of a unitary design will have at least an essentially linear scaling in both $n$ and $k$.

\subsubsection*{Acknowledgments}
\addcontentsline{toc}{section}{\protect\numberline{}Acknowledgments}
The authors would like thank Dorit Aharonov, Thom Bohdanowicz, Elizabeth Crosson, Felix Haehl, Aram Harrow, Tomas Jochym-O'Connor, Hugo Marrochio, Grant Salton, Eugene Tang, and Thomas Vidick for inspiring discussions and valuable feedback.
All authors acknowledge funding provided by the Institute for Quantum Information and Matter, an NSF Physics Frontiers Center (NSF Grant {PHY}-{1733907}). JP is supported in part by DOE Award {DE}-{SC0018407} and by the Simons Foundation It from Qubit Collaboration.
RK is supported in part by the Office of Naval Research (Award N00014-17-1-2146) and the Army Research Office (Award W911NF121054). NHJ would like to thank the IQIM at Caltech, McGill University, and UC Berkeley for their hospitality during the completion of this work.
Research at Perimeter Institute is supported by the Government of Canada through the Department of Innovation, Science and Economic Development Canada and by the Province of Ontario through the Ministry of Research, Innovation and Science.

\appendix

\section{Concentration of measure for Haar-uniform vectors}

\begin{proposition} \label{prop:measure_concentration}
Fix $M \in \mathbb{H}_d$ with $\|M \|_\infty \leq 1$ and suppose that $|\psi \rangle \in \mathbb{C}^d$ is chosen uniformly from the complex unit sphere. Then,
\begin{equation}
\mathrm{Pr} \left[ \left|\langle \psi| M | \psi \rangle - \mathbb{E} \left[ \langle \psi| M | \psi \rangle \right] \right| \geq \tau \right] \leq 2\exp \left( - \frac{d\tau^2}{9 \pi^3} \right)\quad \textrm{for any} \quad \tau \geq 0\,.
\end{equation}
\end{proposition}

The proof is standard and we include it in this appendix for completion. It is based on Levy's Lemma, i.e.\ concentration of measure on the real-unit sphere $\mathbb{S}^{2d-1} \subset \mathbb{R}^{2d}$. A function $f: \mathbb{S}^{2d-1} \to \mathbb{R}$ is $L$-Lipschitz (with respect to the Euclidean norm $\| \cdot \|_{2}$ on $\mathbb{R}^{2d}$) if
\begin{equation}
\left| f(x) - f (y) \right| \leq L \| x - y \|_{2} \quad \textrm{for all} \quad x,y \in \mathbb{S}^{2d-1}\,.
\end{equation}

\begin{theorem}[Levy's Lemma] \label{thm:levi}
Let $f: \mathbb{S}^{2d-1} \to \mathbb{R}$ be a $L$-Lipschitz function on the unit sphere. Then, the following relation is true if $x$ is chosen uniformly from $\mathbb{S}^{2d-1}$:
\begin{equation}
\mathrm{Pr} \left[ \left| f(x) - \mathbb{E} \left[ f(x) \right] \right| \geq \tau \right] \leq 2 \exp \left( - \frac{4d \tau^2}{9 \pi^3 L^2} \right)\,.
\end{equation}
\end{theorem}

\begin{proof}[Proof of \autoref{prop:measure_concentration}]
The complex unit sphere in $d$-dimensions admits an isometric embedding -- with respect to the Euclidean norm --  onto the real-valued unit sphere $\mathbb{S}^{2d-1} \subseteq \mathbb{R}^{2d}$:
\begin{equation}
| \psi \rangle \mapsto |x \rangle =\mathrm{Re}(| \psi \rangle) \rangle \oplus  \mathrm{Im} (| \psi \rangle) \in \mathbb{S}^{2d-1}\,.
\label{eq:embedding}
\end{equation}
This embedding maps the uniform distribution on the complex unit sphere in $\mathbb{C}^d$
to the uniform distribution on the real-valued unit sphere in $\mathbb{R}^{2d}$.
Under this embedding, the function of interest $\langle \psi| M | \psi \rangle$ becomes
\begin{equation}
\langle \psi| M | \psi \rangle = \langle \mathrm{Re} (\psi) |M| \mathrm{Re}(\psi) \rangle + \langle \mathrm{Im}(\psi) |M| \mathrm{Im} (\psi) \rangle
= \langle x| M \oplus M |x \rangle\,,
\end{equation}
because $M$ is Hermitian. Its expectation is also preserved and \autoref{lem:lipschitz} below states that this function is Lipschitz with constant $2 \| M \|_\infty \leq 2$.
The claim then readily follows from Levy's lemma  (\autoref{thm:levi}).
\end{proof}

\begin{lemma} \label{lem:lipschitz}
Fix $M \in \mathbb{H}_d$. Then, the following relation is true for any pair of unit-norm vectors $x,y \in\mathbb{S}^{2d-1} \subset \mathbb{R}^{2d}$
\begin{equation}
\left| \langle x| M \oplus M |x \rangle - \langle y| M \oplus M | y \rangle \right| \leq 2 \| M \|_\infty \| x -y \|_{\ell_2}\,.
\end{equation}
\end{lemma}

\begin{proof}
Fix $x,y \in \mathbb{S}^{2d-1}$ and apply Hoelder's inequality:
\begin{equation}
\left| \langle x| M\!\oplus\! M |x \rangle - \langle y| M\!\oplus\! M| y \rangle \right|^2 =  \Tr \left( M\!\oplus\! M (\ketbra{x} - \ketbra{y})\right)^2
\leq \| M \oplus M \|_\infty^2 \| \ketbra{x} - \ketbra{y} \|_1^2\,.
\end{equation}
The block-structure of $M \oplus M$ ensures $\|M \oplus M \|_\infty = \|M \|_\infty$, while the remaining term is the trace norm of a difference of pure states. 
This can be computed analytically and we obtain
\begin{equation}
\| |x \rangle \! \langle x| - |y \rangle \! \langle y| \|_1^2
= 4 \left( 1-  \langle x,y \rangle^2 \right)
= 4 \left(1+ \langle x,y \rangle \right) \left( 1-  \langle x,y \rangle \right)
\leq 4 \left( 2- 2| \langle x,y \rangle| \right)\,,
\end{equation}
because $\langle x,y \rangle \leq \|x \|_{\ell_2} \| y \|_{\ell_2} \leq 1$
Finally,
\begin{equation}
2 - 2 \langle x,y \rangle = \langle x,x \rangle - \langle x,y \rangle - \langle y,x \rangle + \langle y, y \rangle = \langle x-y,x-y \rangle = \| x-y \|_{2}^2
\end{equation}
and the claim follows.
\end{proof}

\section{Designs and the traditional definition of complexity}
In the bulk of the paper we focused on a stronger notion of complexity than the standard definition, an operational definition involving the complexity of the distinguishing measurement to differentiate the state from the maximally mixed state. A more traditional definition is often considered in the literature, which involves building a quantum circuit which approximates the state when evolved from an initial state. This intuitive notion of complexity is related to the minimal size of such a circuit.

In this appendix, we will work through the counting arguments in Sec.~\ref{sec:results} for the complexity of elements of a $k$-design using the more traditional (albeit weaker) definition of complexity. We will refer to this as the {\it weak complexity} of a state or unitary to distinguish it from the operational definitions presented in Sec.~\ref{sec:defcomp}.

Consider a system of $n$ qudits with local dimension $q$, such that the total dimension is $d = q^n$. Let $\mathsf{G} \subset U(q^2)$ denote a universal gate set of elementary 2-local gates, and let $\mathsf{G}_r$ be the set of circuits of size $r$ built from our gate set $\mathsf{G}$. 

\begin{definition}[weak $\delta$-state complexity] \label{def:weak_state}
For $\delta \in [0,1]$, we say that a state $\ket\psi$ has $\delta$-state complexity of at most $r$ if there exists a unitary circuit $V\in \mathsf{G}_r$ such that
\begin{equation*}
\frac{1}{2} \big\| \ketbra{\psi} - V \ketbra{0} V^\dagger \big\|_1 \leq \delta\,, \quad \textrm{which we denote as} \quad \CC'_\delta(\ket\psi) \leq r\,.
\end{equation*}
\end{definition}

We want to be able to make precise statements about the complexity of sets of states. More specifically, if we consider a complex projective design, the requirement that they form a $k$-design is sufficiently restrictive to deduce a quantitative statement about the complexity of the constituent states. 

\begin{theorem}[weak complexity of state designs] \label{thm:weak_state}
Consider an $\epsilon$-approximate complex projective $k$-design $\CE = \{p_i, \ket{\psi_i}\}_{i=1}^N$. Then there are at least
\begin{equation}
\frac{d^k}{k!} \frac{1}{1+\epsilon} - \frac{n^r |\mathsf{G}|^r}{(1-\delta^2)^k}
\end{equation}
states with weak $\delta$-state complexity  $\CC'_\delta (\ket{\psi_i})>r$.
\end{theorem}
The number of high complexity states is exponentially large in $k$ for complexity 
\begin{equation}
r \lesssim \frac{k(n-\log k)}{\log n}\,.
\end{equation}

Turning now to the complexity of unitaries, the traditional definition of complexity is the minimal size of a circuit, built from our gate set, which approximates that unitary.

\begin{definition}[weak $\delta$-unitary complexity] \label{def:weak_unitary}
For $\delta \in [0,1]$, we say that a unitary $U$ has $\delta$-unitary complexity of at most $r$ if there exists a circuit $V\in \mathsf{G}_r$ such that
\begin{equation*}
\frac{1}{2} \big\| \mathcal{U} - \mathcal{V} \big\|_\diamond \leq \delta\,, \quad \textrm{which we denote as}\quad \CC'_\delta(U) \leq r\,,
\end{equation*}
where $\mathcal{U}(\rho) = U\rho U^\dagger$ and $\mathcal{V}(\rho) = V\rho V^\dagger$.
\end{definition}

Again, we ask if the structure of a unitary $k$-design allows us to conclude anything about the complexity of unitaries. Once more, we find that we can turn the statement that $k$-design elements have a certain expected complexity into a quantitative one.

\begin{theorem}[weak complexity of unitary designs] \label{thm:weak_unitary}
Consider an $\epsilon$-approximate unitary $k$-design $\CE = \{p_i, U_i\}_{i=1}^N$. Then there are at least
\begin{equation}
\frac{d^{2k}}{k!} \frac{1}{1+\epsilon} - \frac{n^r |\mathsf{G}|^r}{(1-\delta^2)^k}
\end{equation}
unitaries in $\CE$ with weak $\delta$-unitary complexity   $\CC'_\delta (U_i)>r$.
\end{theorem}
The number of high complexity unitaries is again exponentially large in $k$ for complexity less than
\begin{equation}
r \lesssim \frac{k(2n-\log k)}{\log n}\,.
\end{equation}

We now provide details and proofs of the above statements about complexity of spherical and unitary designs. 

\subsection{Weak state complexity for spherical designs}
\begin{proof}[Proof of \autoref{thm:weak_state}]
First, as stated in \autoref{lem:state_comp_stronger}, we note that the definition of weak $\delta$-state complexity in \autoref{def:weak_state} is equivalently written as
\begin{equation}
|\vev{\psi|V|0}|^2 \geq 1-\delta^2\,.
\label{eq:equivcond}
\end{equation}
We can show this by first noting that $X := \ketbra{\psi} - V\ketbra{0}V^\dagger$ has rank at most two. Directly computing the eigenvalues of $X$ from
\begin{equation}
\Tr(X) = \lambda_1+\lambda_2 = 0 \and \Tr(X^2) = \lambda_1^2+\lambda_2^2 = 2 - 2|\vev{\psi|V|0}|^2\,,
\end{equation}
we find $\lambda_{1,2} = \pm \sqrt{1-|\vev{\psi|V|0}|^2}$. Then as $\|X\|_1 = |\lambda_1|+|\lambda_2|$ we have that
\begin{equation}
\frac{1}{2} \big\| \ketbra{\psi} - V\ketbra{0}V^\dagger \big\|_1 = \sqrt{1-|\vev{\psi|V|0}|^2}\,,
\end{equation}
from which the claim follows. 

We want to ask, given some state $\ket\psi$ chosen uniformly from an $\epsilon$-approximate spherical $k$-design, what is the probability that the state has $\delta$-complexity at most $r$: $\CC'_\delta(\ket\psi)\leq r$? We know that the state will have $\delta$-complexity $r$ if there exists a $V\in \mathsf{G}_r$  such that Eq.~\ref{eq:equivcond} holds. A union bound then gives that
\begin{equation}
\Pr\big[ \CC'_\delta(\ket\psi) \leq r\big] = \Pr \left[ \bigcup_{V\in \mathsf{G}_r} \left\lbrace |\vev{\psi|V|0}|^2 \geq 1-\delta^2 \right\rbrace \right] \leq \sum_{V\in \mathsf{G}_r} \Pr\Big[ |\vev{\psi|V|0}|^2 \geq 1-\delta^2\Big]\,. 
\label{eq:scompbound}
\end{equation}
We can bound the probability that a state drawn from a spherical $k$-design satisfies Eq.~\ref{eq:equivcond} as a straightforward consequence of Markov's inequality:
\begin{align}
\Pr\Big[ |\vev{\psi|V|0}|^2 \geq 1-\delta^2\Big] &= \Pr\Big[ |\vev{\psi|V|0}|^{2k} \geq \big(1-\delta^2\big)^k\Big]\nn
&\leq \frac{\EE_{\ket\psi} \big[ |\vev{\psi|V|0}|^{2k}\big]}{(1-\delta^2)^k} \leq \frac{ (1+\epsilon) \binom{d+k-1}{k}^{-1}}{(1-\delta^2)^k}\,.
\end{align}
In the last step here, we use Eq.~\eqref{eq:state_moment} and proceeding similarly as in the proof of \autoref{lem:state_weights}, noting that for a fixed state $\ket\phi$ and $\ket\psi$ averaged over an $\epsilon$-approximate spherical $k$-design, we have
\begin{equation}
\EE_{\ket\psi} \big[ |\vev{\psi|\phi}|^{2k}\big] \leq (1+\epsilon) \binom{d+k-1}{k}^{-1}\,.
\end{equation}
This claim readily follows from an argument similar to the proof of \autoref{lem:state_weights}.
Returning to Eq.~\ref{eq:scompbound}, we find that the probability that a state in a spherical design has complexity of at most $r$ is
\begin{equation}
\Pr\big[ \CC'_\delta(\ket\psi) \leq r\big] \leq (1+\epsilon) \binom{d+k-1}{k}^{-1} \frac{n^r |\mathsf{G}|^r}{(1-\delta^2)^k}\,,
\end{equation}
using the bound on the expectation and a bound on the cardinality of $\mathsf{G}_r$. 

We now turn to proving the primary claim. Negating the above assertion implies that
\begin{equation}
\Pr\big[ \CC'_\delta(\ket\psi) > r\big] \geq 1- (1+\epsilon) \binom{d+k-1}{k}^{-1} \frac{n^r |\mathsf{G}|^r}{(1-\delta^2)^k}\,.
\end{equation}
Furthermore, we may also write this probability as the expectation of the associated event, which yields
\begin{align}
\Pr\big[ \CC'_\delta(\ket\psi) > r\big] &= \EE_{\ket\psi} \left[ \1den \big\lbrace \CC'_\delta(\ket\psi) > r \big\rbrace \right] = \sum_{i=1}^N p_i \,\1den \big\lbrace \CC'_\delta(\ket\psi) > r \big\rbrace\nn
&\leq (1+\epsilon) \binom{d+k-1}{k}^{-1} N\,,
\end{align}
where $\1den$ is the indicator function, and in the last step we use the bound on the weights of an $\epsilon$-approximate spherical $k$-design in \autoref{lem:state_weights}. $M$ denotes the number of states in the spherical design $\ket{\psi_i}$ with weak $\delta$-complexity greater than $r$. Combining the previous two equations, we find that
\begin{equation}
N \geq \frac{d^k}{k!} \frac{1}{1+\epsilon} - \frac{n^r |\mathsf{G}|^r}{(1-\delta^2)^k}\,,
\end{equation}
which completes the proof. 
\end{proof}

\subsection{Weak unitary complexity for unitary designs}
\begin{proof}[Proof of \autoref{thm:weak_unitary}]
We start by noting an equivalent definition of weak $\delta$-unitary complexity as shown in the proof of \autoref{lem:unitary_complexity_stronger}. A necessary, but in general not sufficient, condition for weak unitary complexity in \autoref{def:weak_unitary} is
\begin{equation}
\big| \Tr(V^\dagger U)\big|^2 \geq d^2 (1-\delta^2)\,.
\end{equation}
Now we again ask, given some unitary $U$ chosen uniformly from an $\epsilon$-approximate unitary $k$-design, what is the probability that it has $\delta$-unitary complexity at most $r$: $\CC'_\delta(U)\leq r$? As this holds if there exists a $V\in \mathsf{G}_r$ such that the channels are close in diamond distance, a union bound then gives that
\begin{align}
\Pr\big[ \CC'_\delta(U) \leq r\big] &= \Pr \left[ \bigcup_{V\in G_r} \left\lbrace \frac{1}{2}\big\|\mathcal{U} - \mathcal{V}\big\|_\diamond \leq \delta \right\rbrace \right]\nn
&\leq \sum_{V\in G_r} \Pr\Big[ \big| \Tr(V^\dagger U)\big|^2 \geq d^2 (1-\delta^2)\Big]\,,
\label{eq:ucompbound}
\end{align}
using the reformulation above. We can bound the probability that a unitary drawn from a $k$-design satisfies this condition again by using Markov's inequality:
\begin{align}
\Pr\Big[ \big| \Tr(V^\dagger U)\big|^2 \geq d^2 (1-\delta^2)\Big] &= \Pr\Big[ \big| \Tr(V^\dagger U)\big|^{2k} \geq d^{2k} (1-\delta^2)^k\Big]\nn
&\leq \frac{\EE_{\CE} \big[ \big| \Tr(V^\dagger U)\big|^{2k}}{d^{2k}(1-\delta^2)^k} \leq \frac{(1+\epsilon) k!}{d^{2k}(1-\delta^2)^k}\,,
\end{align}
where in the last step, we use the moments of traces for unitary designs and as in \autoref{lem:weights-restatement} find that for a fixed unitary $V$ and a unitary $U$ averaged over an $\epsilon$-approximate unitary $k$-design, we have
\begin{equation}
\EE_\CE \big[ \Tr(V^\dagger U)\big|^{2k}\big] \leq (1+\epsilon) k!\,.
\end{equation}
Returning to the expression above in Eq.~\ref{eq:ucompbound}, we find that the probability $\CC'_\delta(U)\leq r$ is
\begin{equation}
\Pr\big[ \CC'_\delta(U) \leq r\big] \leq (1+\epsilon) \frac{k!}{d^{2k}} \frac{n^r |\mathsf{G}|^r}{(1-\delta^2)^k}\,,
\end{equation}
using the bound on the expectation and a bound on the cardinality of $\mathsf{G}_r$. Negating the expression gives a lower bound on the probability that a unitary in a $k$-design has complexity greater than $r$. Furthermore, we may also write this probability as the expectation
\begin{align}
\Pr\big[ \CC'_\delta(U) > r\big] &= \sum_{i=1}^N p_i \,\1den \big\lbrace \CC'_\delta(U_i) > r \big\rbrace \leq (1+\epsilon) \frac{k!}{d^{2k}} N\,,
\end{align}
where we use the bound on the unitary design weights in \autoref{lem:weights-restatement}. $M$ denotes the number of untiaries in a $k$-design with weak $\delta$-complexity greater than $r$. Combining the previous two equations, we find that
\begin{equation}
N \geq \frac{d^{2k}}{k!} \frac{1}{1+\epsilon} - \frac{n^r |\mathsf{G}|^r}{(1-\delta^2)^k}\,,
\end{equation}
which completes the proof. 
\end{proof}

\bibliographystyle{utphys}
\bibliography{complexity_growth}

\providecommand{\href}[2]{#2}\begingroup\raggedright\begin{thebibliography}{10}

\bibitem{poulin2011quantum}
D.~Poulin, A.~Qarry, R.~Somma, and F.~Verstraete, ``{Quantum Simulation of
  Time-Dependent Hamiltonians and the Convenient Illusion of Hilbert Space},''
  \href{http://dx.doi.org/10.1103/PhysRevLett.106.170501}{{\em Phys. Rev.
  Lett.} {\bfseries 106} (2011) 170501},
  \href{http://arxiv.org/abs/1102.1360}{{\ttfamily arXiv:1102.1360
  [quant-ph]}}.

\bibitem{bernstein_complexity_1997}
E.~Bernstein and U.~Vazirani, ``Quantum complexity theory,''
  \href{http://dx.doi.org/10.1137/S0097539796300921}{{\em SIAM J. Comput.}
  {\bfseries 26} (1997) 1411}.

\bibitem{Chen2010}
X.~Chen, Z.~C. Gu, and X.~G. Wen, ``{Local unitary transformation, long-range
  quantum entanglement, wave function renormalization, and topological
  order},'' \href{http://dx.doi.org/10.1103/PhysRevB.82.155138}{{\em Phys.
  Rev.} {\bfseries B82} (2010) 155138},
\href{http://arxiv.org/abs/1004.3835}{{\ttfamily arXiv:1004.3835
  [cond-mat.str-el]}}.

\bibitem{SusskindCCBH14}
L.~Susskind, ``{Computational Complexity and Black Hole Horizons},''
  \href{http://dx.doi.org/10.1002/prop.201500092}{{\em Fortsch. Phys.}
  {\bfseries 64} (2016) 44}, \href{http://arxiv.org/abs/1402.5674}{{\ttfamily
  arXiv:1402.5674 [hep-th]}}.
  \href{http://dx.doi.org/10.1002/prop.201500093}{[Fortsch. Phys.64,24(2016)]}.

\bibitem{SScomp14}
D.~Stanford and L.~Susskind, ``{Complexity and Shock Wave Geometries},''
  \href{http://dx.doi.org/10.1103/PhysRevD.90.126007}{{\em Phys. Rev.}
  {\bfseries D90} (2014) 126007},
\href{http://arxiv.org/abs/1406.2678}{{\ttfamily arXiv:1406.2678 [hep-th]}}.

\bibitem{CABH15}
A.~R. Brown, D.~A. Roberts, L.~Susskind, B.~Swingle, and Y.~Zhao,
  ``{Complexity, action, and black holes},''
  \href{http://dx.doi.org/10.1103/PhysRevD.93.086006}{{\em Phys. Rev.}
  {\bfseries D93} (2016) 086006},
\href{http://arxiv.org/abs/1512.04993}{{\ttfamily arXiv:1512.04993 [hep-th]}}.

\bibitem{brown_second_2018}
A.~R. Brown and L.~Susskind, ``{Second law of quantum complexity},''
  \href{http://dx.doi.org/10.1103/PhysRevD.97.086015}{{\em Phys. Rev.}
  {\bfseries D97} (2018) 086015},
\href{http://arxiv.org/abs/1701.01107}{{\ttfamily arXiv:1701.01107 [hep-th]}}.

\bibitem{susskind2018black}
L.~Susskind, ``{Black Holes and Complexity Classes},''
\href{http://arxiv.org/abs/1802.02175}{{\ttfamily arXiv:1802.02175 [hep-th]}}.

\bibitem{aaronson2016complexity}
S.~Aaronson, ``{The Complexity of Quantum States and Transformations: From
  Quantum Money to Black Holes},''
\href{http://arxiv.org/abs/1607.05256}{{\ttfamily arXiv:1607.05256
  [quant-ph]}}.

\bibitem{bohdanowicz2017universal}
T.~C. Bohdanowicz and F.~G. S.~L. Brand\~ao, ``{Universal Hamiltonians for
  Exponentially Long Simulation},''
\href{http://arxiv.org/abs/1710.02625}{{\ttfamily arXiv:1710.02625
  [quant-ph]}}.

\bibitem{ChaosDesign}
D.~A. Roberts and B.~Yoshida, ``{Chaos and complexity by design},''
  \href{http://dx.doi.org/10.1007/JHEP04(2017)121}{{\em JHEP} {\bfseries 04}
  (2017) 121},
\href{http://arxiv.org/abs/1610.04903}{{\ttfamily arXiv:1610.04903
  [quant-ph]}}.

\bibitem{brandao_local_2016}
F.~G.~S.~L. {Brand{\~a}o}, A.~W. {Harrow}, and M.~{Horodecki}, ``{Local Random
  Quantum Circuits are Approximate Polynomial-Designs},''
  \href{http://dx.doi.org/10.1007/s00220-016-2706-8}{{\em Commun. Math. Phys.}
  {\bfseries 346} (2016) 397}, \href{http://arxiv.org/abs/1208.0692}{{\ttfamily
  arXiv:1208.0692 [quant-ph]}}.

\bibitem{Dankert09}
C.~Dankert, R.~Cleve, J.~Emerson, and E.~Livine, ``{Exact and approximate
  unitary 2-designs and their application to fidelity estimation},''
  \href{http://dx.doi.org/10.1103/PhysRevA.80.012304}{{\em Phys. Rev.}
  {\bfseries A80} (2009) 012304},
  \href{http://arxiv.org/abs/quant-ph/0606161}{{\ttfamily
  arXiv:quant-ph/0606161}}.

\bibitem{gross_evenly_2007}
D.~Gross, K.~Audenaert, and J.~Eisert, ``Evenly distributed unitaries: On the
  structure of unitary designs,''
  \href{http://dx.doi.org/10.1063/1.2716992}{{\em J. Math. Phys.} {\bfseries
  48} (2007) 052104}, \href{http://arxiv.org/abs/quant-ph/0611002}{{\ttfamily
  arXiv:quant-ph/0611002}}.

\bibitem{webb_clifford_2015}
Z.~Webb, ``{The Clifford group forms a unitary 3-design},'' {\em Quantum Info.
  Comput.} {\bfseries 16} (2016) 1379,
  \href{http://arxiv.org/abs/1510.02769}{{\ttfamily arXiv:1510.02769
  [quant-ph]}}.

\bibitem{zhu_multiqubit_2017}
H.~Zhu, ``Multiqubit clifford groups are unitary 3-designs,''
  \href{http://dx.doi.org/10.1103/PhysRevA.96.062336}{{\em Phys. Rev. A}
  {\bfseries 96} (2017) 062336},
  \href{http://arxiv.org/abs/1510.02619}{{\ttfamily arXiv:1510.02619
  [quant-ph]}}.

\bibitem{Kueng15}
R.~{Kueng} and D.~{Gross}, ``{Qubit stabilizer states are complex projective
  3-designs},'' \href{http://arxiv.org/abs/1510.02767}{{\ttfamily
  arXiv:1510.02767 [quant-ph]}}.

\bibitem{Ambainis2007}
A.~{Ambainis} and J.~{Emerson}, ``{Quantum t-designs: t-wise Independence in
  the Quantum World},'' \href{http://dx.doi.org/10.1109/CCC.2007.26}{{\em
  Twenty-Second Annual IEEE Conference on Computational Complexity (CCC'07)}
  (2007) 129}, \href{http://arxiv.org/abs/quant-ph/0701126}{{\ttfamily
  arXiv:quant-ph/0701126}}.

\bibitem{szehr2013}
O.~Szehr, F.~Dupuis, M.~Tomamichel, and R.~Renner, ``Decoupling with unitary
  approximate two-designs,''
  \href{http://dx.doi.org/10.1088/1367-2630/15/5/053022}{{\em New J. Phys.}
  {\bfseries 15} (2013) 053022},
  \href{http://arxiv.org/abs/1109.4348}{{\ttfamily arXiv:1109.4348
  [quant-ph]}}.

\bibitem{scott_tight_2006}
A.~J. {Scott}, ``{Tight informationally complete quantum measurements},''
  \href{http://dx.doi.org/10.1088/0305-4470/39/43/009}{{\em J. Phys. A: Math.
  Gen.} {\bfseries 39} (2006) 13507},
  \href{http://arxiv.org/abs/quant-ph/0604049}{{\ttfamily
  arXiv:quant-ph/0604049}}.

\bibitem{kueng_low_2017}
R.~{Kueng}, H.~{Rauhut}, and U.~{Terstiege}, ``Low rank matrix recovery from
  rank one measurements,''
  \href{http://dx.doi.org/10.1016/j.acha.2015.07.007}{{\em Appl. Comput.
  Harmon. Anal.} {\bfseries 42} (2017) 88},
  \href{http://arxiv.org/abs/1410.6913}{{\ttfamily arXiv:1410.6913 [cs.IT]}}.

\bibitem{emerson05}
J.~Emerson, R.~Alicki, and K.~{\.{Z}}yczkowski, ``Scalable noise estimation
  with random unitary operators,''
  \href{http://dx.doi.org/10.1088/1464-4266/7/10/021}{{\em J. Opt. B: Quantum
  Semiclass. Opt} {\bfseries 7} (2005) S347},
  \href{http://arxiv.org/abs/quant-ph/0503243}{{\ttfamily
  arXiv:quant-ph/0503243}}.

\bibitem{HaydenPreskill}
P.~Hayden and J.~Preskill, ``{Black holes as mirrors: Quantum information in
  random subsystems},''
  \href{http://dx.doi.org/10.1088/1126-6708/2007/09/120}{{\em JHEP} {\bfseries
  09} (2007) 120},
\href{http://arxiv.org/abs/0708.4025}{{\ttfamily arXiv:0708.4025 [hep-th]}}.

\bibitem{FastScrambling}
N.~Lashkari, D.~Stanford, M.~Hastings, T.~Osborne, and P.~Hayden, ``{Towards
  the Fast Scrambling Conjecture},''
  \href{http://dx.doi.org/10.1007/JHEP04(2013)022}{{\em JHEP} {\bfseries 04}
  (2013) 022},
\href{http://arxiv.org/abs/1111.6580}{{\ttfamily arXiv:1111.6580 [hep-th]}}.

\bibitem{onorati_mixing_2017}
E.~Onorati, O.~Buerschaper, M.~Kliesch, W.~Brown, A.~H. Werner, and J.~Eisert,
  ``{Mixing properties of stochastic quantum Hamiltonians},''
  \href{http://dx.doi.org/10.1007/s00220-017-2950-6}{{\em Commun. Math. Phys.}
  {\bfseries 355} (2017) 905},
\href{http://arxiv.org/abs/1606.01914}{{\ttfamily arXiv:1606.01914
  [quant-ph]}}.

\bibitem{Nakata16}
Y.~Nakata, C.~Hirche, M.~Koashi, and A.~Winter, ``{Efficient Quantum
  Pseudorandomness with Nearly Time-Independent Hamiltonian Dynamics},''
  \href{http://dx.doi.org/10.1103/PhysRevX.7.021006}{{\em Phys. Rev.}
  {\bfseries X7} (2017) 021006},
\href{http://arxiv.org/abs/1609.07021}{{\ttfamily arXiv:1609.07021
  [quant-ph]}}.

\bibitem{NHJ19}
N.~Hunter-Jones, ``{Unitary designs from statistical mechanics in random
  quantum circuits},''
\href{http://arxiv.org/abs/1905.12053}{{\ttfamily arXiv:1905.12053
  [quant-ph]}}.

\bibitem{Holevo1973}
A.~S. Holevo, ``Optimal quantum measurements,'' {\em Teoret. Mat. Fiz.}
  {\bfseries 17} (1973) 319--326.

\bibitem{Helstrom1976}
C.~W. Helstrom, {\em Quantum Detection and Estimation Theory}.
\newblock Mathematics in Science and Engineering. Academic Press, 1976.

\bibitem{dawson_kitaev_2005}
C.~M. Dawson and M.~A. Nielsen, ``{The Solovay-Kitaev algorithm},'' {\em
  Quantum Info. Comput.} {\bfseries 6} (2006) 81,
  \href{http://arxiv.org/abs/quant-ph/0505030}{{\ttfamily
  arXiv:quant-ph/0505030}}.

\bibitem{Watrous2018}
J.~Watrous, {\em The theory of quantum information}.
\newblock Cambridge University Press, 2018.

\bibitem{HM18}
A.~{Harrow} and S.~{Mehraban}, ``{Approximate unitary $t$-designs by short
  random quantum circuits using nearest-neighbor and long-range gates},''
  \href{http://arxiv.org/abs/1809.06957}{{\ttfamily arXiv:1809.06957
  [quant-ph]}}.

\bibitem{HL09}
A.~W. Harrow and R.~A. Low, ``Efficient quantum tensor product expanders and
  k-designs,'' \href{http://dx.doi.org/10.1007/978-3-642-03685-9_41}{{\em
  Lecture Notes in Computer Science} {\bfseries 5687} (2009) 548},
  \href{http://arxiv.org/abs/0811.2597}{{\ttfamily arXiv:0811.2597
  [quant-ph]}}.

\bibitem{Fulton91}
W.~Fulton and J.~Harris,
  \href{http://dx.doi.org/10.1007/978-1-4612-0979-9}{{\em {Representation
  Theory: A First Course}}}.
\newblock Graduate Texts in Mathematics. Springer, New York, 1991.

\bibitem{christandl_phd_2006}
M.~Christandl, {\em The Structure of Bipartite Quantum States}.
\newblock PhD thesis, University of Cambridge, 2006.

\bibitem{Weingarten78}
D.~Weingarten, ``{Asymptotic Behavior of Group Integrals in the Limit of
  Infinite Rank},''
\href{http://dx.doi.org/10.1063/1.523807}{{\em J. Math. Phys.} {\bfseries 19}
  (1978) 999}.

\bibitem{Collins04}
B.~{Collins} and P.~{{\'S}niady}, ``{Integration with Respect to the Haar
  Measure on Unitary, Orthogonal and Symplectic Group},''
  \href{http://dx.doi.org/10.1007/s00220-006-1554-3}{{\em Commun. Math. Phys.}
  {\bfseries 264} (2006) 773},
  \href{http://arxiv.org/abs/math-ph/0402073}{{\ttfamily
  arXiv:math-ph/0402073}}.

\bibitem{bridgeman_handwaving_2017}
J.~C. Bridgeman and C.~T. Chubb, ``{Hand-waving and Interpretive Dance: An
  Introductory Course on Tensor Networks},''
  \href{http://dx.doi.org/10.1088/1751-8121/aa6dc3}{{\em J. Phys.} {\bfseries
  A50} (2017) 223001},
\href{http://arxiv.org/abs/1603.03039}{{\ttfamily arXiv:1603.03039
  [quant-ph]}}.

\bibitem{kliesch_guaranteed_2017}
M.~{Kliesch}, R.~{Kueng}, J.~{Eisert}, and D.~{Gross}, ``{Guaranteed recovery
  of quantum processes from few measurements},''
  \href{http://dx.doi.org/10.22331/q-2019-08-12-171}{{\em {Quantum}} {\bfseries
  3} (2019) 171}, \href{http://arxiv.org/abs/1701.03135}{{\ttfamily
  arXiv:1701.03135 [quant-ph]}}.

\bibitem{Roc70:Convex-Analysis}
R.~T. Rockafellar, {\em Convex analysis}, vol.~28 of {\em Princeton
  Mathematical Series}.
\newblock Princeton University Press, Princeton, NJ, 1970.

\bibitem{Bar02:Course-Convexity}
A.~Barvinok, \href{http://dx.doi.org/10.1090/gsm/054}{{\em {A Course in
  Convexity}}}, vol.~54 of {\em Graduate Studies in Mathematics}.
\newblock American Mathematical Society, Providence, RI, 2002.

\bibitem{ChaosRMT}
J.~Cotler, N.~Hunter-Jones, J.~Liu, and B.~Yoshida, ``{Chaos, Complexity, and
  Random Matrices},'' \href{http://dx.doi.org/10.1007/JHEP11(2017)048}{{\em
  JHEP} {\bfseries 11} (2017) 048},
\href{http://arxiv.org/abs/1706.05400}{{\ttfamily arXiv:1706.05400 [hep-th]}}.

\bibitem{gross_most_2009}
D.~Gross, S.~T. Flammia, and J.~Eisert, ``Most quantum states are too entangled
  to be useful as computational resources,''
  \href{http://dx.doi.org/10.1103/PhysRevLett.102.190501}{{\em Phys. Rev.
  Lett.} {\bfseries 102} (2009) 190501},
  \href{http://arxiv.org/abs/0810.4331}{{\ttfamily arXiv:0810.4331
  [quant-ph]}}.

\bibitem{bouland2019computational}
A.~Bouland, B.~Fefferman, and U.~Vazirani, ``{Computational pseudorandomness,
  the wormhole growth paradox, and constraints on the AdS/CFT duality},''
\href{http://arxiv.org/abs/1910.14646}{{\ttfamily arXiv:1910.14646
  [quant-ph]}}.

\bibitem{ji2017pseudorandom}
Z.~Ji, Y.-K. Liu, and F.~Song,
  \href{http://dx.doi.org/10.1007/978-3-319-96878-0_5}{``Pseudorandom quantum
  states,''} in {\em Advances in Cryptology -- CRYPTO 2018}, p.~126.
\newblock Springer, 2018.
\newblock \href{http://arxiv.org/abs/1711.00385}{{\ttfamily arXiv:1711.00385
  [quant-ph]}}.

\bibitem{RQCseeds19}
R.~Mezher, J.~Ghalbouni, J.~Dgheim, and D.~Markham, ``Unitary $t$-designs from
  {\it relaxed} seeds,'' \href{http://arxiv.org/abs/1911.03704}{{\ttfamily
  arXiv:1911.03704 [quant-ph]}}.

\bibitem{RQCepnets19}
M.~Horodecki, A.~Sawicki, and M.~Oszmaniec, ``Epsilon nets, $t$-designs and
  random quantum circuits.'' To appear.

\bibitem{NVH17}
A.~Nahum, S.~Vijay, and J.~Haah, ``{Operator Spreading in Random Unitary
  Circuits},'' \href{http://dx.doi.org/10.1103/PhysRevX.8.021014}{{\em Phys.
  Rev.} {\bfseries X8} (2018) 021014},
\href{http://arxiv.org/abs/1705.08975}{{\ttfamily arXiv:1705.08975
  [cond-mat.str-el]}}.

\bibitem{RQCstatmech}
T.~Zhou and A.~Nahum, ``{Emergent statistical mechanics of entanglement in
  random unitary circuits},''
  \href{http://dx.doi.org/10.1103/PhysRevB.99.174205}{{\em Phys. Rev.}
  {\bfseries B99} (2019) 174205},
\href{http://arxiv.org/abs/1804.09737}{{\ttfamily arXiv:1804.09737
  [cond-mat.stat-mech]}}.

\bibitem{kinvchaos}
J.~Cotler and N.~Hunter-Jones, ``{Spectral decoupling in many-body quantum
  chaos},''
\href{http://arxiv.org/abs/1911.02026}{{\ttfamily arXiv:1911.02026 [hep-th]}}.

\bibitem{SusskindEnt14}
L.~Susskind, ``{Entanglement is not enough},''
  \href{http://dx.doi.org/10.1002/prop.201500095}{{\em Fortsch. Phys.}
  {\bfseries 64} (2016) 49},
\href{http://arxiv.org/abs/1411.0690}{{\ttfamily arXiv:1411.0690 [hep-th]}}.

\bibitem{CA15}
A.~R. Brown, D.~A. Roberts, L.~Susskind, B.~Swingle, and Y.~Zhao,
  ``{Holographic Complexity Equals Bulk Action?},''
  \href{http://dx.doi.org/10.1103/PhysRevLett.116.191301}{{\em Phys. Rev.
  Lett.} {\bfseries 116} (2016) 191301},
\href{http://arxiv.org/abs/1509.07876}{{\ttfamily arXiv:1509.07876 [hep-th]}}.

\bibitem{holocomp1}
S.~Chapman, H.~Marrochio, and R.~C. Myers, ``{Complexity of Formation in
  Holography},'' \href{http://dx.doi.org/10.1007/JHEP01(2017)062}{{\em JHEP}
  {\bfseries 01} (2017) 062},
\href{http://arxiv.org/abs/1610.08063}{{\ttfamily arXiv:1610.08063 [hep-th]}}.

\bibitem{holocomp2}
D.~Carmi, R.~C. Myers, and P.~Rath, ``{Comments on Holographic Complexity},''
  \href{http://dx.doi.org/10.1007/JHEP03(2017)118}{{\em JHEP} {\bfseries 03}
  (2017) 118},
\href{http://arxiv.org/abs/1612.00433}{{\ttfamily arXiv:1612.00433 [hep-th]}}.

\bibitem{holocomp3}
M.~Alishahiha, ``{Holographic Complexity},''
  \href{http://dx.doi.org/10.1103/PhysRevD.92.126009}{{\em Phys. Rev.}
  {\bfseries D92} (2015) 126009},
\href{http://arxiv.org/abs/1509.06614}{{\ttfamily arXiv:1509.06614 [hep-th]}}.

\bibitem{holocomp4}
D.~Carmi, S.~Chapman, H.~Marrochio, R.~C. Myers, and S.~Sugishita, ``{On the
  Time Dependence of Holographic Complexity},''
  \href{http://dx.doi.org/10.1007/JHEP11(2017)188}{{\em JHEP} {\bfseries 11}
  (2017) 188},
\href{http://arxiv.org/abs/1709.10184}{{\ttfamily arXiv:1709.10184 [hep-th]}}.

\bibitem{holocomp5}
P.~Caputa, N.~Kundu, M.~Miyaji, T.~Takayanagi, and K.~Watanabe, ``{Liouville
  Action as Path-Integral Complexity: From Continuous Tensor Networks to
  AdS/CFT},'' \href{http://dx.doi.org/10.1007/JHEP11(2017)097}{{\em JHEP}
  {\bfseries 11} (2017) 097},
\href{http://arxiv.org/abs/1706.07056}{{\ttfamily arXiv:1706.07056 [hep-th]}}.

\bibitem{holocomp6}
C.~A. Ag{\'o}n, M.~Headrick, and B.~Swingle, ``{Subsystem Complexity and
  Holography},'' \href{http://dx.doi.org/10.1007/JHEP02(2019)145}{{\em JHEP}
  {\bfseries 02} (2019) 145},
\href{http://arxiv.org/abs/1804.01561}{{\ttfamily arXiv:1804.01561 [hep-th]}}.

\bibitem{holocomp7}
K.~Goto, H.~Marrochio, R.~C. Myers, L.~Queimada, and B.~Yoshida, ``{Holographic
  Complexity Equals Which Action?},''
  \href{http://dx.doi.org/10.1007/JHEP02(2019)160}{{\em JHEP} {\bfseries 02}
  (2019) 160},
\href{http://arxiv.org/abs/1901.00014}{{\ttfamily arXiv:1901.00014 [hep-th]}}.

\bibitem{Liu2018}
Z.-W. Liu, S.~Lloyd, E.~Y. Zhu, and H.~Zhu, ``{Entanglement, quantum
  randomness, and complexity beyond scrambling},''
  \href{http://dx.doi.org/10.1007/JHEP07(2018)041}{{\em JHEP} {\bfseries 07}
  (2018) 041},
\href{http://arxiv.org/abs/1703.08104}{{\ttfamily arXiv:1703.08104
  [quant-ph]}}.

\bibitem{alon_probabilistic_2016}
N.~Alon and J.~H. Spencer, {\em The probabilistic method}.
\newblock Wiley Series in Discrete Mathematics and Optimization. John Wiley \&
  Sons, fourth~ed., 2016.

\bibitem{Kueng2019}
R.~Kueng, ``Quantum and classical information processes with tensors (lecture
  notes),'' Spring, 2019.
\newblock Caltech course notes: \url{https://iqim.caltech.edu/classes}.

\bibitem{Boyd04}
S.~Boyd and L.~Vandenberghe, {\em {Convex Optimization}}.
\newblock Cambridge University Press, 2004.

\bibitem{Kitaev1997}
A.~Y. Kitaev, ``Quantum computations: algorithms and error correction,''
  \href{http://dx.doi.org/10.1070/RM1997v052n06ABEH002155}{{\em Russ. Math.
  Surv.} {\bfseries 52} (1997) 1191}.

\bibitem{Watrous2009}
J.~Watrous, ``Semidefinite programs for completely bounded norms,''
  \href{http://dx.doi.org/10.4086/toc.2009.v005a011}{{\em Theory Comput.}
  {\bfseries 5} (2009) 217}, \href{http://arxiv.org/abs/0901.4709}{{\ttfamily
  arXiv:0901.4709 [quant-ph]}}.

\bibitem{Aroya2010}
A.~Ben-Aroya and A.~Ta-Shma, ``On the complexity of approximating the diamond
  norm,'' {\em Quantum Info. Comput.} {\bfseries 10} (2010) 77,
  \href{http://arxiv.org/abs/0902.3397}{{\ttfamily arXiv:0902.3397
  [quant-ph]}}.

\bibitem{Watrous2013}
J.~Watrous, ``Simpler semidefinite programs for completely bounded norms,''
  \href{http://dx.doi.org/10.4086/cjtcs.2013.008}{{\em Chic. J. Theoret.
  Comput. Sci.} {\bfseries 8} (2013) 1},
  \href{http://arxiv.org/abs/1207.5726}{{\ttfamily arXiv:1207.5726
  [quant-ph]}}.

\bibitem{kliesch_improving_2016}
M.~{Kliesch}, R.~{Kueng}, J.~{Eisert}, and D.~{Gross}, ``Improving compressed
  sensing with the diamond norm,''
  \href{http://dx.doi.org/10.1109/TIT.2016.2606500}{{\em IEEE Trans. Inf.
  Theory} {\bfseries 62} (2016) 7445},
  \href{http://arxiv.org/abs/1511.01513}{{\ttfamily arXiv:1511.01513 [cs.IT]}}.

\bibitem{michel_comments_2018}
U.~{Michel}, M.~{Kliesch}, R.~{Kueng}, and D.~{Gross}, ``{Comments on
  ``Improving Compressed Sensing With the Diamond Norm"--Saturation of the Norm
  Inequalities Between Diamond and Nuclear Norm},''
  \href{http://dx.doi.org/10.1109/TIT.2018.2861887}{{\em IEEE Trans. Inf.
  Theory} {\bfseries 64} (2018) 7443},
  \href{http://arxiv.org/abs/1612.07931}{{\ttfamily arXiv:1612.07931 [cs.IT]}}.

\bibitem{Paulsen2002}
V.~Paulsen, {\em Completely bounded maps and operator algebras}, vol.~78 of
  {\em Cambridge Studies in Advanced Mathematics}.
\newblock Cambridge University Press, Cambridge, 2002.

\bibitem{Bhatia1997}
R.~Bhatia, \href{http://dx.doi.org/10.1007/978-1-4612-0653-8}{{\em Matrix
  analysis}}, vol.~169 of {\em Graduate Texts in Mathematics}.
\newblock Springer-Verlag, New York, 1997.

\bibitem{Collins02}
B.~Collins, ``{Moments and cumulants of polynomial random variables on unitary
  groups, the Itzykson-Zuber integral, and free probability},''
  \href{http://dx.doi.org/10.1155/S107379280320917X}{{\em Int. Math. Res. Not.}
  {\bfseries 2003} (2003) 953},
  \href{http://arxiv.org/abs/math-ph/0205010}{{\ttfamily
  arXiv:math-ph/0205010}}.

\bibitem{Renes2004}
J.~M. Renes, R.~Blume-Kohout, A.~J. Scott, and C.~M. Caves, ``Symmetric
  informationally complete quantum measurements,''
  \href{http://dx.doi.org/10.1063/1.1737053}{{\em J. Math. Phys.} {\bfseries
  45} (2004) 2171}, \href{http://arxiv.org/abs/quant-ph/0310075}{{\ttfamily
  arXiv:quant-ph/0310075}}.

\bibitem{CollinsMat17}
B.~{Collins} and S.~{Matsumoto}, ``{Weingarten calculus via orthogonality
  relations: new applications},'' {\em Lat. Am. J. Probab. Math. Stat.}
  {\bfseries 14} (2017) 631, \href{http://arxiv.org/abs/1701.04493}{{\ttfamily
  arXiv:1701.04493 [math.CO]}}.

\bibitem{montanaro2013weak}
A.~Montanaro, ``Weak multiplicativity for random quantum channels,''
  \href{http://dx.doi.org/10.1007/s00220-013-1680-7}{{\em Commun. Math. Phys.}
  {\bfseries 319} (2013) 535}, \href{http://arxiv.org/abs/1112.5271}{{\ttfamily
  arXiv:1112.5271 [quant-ph]}}.

\end{thebibliography}\endgroup


\providecommand{\href}[2]{#2}\begingroup\raggedright\endgroup

\end{document}